\PassOptionsToPackage{pdftex,hyperfootnotes=false,pdfpagelabels}{hyperref}
\RequirePackage[2024-06-01]{latexrelease}

\documentclass[%
aps,pra,
reprint,
superscriptaddress,
twocolumn, 
nofootinbib,
floats,floatfix
]{revtex4-2}

\usepackage[english]{babel}
\usepackage{amsmath}
\usepackage{pifont}
\usepackage{amssymb}
\usepackage{amsthm}
\usepackage{graphicx}
\usepackage{xcolor}
\usepackage[english]{babel}
\usepackage{csquotes}
\usepackage{braket}
\usepackage{natbib}
\usepackage{mathtools}
\usepackage{enumitem}
\usepackage{IEEEtrantools}
\usepackage{relsize}
\usepackage{placeins}
\usepackage{comment}
\usepackage[normalem]{ulem}
\usepackage{systeme}
\usepackage{physics}
\usepackage{algpseudocode}
\usepackage{hhline}

\usepackage{amsmath,amssymb}
\usepackage[table]{xcolor}
\usepackage{makecell}

\usepackage{amsmath,amssymb}
\usepackage[table]{xcolor}
\usepackage{array}
\usepackage{tabularx}
\usepackage{makecell}
\usepackage{leftindex}

\definecolor{WildSand}{HTML}{F7F7F7}
\definecolor{HeaderGray}{HTML}{E5E5E5}
\definecolor{Feijoa}{HTML}{BBE9E2}
\definecolor{TablePink}{HTML}{BBD5EA}

\newcommand{\fullrowcolor}[1]{%
  \rowcolor{#1}%
  [\dimexpr \tabcolsep+\arrayrulewidth\relax]%
  [\dimexpr \tabcolsep+\arrayrulewidth\relax]%
}

\newcolumntype{L}[1]{>{\raggedright\arraybackslash}p{#1}}
\newcolumntype{Y}{>{\raggedright\arraybackslash}X}

\usepackage{titlesec}

\usepackage{array,makecell}
\usepackage{multirow}

\renewcommand{\vec}[1]{\boldsymbol{#1}}  

\makeatletter
\newcommand{\varoverset}[2]{\binrel@{#2}%
  \binrel@@{\mathop{\vphantom{M}{#2}}\limits^{#1}}}
\makeatother
  
\usepackage[noline]{algorithm2e}
\RestyleAlgo{ruled}
\SetKwInput{KwData}{Input}
\SetKwInput{KwResult}{Output}

\makeatletter
\newcommand\footnoteref[1]{\protected@xdef\@thefnmark{\ref{#1}}\@footnotemark}
\makeatother

\usepackage{xparse}
\usepackage{tikz}
\usetikzlibrary{quantikz}

\NewDocumentEnvironment{alignb}{b}{%
  \begin{align*}
  \refstepcounter{equation} #1 \tag{\theequation}
  \end{align*}
}{\ignorespacesafterend}

\newenvironment{proofsketch}{%
  \renewcommand{\proofname}{Proof sketch}\proof}{\endproof}

\newenvironment{proofof}[1][\unskip]{%
\renewcommand{\proofname}{Proof of #1}\proof}{\endproof}

\newcommand{\bigo}[1]{\mathcal{O}\left({#1}\right)}

\newsavebox{\mstrut}
\newcommand{\bbra}[1]{%
    \sbox{\mstrut}{\(#1\)}%
    \mathinner{\left\langle\kern-0.25\ht\mstrut\left\langle{#1}\right|\right.\!}%
}
\newcommand{\kett}[1]{%
    \sbox{\mstrut}{\(#1\)}%
    \mathinner{\!\left.\left|{#1}\right\rangle\kern-0.25\ht\mstrut\right\rangle}%
}
\newcommand{\bbrakett}[2]{%
    \sbox{\mstrut}{\(#1\)}%
    \sbox{\mstrut}{\(#2\)}%
    \mathinner{\left\langle\kern-0.25\ht\mstrut\left\langle{#1}\right.\right.}\!\! | \!\!\mathinner{\left.\left.{#2}\right\rangle\kern-0.25\ht\mstrut\right\rangle}%
}

\newcommand{\kettbbra}[2]{
    \kett{#1} \! \bbra{#1}
}

\newcommand{\ketbraa}[2]{%
    | #1 \rangle \! \langle #1|
}

\newcommand{\exptsuper}[3]{%
    \sbox{\mstrut}{\(#1\)}%
    \sbox{\mstrut}{\(#2\)}%
    \sbox{\mstrut}{\(#3\)}%
    \mathinner{\left\langle\kern-0.25\ht\mstrut\left\langle{#1}\right.\right|}\! {#2} \!\mathinner{\left|\left.{#3}\right\rangle\kern-0.25\ht\mstrut\right\rangle}%
}

\usepackage{tikz-cd} 
\usepackage{adjustbox}

\usepackage{tikz}						
\usetikzlibrary{math, arrows,shapes.misc,
		       automata,backgrounds,
		       petri,topaths, decorations.pathmorphing, tikzmark}	

\usepackage{pgffor}

\usepackage{pgfplots}
\usetikzlibrary{pgfplots.groupplots}
\pgfplotsset{compat=1.11}
\usepgfplotslibrary{fillbetween}

\usepackage{tikz}
\usetikzlibrary{
  shapes,
  shapes.geometric,
	trees,
	matrix,
  positioning,
    pgfplots.groupplots,
  }

\usepackage[toc,page]{appendix}
\usepackage[tight]{minitoc}
\usepackage{tocloft}

\doparttoc 
\faketableofcontents 

\renewcommand \partname{}
\usepackage{hyperref}
\hypersetup{
	colorlinks = true,
    linkcolor = [rgb]{0.25, 0.34, 0.63},
	citecolor = [rgb]{0.13,0.55,0.13},
	urlcolor = [rgb]{0.25, 0.34, 0.63}}

\definecolor{myred}{RGB}{241,60,32}
\definecolor{mygreen}{RGB}{49,177,159}

\def\01{\{0,1\}}

\makeatletter
\let\oldabs\abs
\def\abs{\@ifstar{\oldabs}{\oldabs*}}
\let\oldnorm\norm
\def\norm{\@ifstar{\oldnorm}{\oldnorm*}}
\makeatother

\newtheorem{theorem}{Theorem}
\newtheorem{lemma}{Lemma}
\newtheorem{corollary}{Corollary}
\newtheorem{proposition}{Proposition}

\makeatletter 
    
\renewcommand\onecolumngrid{
\do@columngrid{one}{\@ne}%
\def\set@footnotewidth{\onecolumngrid}
\def\footnoterule{\kern-6pt\hrule width 1.5in\kern6pt}%
}

\renewcommand\twocolumngrid{
        \def\footnoterule{
        \dimen@\skip\footins\divide\dimen@\thr@@
        \kern-\dimen@\hrule width.5in\kern\dimen@}
        \do@columngrid{mlt}{\tw@}
}%

\makeatother

\begin{document}

\title{The Complexity of Dynamical Correlators: \\ Operator Shadows and Exponential Learning Separations}
\date{\today}

\author{Shao-Hen Chiew}
\affiliation{Institute of Physics, Ecole Polytechnique Fédérale de Lausanne (EPFL),  Lausanne CH-1015, Switzerland}
\affiliation{Centre for Quantum Science and Engineering, Ecole Polytechnique F\'{e}d\'{e}rale de Lausanne (EPFL), CH-1015 Lausanne, Switzerland}

\author{Armando Angrisani}
\affiliation{Institute of Physics, Ecole Polytechnique Fédérale de Lausanne (EPFL),  Lausanne CH-1015, Switzerland}
\affiliation{Centre for Quantum Science and Engineering, Ecole Polytechnique F\'{e}d\'{e}rale de Lausanne (EPFL), CH-1015 Lausanne, Switzerland}

\author{Zo\"e Holmes}
\affiliation{Institute of Physics, Ecole Polytechnique Fédérale de Lausanne (EPFL),  Lausanne CH-1015, Switzerland}
\affiliation{Centre for Quantum Science and Engineering, Ecole Polytechnique F\'{e}d\'{e}rale de Lausanne (EPFL), CH-1015 Lausanne, Switzerland}

\begin{abstract}
\normalsize
{Quantum platforms can realize many-body dynamics beyond classical simulation yet complete readout remains intractable: the cost of extracting accessible information scales exponentially with system size. Classical shadows and Bell sampling offer scalable, multi-observable estimation from randomized or entanglement-assisted measurements. Here we aim to push these ideas beyond static snapshots to dynamical correlators, including out-of-time-ordered correlators (OTOCs) and two-point functions. In particular, we introduce the notion of \textit{the shadow of an operator}, defined as the classical shadow of the vectorized time-evolved operator. Pauli operator-shadows enable simultaneous estimation of all local OTOCs, while Clifford operator-shadows enable efficient simultaneous estimation of all two-point correlators. Alternatively, Bell sampling allows one to simultaneously compute all diagonal OTOCs. We also prove information-theoretic lower bounds for learning OTOCs, fully characterizing their query complexities in many cases, and yielding exponential separations that formalize when the vectorized approach provides measurement-efficiency advantages.}

\end{abstract}

\maketitle

\makeatletter

\section{Introduction} 

Quantum computers and simulators evade the exponential computational overhead of representing many-body quantum states, yet the information encoded in quantum systems still grows exponentially with system size, making complete classical readout impossible. Limited experimental time and data-storage capacity therefore force deliberate choices about what to measure and how to measure it, so as to optimally extract physically relevant information about quantum processes.

There is growing recognition that entanglement- \cite{chen2021exponential,huang2021information,chen2022quantum,caro2024learning} and randomization-based \cite{huang2020predicting} measurement schemes can dramatically improve measurement efficiency. Classical shadows provide a framework for predicting many expectation values of quantum states: randomized unitaries followed by simple measurements produce compact classical “snapshots” that can be post-processed to estimate a wide family of observables without designing a separate protocol for each~\cite{huang2020predicting, elben2022randomized}. Complementarily, Bell sampling (an established entanglement-assisted strategy based on joint Bell-basis measurements on two copies) gives direct access to overlap/swap-type quantities, while enabling simultaneous inference of multiple state properties~\cite{montanaro2017learning,gross2021schur,hangleiter2024bell}. To date, however, these approaches have been largely confined to characterizing states rather than dynamics.

In many settings, the goal is to probe dynamical processes directly. Out-of-time-ordered correlators (OTOCs) and multi-point correlators are key diagnostics of many-body dynamics, capturing information scrambling, thermalization, and the spread of quantum correlations~\cite{swingle2016measuring,schuster2023operator,google2025observation, von2018operator,nahum2018operator,rakovszky2018diffusive,khemani2018operator,roberts2018operator,qi2019quantum,fisher2023random,xu2024scrambling,zhang2025quantum}. Here we develop entanglement- and randomization-enhanced protocols that enable the efficient simultaneous estimation of large families of such quantities.

Concretely, our first set of results introduce the concept of the shadow of an operator, defined as the classical shadow of the vectorized time-evolved operator. This framework enables efficient inference of dynamical observables from compact measurement data: Pauli operator-shadow supports the simultaneous estimation of all local OTOCs, while Clifford operator-shadow enables the simultaneous estimation of a large family of two-point correlators. Alternatively, approaches based on Bell sampling allow one to simultaneously access all diagonal OTOCs. 

Our second set of results establish optimal sample-complexity guarantees for operator shadows. We then prove exponential advantages for this vectorized approach over existing single-copy protocols for OTOC-learning tasks. Concretely we show that (i) estimating all $4^n\!-\!1$ diagonal OTOCs admits an exponential separation:
without ancillas any single-copy protocol (even adaptive, and posessing quantum memory) requires
$\Omega(2^n/\varepsilon^2)$ runs, whereas with $n$ ancillas (i.e.\ the $2n$-qubit vectorized setting) Bell/operator
sampling requires only $O(n/\varepsilon^2)$ runs; and (ii) estimating \emph{all} OTOCs exhibits a further exponential
separation with respect to quantum memory, requiring $\Omega(4^n/\varepsilon^2)$ runs without memory but only
$O(n/\varepsilon^4)$ copies with memory. We thus demonstrate that operator-Bell sampling and shadows deliver measurement-efficiency gains for characterizing many-body quantum dynamics.


\begin{table*}[ht!]
\centering

\renewcommand{\tabularxcolumn}[1]{m{#1}}

\newcolumntype{Z}[1]{%
  >{\hsize=#1\hsize\linewidth=\hsize\centering\arraybackslash}X%
}

\newlength{\ResultsTableWidth}
\setlength{\ResultsTableWidth}{0.90\textwidth}

\renewcommand{\arraystretch}{1.45}
\setlength{\tabcolsep}{8pt}
\arrayrulecolor{black!65}

\centering
\begin{tabularx}{\ResultsTableWidth}{
  @{}
  !{\vrule width 1pt}
  Z{0.80} 
  !{\vrule width 1pt}
  Z{1.15} 
  !{\vrule width 1pt}
  Z{1.05} 
  !{\vrule width 1pt}
  @{}
}
\Xhline{1pt}

\fullrowcolor{HeaderGray}
\textbf{Problem}
&
\textbf{Upper bound / algorithm}
&
\textbf{Lower bound}
\\
\hline

\fullrowcolor{Feijoa}
\makecell[c]{\\[0pt]
\textbf{General OTOCs}}
&
\makecell[c]{
    \\[-10pt]
    \(\displaystyle
    \mathcal{O}\!\left(
        \frac{9^w\log(M/\delta)}{\epsilon^2}
    \right)\)
    \\[2pt]
    {\footnotesize
    Theorems~\ref{theorem:simul_otoc_shadows_nq},~\ref{theorem:general_op_shad}}
}
&
\makecell[c]{
    \\[-10pt]
    \(\displaystyle
    \Omega\!\left(
        \frac{9^w\log M}{\epsilon^2}
    \right)\)
    \\[2pt]
    {\footnotesize
    Theorems~\ref{theorem:shadow_lowerbound},~\ref{theorem:shadow_gen_lowerbound_formal}, \ref{theorem:shadow_gen_lowerbound_formal_choi}}
}
\\
\hline

\fullrowcolor{Feijoa}
\makecell[c]{\\[0pt]
\textbf{Diagonal OTOCs}}
&
\makecell[c]{
    \\[-10pt]
    \(\displaystyle
    \mathcal{O}\!\left(
        \frac{3^w\log(M/\delta)}{\epsilon^2}
    \right)\)
    \\[2pt]
    {\footnotesize
    Theorems~\ref{theorem:simul_otoc_shadows_nq},~\ref{theorem:general_op_shad_diag}}
}
&
\makecell[c]{
    \\[-10pt]
    \(\displaystyle
    \Omega\!\left(
        \frac{3^w\log M}{\epsilon^2}
    \right)\)
    \\[2pt]
    {\footnotesize
    Theorems~\ref{theorem:shadow_lowerbound},~\ref{theorem:shadow_diag_lowerbound_formal}}
}
\\
\hline

\fullrowcolor{TablePink}
\makecell[c]{\\[0pt]
\textbf{All OTOCs}}
&
\makecell[c]{
  \\[-10pt]
  \(\displaystyle
  \mathcal{O}\!\left(
    \frac{n}{\epsilon^4}
  \right)\) {\footnotesize with quantum memory}
  \\[8pt]
  {\footnotesize
  Theorem~\ref{theorem:exp_sep_all_otocs_query}, Lemma~\ref{app:qmem_ub}}
}
&
\makecell[c]{
    \\[-10pt]
    \(\displaystyle
    \Omega\!\left(
        \frac{4^n}{\epsilon^2}
    \right)\) {\footnotesize w/o quantum memory}
    \\[2pt]
    {\footnotesize
    Theorem~\ref{theorem:exp_sep_all_otocs_query}}
}
\\
\hline

\fullrowcolor{TablePink}
\makecell[c]{\\[0pt]
\textbf{All diagonal OTOCs}}
&
\makecell[c]{
    \\[-10pt]
    \(\displaystyle
    \mathcal{O}\!\left(
        \frac{n}{\epsilon^2}
    \right)\) {\footnotesize with \(n\) ancillas}
    \\[7pt]
    {\footnotesize
    Theorems~\ref{theorem:expsep_diag},~\ref{theorem:diagonal_otoc_algo_nk}}
}
&
\makecell[c]{
    \\[-10pt]
    \(\displaystyle
    \Omega\!\left(
        \frac{2^n}{\epsilon^2} 
    \right)\) {\footnotesize w/o ancillas}
    \\[2pt]
    {\footnotesize
    Theorem~\ref{theorem:expsep_diag}}
}
\\
\hhline{>{\arrayrulecolor{TablePink}}->{\arrayrulecolor{black!65}}--}

\fullrowcolor{TablePink}
&
\makecell[c]{
    \\[-10pt]
    \(\displaystyle
    \mathcal{O}\!\left(
        \frac{
            n\,2^{\,n-n_{\mathrm{anc}}}\log(1/\delta)
        }{
            \epsilon^2
        }
    \right)\)
    \\[2pt]
    {\footnotesize
    with \(n_{\mathrm{anc}}\) ancillas; non-adaptive}
    \\[2pt]
    {\footnotesize
    Theorems~\ref{theorem:expsep_diag}, \ref{theorem:diagonal_otoc_algo_nk}}
}
&
\makecell[c]{
    \\[-10pt]
    \(\displaystyle
    \Omega\!\left(
        \frac{
            n\,2^{\,n-n_{\mathrm{anc}}}
        }{
            \epsilon^2
        }
    \right)\)
    \\[2pt]
    {\footnotesize
    with \(n_{\mathrm{anc}}\) ancillas; non-adaptive}
    \\[2pt]
    {\footnotesize
    Theorems~\ref{theorem:expsep_diag},~\ref{theorem:expsep_diag_nonad}}
}
\\
\hline

\fullrowcolor{WildSand}
\makecell[c]{\\[0pt]
\textbf{Two-point correlators}}
&
\makecell[c]{
    \\[-10pt]
    \(\displaystyle
    \mathcal{O}\!\left(
        \frac{\log(M/\delta)}{\epsilon^4}
    \right)\)
    \\[2pt]
    {\footnotesize
    Theorem~\ref{theorem:simul_2pc_shadows}}
}
&
\makecell[c]{\\[0pt]
\(\text{?}\)}
\\
\hline

\fullrowcolor{Feijoa}
\textbf{Channel PTM elements}
&
\makecell[c]{
    \\[-20pt]
    \(\displaystyle
    \mathcal{O}\!\left(
        \frac{9^w\log(M/\delta)}{\epsilon^2}
    \right)\)
    \\[2pt]
    {\footnotesize
    Corollary~\ref{theorem:ptm_tight}, Theorem~\ref{theorem:general_ptm_shad}}
}
&
\makecell[c]{
    \\[-20pt]
    \(\displaystyle
    \Omega\!\left(
        \frac{9^w\log M}{\epsilon^2}
    \right)\)
    \\[2pt]
    {\footnotesize
    Corollary~\ref{theorem:ptm_tight}, Theorem~\ref{theorem:ptm_gen_lowerbound_formal}}
}
\\
\hline

\fullrowcolor{Feijoa}
\textbf{Channel diagonal PTM elements}
&
\makecell[c]{
    \\[-20pt]
    \(\displaystyle
    \mathcal{O}\!\left(
        \frac{3^w\log(M/\delta)}{\epsilon^2}
    \right)\)
    \\[2pt]
    {\footnotesize
    Corollary~\ref{theorem:ptm_tight}, Theorem~\ref{theorem:general_ptm_shad_diag}}
}
&
\makecell[c]{
    \\[-20pt]
    \(\displaystyle
    \Omega\!\left(
        \frac{3^w\log M}{\epsilon^2}
    \right)\)
    \\[2pt]
    {\footnotesize
    Corollary~\ref{theorem:ptm_tight}, Theorem~\ref{theorem:ptm_diag_lowerbound_formal}}
}
\\
\Xhline{1pt}

\end{tabularx}

\caption{
Summary of the algorithmic upper bounds and information-theoretic
lower bounds established in this work.
Here, \(w\) denotes the operator weight and \(M\) the number of
quantities estimated simultaneously.
Green rows indicate matching upper and lower bounds, while blue rows
highlight exponential separations under the stated resource
restrictions.
All upper bounds are constructive.
}
\label{tab:algorithms_index}

\end{table*}

\section{Framework}

Two-point correlators and OTOCs are popular operator dynamics metrics that provide natural probes of information scrambling and thermalization~\cite{swingle2016measuring,schuster2023operator,google2025observation,von2018operator,nahum2018operator,rakovszky2018diffusive,khemani2018operator,roberts2018operator,qi2019quantum,fisher2023random,xu2024scrambling,zhang2025quantum}.
Two-point correlators take the form $\tr(O_1 O_2)/2^n$, where $O_1,O_2$ are Heisenberg operators, and $\tr(\cdot) \equiv \sum_{i}\bra{i}(\cdot)\ket{i}$ denotes the trace operator. When one of the operators correspond to a Pauli operator (e.g. $\tr(P_k O)/2^n$), it corresponds to a Pauli amplitude. OTOCs take the general form $\tr(O(t) P O(t) Q)/2^n$, where $P, Q \in {\mathcal{P}}_n$. They are said to be diagonal if $P=Q$, and off-diagonal otherwise. In this manner, OTOCs and two-point correlators probe the properties of a time-evolved operator $O(t)$.

In the Heisenberg picture of dynamics, operators evolve under unital maps $\mathcal{E}^\dagger$ as $O \rightarrow O(t) = \mathcal{E}^\dagger(O)$, while states remain fixed in time. Here, $\mathcal{E}^\dagger$ denotes the adjoint (dual) of the quantum channel $\mathcal{E}$ with respect to the Hilbert–Schmidt (HS) inner product $\tr(A^\dagger B)$. Any $n$-qubit Heisenberg operator $O(t)$ can be expanded in an orthogonal operator basis such as the Pauli basis ${\mathcal{P}}_n = \{P_k\}_{k=1}^{4^n}$, i.e. $O(t) = \sum_k c_k P_k$, where $c_k = \tr(P_k O(t))/2^n \in \mathbb{R}$ are called the \textit{Pauli amplitudes of $O(t)$}. The normalized squared values of $c_k$ define a probability distribution $p_k \equiv c^2_k/\sum_k c_k^2$, called the \textit{Pauli distribution of $O(t)$}.

Ref.~\cite{chiew2026quantum} observed that, by mapping the Heisenberg-evolved operator ($O(t)$) to a quantum state ($\kett{O(t)}$) through vectorization, tasks naturally formulated in the Heisenberg picture can be recast as state-based tasks in the Schrödinger picture. These include sampling from the Pauli distribution (of $O(t)$), computing infinite-temperature OTOCs and two-point correlators, and estimating operator stabilizer entropies (OSE) and local operator entanglement (LOE).

The (normalized) vectorization map with respect to an $n$-qubit orthogonal operator basis $\mathcal{Q} = \{Q_k\}_{k=1}^{2^n}$ and a $2n$-qubit orthonormal state basis $\{\ket{k}\}_{k=1}^{4^n}$ is defined as:
\begin{equation} \label{eq:vect_map}
    \kett{O}_\mathcal{Q} \equiv \sum_{k=1}^{4^n} \frac{\tr(Q_k O)}{\sqrt{\sum_i \abs{\tr(Q_i O)}^2}} \ket{k},
\end{equation}
and maps between $n$-qubit linear operators $O \in \mathcal{L}(\mathcal{H})$ and $2n$-qubit pure states $\kett{O} \in \mathcal{H}'$, where $\mathcal{H}$ and $\mathcal{H}'$ denote the $n$- and $2n$-qubit Hilbert spaces respectively. It is one-to-one, and distance preserving. Two important cases are when $\mathcal{Q} = \mathcal{C} = \{\ketbra{i}{j}\}$ (the computational operator basis), and when $\mathcal{Q} = \mathcal{P}_n$, the Pauli basis; $\forall O$, $\kett{O}_\mathcal{C}$ and $\kett{O}_\mathcal{P}$ are related by a Bell basis transformation. We also drop the subscript in the former case, and simply write $\kett{\cdot}_{\mathcal{C}} = \kett{\cdot}$.

The states $\kett{O(t)}$ can be prepared on quantum computers by exploiting the ricochet identity $A \otimes C^T \kett{B} = \kett{A B C}$, which implies $(U^\dagger O U \otimes \mathbb{I}) \kett{\mathbb{I}} = \kett{O(t)}$, where $\kett{\mathbb{I}}$ is the $2n$-qubit maximally entangled state, consisting of $n$ Bell pairs. Written in this form, the unitary applied to $\kett{\mathbb{I}}$ (written in channel form as $\tilde{U}^\dagger \circ \tilde{O} \circ \tilde{U} (\cdot)$ where $\tilde{X}(\cdot) \equiv X(\cdot)X^\dagger$ denotes conjugation) implements a Loschmidt echo \cite{peres1984stability,jalabert2001environment}, with $O$ playing the role of a local perturbation.

The Pauli distribution of $p_k$ of $O(t)$, obtained by sampling from $\kett{O(t)}$ in the Bell basis encodes a significant amount of operatorial properties, including diagonal OTOCs, and more generally all geometrical properties of $O(t)$, such as its volume, boundary position, etc. This renders it convenient to generate and store an empirical estimate of the Pauli distribution, which subsequently allows many operatorial properties to be extracted in an \textit{a posteriori} manner.
Moreover, two-point correlators $\tr(P_k O)/2^n$ are encoded as the wavefunction amplitudes of $\kett{O(t)}$ (evident from Eq.~(\ref{eq:vect_map}) with $\mathcal{Q} = \mathcal{P}_n$), and OTOCs, or more generally the expectation value of self-adjoint superoperators over $O(t)$, are encoded as Pauli expectation values of $\kett{O(t)}$. This enables the exploitation of the basic but powerful observation that the expectation values of commuting observables can be simultaneously estimated. Examples are the set of OTOCs $\{\tr(O(t)Z_i O(t)Z_j)/2^n\}$ via computational basis measurements, or the set of all diagonal OTOCs $\{\tr(O(t) P O(t) P)/2^n\}$ via Bell basis measurements.

\section{Algorithms}

Our first set of results extends the `measure first, think later' approach, to families of operatorial properties. The core idea is to apply the framework of classical shadows \cite{huang2020predicting} to the vectorized operators.

In its original formulation, classical shadows seeks to output a compressed classical description of given quantum states $\rho$ via randomized measurements. This is achieved by randomly sampling unitaries $U$ from a suitable unitary ensemble $\mathcal{U}$, and applying them to given samples of $\rho$ before measuring in the computational basis. After classical post-processing, a compact representation of $\rho$ is constructed and stored in classical memory, which can be used to simultaneously estimate many expectation values $\{ \tr(\rho M_i) \}$ to low error $\epsilon$, with $\{M_i\}$ depending on the choice of $\mathcal{U}$.

Two classes of shadows can be distinguished: (1) \textit{Pauli shadows} are obtained if the unitaries are random single-qubit Clifford gates, i.e. $\mathcal{U} = \text{Cl}(2)^{\otimes n}$. This enables efficient estimation for observables $M_i$ that are of small support/weight, e.g. low-weight Pauli operators. (2) \textit{Clifford shadows} are obtained with random $n$-qubit unitaries, i.e. $\mathcal{U} = \text{Cl}(2^n)$. This enables efficient estimation for $M_i$'s with bounded Hilbert-Schmidt norm, and when they admit efficient classical representations, e.g. fidelities with stabilizer states. \\

\paragraph*{Pauli operator shadow.}

Firstly, we provide efficient $n$- and $2n$-qubit algorithms for the simultaneous estimation of many low-weight OTOCs:

\begin{theorem}[Pauli operator shadow] \label{theorem:simul_otoc_shadows_nq}
Consider the task of simultaneously estimating $M$ OTOCs $\{ \tr(O(t) P_{i} O(t) Q_{i})/2^n\}_{i=1}^M$ each to additive error $\epsilon$, with success probability $1-\delta$, where $P_{i}, Q_{i}$ are of weight $w$. There are non-adaptive algorithms using $n$ or $2n$ qubits that achieve this using $N = \bigo{9^w \log(M/\delta)/\epsilon^2}$ queries to the unitary $O(t) = U^\dagger O U$ (or alternatively $(U^\dagger \otimes U^T)(O \otimes \mathbb{I})$ with $2n$ qubits). \\
If, additionally, the OTOCs are promised to be diagonal, i.e. $\forall~i ~ P_i = Q_i$, then the scaling in $w$ can be improved to $\bigo{3^w}$.
\end{theorem}

The proof of Theorem~\ref{theorem:simul_otoc_shadows_nq}, given in Appendix~\ref{app:pauli_shadow}, is a direct consequence of the equivalence between OTOCs and expectation values of Pauli operators over $\kett{O(t)}$, and is based on constructing the Pauli shadow of the state $\kett{O(t)}$ -- which we directly call the \textit{Pauli shadow of $O(t)$} -- that yields a classical description of $O(t)$. 

In particular, with access to $2n$ qubits, the Pauli shadow of $O(t)$ can be obtained by applying the standard Pauli classical shadows protocol to the state $\kett{O(t)}$, which can be prepared by either applying $(U^\dagger O U) \otimes \mathbb{I}$ or $(U^\dagger \otimes U^T)(O \otimes \mathbb{I})$ to $\kett{\mathbb{I}}$. The latter choice parallelizes the application of $U^\dagger$ and $U^T$, but requires access to $U^T$. On the other hand, with access to only $n$ qubits, the same Pauli shadow can be obtained by randomly initializing $n$-qubit product stabilizer states, applying $U^\dagger O U$, and performing random Pauli measurements. This process exactly simulates the measurement of the $2n$-qubit approach using only $n$ qubits, but requires $U$ and $U^\dagger$ to be applied sequentially.

The above approaches based on randomized measurements complement the Pauli operator sampling/Bell-sampling approach of \cite{chiew2026quantum} by efficiently capturing low-weight OTOCs, regardless of (superoperator) commutativity and diagonality. It is therefore expected to be advantageous for tasks that only involve estimating low-weight OTOCs, such as when probing the average operator size, which consists of only weight-1 OTOCs \cite{nahum2018operator,roberts2018operator}, and the LOE over small spatial partitions. It can also be obtained using only $n$ qubits and local Pauli initialization/measurements, rendering it amenable to near-term implementation. 

A question that arises is whether the protocol of Theorem~\ref{theorem:simul_otoc_shadows_nq} is optimal. We will answer this in the affirmative in the next section by providing matching information-theoretic lower bounds (Theorem~\ref{theorem:shadow_lowerbound}). \\

\begin{figure*}[!htb]
    \centering
    \includegraphics[width=1\linewidth]{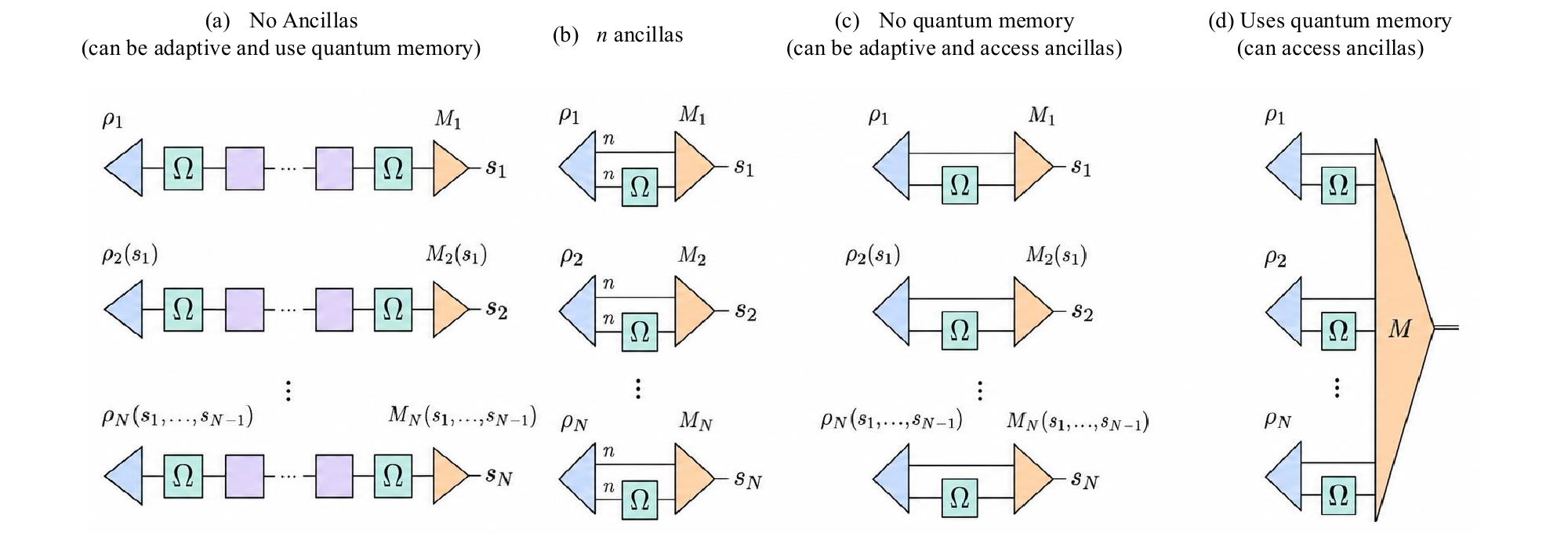}
    \caption{Classes of protocols for learning OTOCs considered in our work, which also captures the majority of existing protocols. For compactness we here write $\Omega = \mathcal{O}_{\mathcal{E}, O}$ to denote the generalized Loschmidt echo channel.
    }
    \label{fig:protocols}
\end{figure*}

\paragraph*{Clifford Operator Shadow.}
Next, we turn to the problem of estimating two-point Pauli correlators, which are of the form $\tr(P_kP_l(t))/2^n$, $P_k, P_l \in {\mathcal{P}}_n$. As useful probes of operator growth \cite{von2018operator,parker2019universal,von2022operator,nahum2022real}, scrambling, and chaos \cite{cotler2017chaos,fisher2023random,moudgalya2024symmetries,yoshimura2025theory}, they also appear in the computation of other physical quantities such as the Operator Stabilizer Entropy (OSE) \cite{dowling2025magic}, dynamical structure factor \cite{baez2020dynamical}, and more. When the index $l$ is held fixed, the set of $4^n-1$ two-point correlators $\{\tr(P_k P_l(t))/2^n\}_{P_k \in {\mathcal{P}}_n \backslash\mathbb{I}}$ correspond to the Pauli amplitudes $c_l \equiv \tr(P_k O(t))/2^n$ of the Heisenberg operator $O(t) \equiv P_l(t) = U^\dagger P_l U $. 

We provide an efficient algorithm that is able to simultaneously estimate \textit{all} $4^n-1$ Pauli amplitudes of a Heisenberg operator efficiently:
\begin{theorem}[Correlators via Clifford classical shadows]\label{theorem:simul_2pc_shadows}
Consider the task of simultaneously estimating $M$ correlators $\{ \tr(O P_{k})/2^n\}_{k=1}^M$ each to additive error $\epsilon$, with success probability $1-\delta$, where $P_i \in \mathcal{P}_n$ are Pauli operators. There is a $2n$-qubit algorithm that achieves this using $N=\bigo{\log(M/\delta)/\epsilon^4}$ queries to the unitary $(U^\dagger R_O(\pi/4) U) \otimes \mathbb{I}$ or $(U^\dagger \otimes U^T)(R_O(\pi/4) \otimes \mathbb{I})$, where $R_O(\pi/4) = e^{-i \pi O/4}$.
\end{theorem}

The above algorithm is based on constructing the Clifford classical shadow of a closely related state $\kett{\sigma(O(t))} \equiv \frac{1}{\sqrt{2}}(\kett{\mathbb{I}} + i\kett{O(t)})$, which we (by slight abuse of terminology) directly call the \textit{Clifford shadow of $O(t)$}. These states can be prepared by applying the unitary $(U^\dagger R_O(\pi/4) U) \otimes \mathbb{I}$ or $(U^\dagger \otimes U^T)(R_O(\pi/4) \otimes \mathbb{I})$ to $\kett{\mathbb{I}}$. They are related to the \textit{Pseudo-Choi state} recently introduced in the context of Hamiltonian learning \cite{castaneda2025hamiltonian, zhao2024learning}.

We emphasize that a naive application of Clifford classical shadows to the state $\kett{O(t)}$ only yields \textit{unsigned} correlators $\{\abs{\tr(P_k O(t))/2^n}\}_{P_k \in {\mathcal{P}}_n}$. The above construction circumvents this limitation through the $i\mathbb{I}$ term, which reduces an overlap estimation problem to a fidelity estimation problem, enabling Clifford shadows to also acquire sign information. As a bonus of the reduction, the latter problem requires less stringent access to $U$ and $U^\dagger$, and yields yet more algorithms for the computation of two-point correlators via methods such as the SWAP test and Direct Fidelity Estimation. Full details are delegated to Appendix~\ref{app:2pc}.

Notably, compared to existing Schrodinger-picture Pauli amplitude algorithms for estimating all $\{\tr(\rho P_k)\}_{P_k \in {\mathcal{P}}_n \backslash\mathbb{I}}$ \cite{huang2021information}, the above algorithm does not require quantum memory (i.e. the need to entangle multiple copies of $\kett{\sigma(O(t))}$). This is a consequence of the unitality of Heisenberg time-evolution, which is not necessarily present in Schrodinger time-evolution.

Theorem~\ref{theorem:simul_2pc_shadows} extends existing $n$-qubit methods based on accessing the channel $\mathcal{E}$ \cite{chiew2026quantum,caro2024learning}, which either requires the set of target quantities to possess additional structure due to commutation \cite{chiew2026quantum}, or require $\bigo{\text{exp(n)}}$ samples (but is suited for the estimation of all possible correlators $\{\tr(P_k P_l(t))/2^n\}_{k,l=1,1}^{4^n,4^n}$ \cite{caro2024learning}). 

The same constraint that precludes Bell sampling on $n$ qubits \cite{chiew2026quantum} also applies to converting our $2n$ qubit Clifford operator shadow protocol into an $n$ qubit protocol. Namely, the requisite global randomization and measurement acts on the full $2n$-qubit Hilbert space and therefore generally requires entangling operations across the two $n$-qubit halves. Establishing a formal impossibility (or tight overhead) result for realizing these global operations with only $n$-qubit access is an interesting open direction.

Finally, we do not provide a matching lower bound for this task at present. However, the \(1/\epsilon^2\) dependence is already unavoidable for estimating a single correlator in the independent single-query model, as a consequence of standard binary state-discrimination bounds \cite{fuchs1999cryptographic}. Moreover, by analogy with standard multi-observable estimation lower bounds, we expect the logarithmic dependence on $M$ to be unavoidable in the worst case. We therefore expect an $\Omega(\log(M)/\epsilon^2)$ lower bound to hold for the restricted family of shifted vectorized operators under our channel-access model.



\section{Learning lower bounds and exponential separations}
Next, we study the hardness of learning OTOCs. In doing so, we complement our algorithms with information-theoretic optimality guarantees, and expose exponential separations in the difficulty of certain OTOC-learning tasks.

For the purposes of this section, we extend our framework beyond closed-system unitary dynamics to encompass non-unitary maps $\mathcal{E}$, which yields $O(t) = \mathcal{E}^\dagger(O)$, where $\mathcal{E}^\dagger$ is the (unital) adjoint of $\mathcal{E}$. In contrast to the unitary case where the pure state $\kett{O(t)}$ is prepared by applying the $n$-qubit unitary $U^\dagger O U$, we consider the natural generalization of this linear map to:
\begin{equation}
    \mathcal{O}_{\mathcal{E}, O}(\cdot) \;\equiv\;\mathcal{E}^\dagger \circ \widetilde{O} \circ \mathcal{E}(\cdot),
\end{equation}
where $\widetilde{O} \equiv O(\cdot)O^\dagger$ denotes conjugation by the initial Heisenberg operator $O \in {\mathcal{P}}_n$. This linear map can be interpreted as the generalization of the Loschmidt echo beyond the unitary case.

Concretely, we consider protocols that estimate OTOCs by querying $\mathcal{O}_{\mathcal{E}, O}$. Summarized in Fig.~\ref{fig:protocols}, they are allowed to possibly be \textit{adaptive} (able to condition the initial state and measurement setting on results of previous runs), possess \textit{quantum memory} (able to repeatedly query $\mathcal{O}_{\mathcal{E}, O}$ and interleave processing channels in each coherent run), and/or access \textit{ancillas} (able to work on $n+n_{\text{anc}}$ qubits). To ensure that $\mathcal{O}_{\mathcal{E}, O}$ can be queried as a physical map (i.e. it is CPTP), it is sufficient to assume that the channel $\mathcal{E}$ is also unital, i.e. $\mathcal{E}$ is doubly-stochastic (though this assumption need not be present for our lower bounds to hold). Formal definitions are delegated to Appendix~\ref{app:models}.

These models capture the vast majority of existing protocols for learning OTOCs, based on using $n$ \cite{swingle2016measuring,mi2021information,schuster2022many,schuster2023operator,schuster2023learning,cotler2023information,google2025observation,algorithmiq} and $2n$ qubits \cite{yoshida2019disentangling,landsman2019verified,schuster2022many,sundar2022proposal,green2022experimental,xu2024scrambling}. For instance, our learning protocols based on Pauli operator shadows (Theorem~\ref{theorem:simul_otoc_shadows_nq}) is also an instance of protocols with access to either $0$ or $n$ ancillas, while requiring neither adaptivity nor quantum memory. \\

\paragraph*{Tight lower bounds on learning OTOCs.}

Firstly, we show that the algorithm of Theorem~\ref{theorem:simul_otoc_shadows_nq} based on Pauli operator shadows is indeed optimal. This is reflected by the following sample-complexity lower bounds:
\begin{theorem}[Sample complexity lower bound for learning low-weight OTOCs, informal]\label{theorem:shadow_lowerbound}
Consider the task of estimating $M$ OTOCs $\{\tr(\mathcal{E}(P_i) O \mathcal{E}(Q_i) O)/2^n\}_{i=1}^M$ where $P_i$ and $Q_i$ are weight-$w$ to additive precision $\epsilon$ with high success probability. Any protocol using $n$ (or more) qubits based on initializing as product states, querying the channel $\mathcal{O}_{\mathcal{E}, O}$, and performing local measurements in each of the $N$ runs requires at least $N \geq \Omega\left(9^w \log(M)/\epsilon^2\right)$ runs to solve this problem.\\
If, additionally, the OTOCs are promised to be diagonal, i.e. $\forall~ i ~ P_i = Q_i$, then at least $N \geq \Omega\left(3^w \log(M)/\epsilon^2\right)$ runs are required.
\end{theorem}

\begin{proofsketch}
A detailed proof, together with generalizations to arbitrary initial states and vectorized-state access, is given in Theorems~\ref{theorem:shadow_diag_lowerbound_formal}, \ref{theorem:shadow_gen_lowerbound_formal}, \ref{theorem:shadow_gen_lowerbound_formal_choi} in Appendix~\ref{app:otoc_lower_shadow}. We employ a proof technique based on \cite{huang2020predicting} and \cite{chen2022quantum}. It consists of showing that any protocol that satisfies the theorem's assumptions can also be used to distinguish between a set of channels expressible in the form $\mathcal{E} \circ \tilde{O} \circ \mathcal{E}$, chosen to be as hard to distinguish as possible. A lower bound on the complexity of the latter task then implies a lower bound on the former. The lower bound on the distinguishability task can be shown by encoding it as a communication protocol between three parties -- Alice (the encoding party), Bob (the decoding party), and Loki (the interfering party) -- and subsequently bounding the mutual information between the encoding and decoding parties.
\end{proofsketch}

As a bonus, Theorems~\ref{theorem:simul_otoc_shadows_nq} and \ref{theorem:shadow_lowerbound} imply tight bounds for learning \textit{arbitrary, weight-$w$} PTM elements of general unknown channels:
\begin{corollary}[Tight bounds for learning PTM elements, informal] \label{theorem:ptm_tight}
Consider the task of simultaneously estimating $M$ PTM elements $\{ \tr(P_i \mathcal{N}(Q_i))/2^n \}_{i=1}^M$ to additive error $\epsilon$, where $P_i,Q_i$ are of weight $w$. In the non-adaptive setting where each run consists of initializing as product states, querying the channel $\mathcal{N}$, and performing local measurements, the query complexity of this problem is $N = \Theta(9^w \log(M)/\epsilon^2)$.

\noindent If, additionally, the PTM elements are promised to be diagonal, i.e. $\forall~i,\ P_i=Q_i$, then the query complexity improves to $N = \Theta(3^w \log(M)/\epsilon^2)$.
\end{corollary}

These results supplement the upper/lower bounds for learning \textit{all} PTM elements, studied in \cite{caro2024learning} for general channels and in \cite{chen2022quantum} for Pauli channels. The optimal algorithms are based on applying Pauli shadows to the Choi state of $\mathcal{N}$.\\

\paragraph*{Exponential learning separations.}
Our next set of main results exposes exponential separations in the difficulty of certain OTOC-learning tasks.

The first task that we consider is the simultaneous estimation of all \textit{diagonal} OTOCs $\left \{\tr(O \mathcal{E}(Q) O \mathcal{E}(Q))/2^n \right \}_{Q \in {\mathcal{P}}_n}$. We show an exponential separation in query complexity between protocols with ancillas versus those without ancillas:
\begin{theorem}[Exponential separation in query complexity for estimating diagonal OTOCs with and without ancillas] \label{theorem:expsep_diag}
Let $\mathcal{E}$ be an arbitrary doubly-stochastic map and $O$ be a Pauli operator.
\begin{enumerate}
    \item Consider a learning protocol that queries $\mathcal{O}_{\mathcal{E}, O}$ on arbitrary $n$-qubit initial density matrices with no ancillas, but can be adaptive and possess quantum memory. 
    
    Such a protocol requires $\Omega(2^n/\epsilon^2)$ queries to solve the diagonal OTOC estimation problem.

    \item In contrast, consider the same learning protocol, but with access to $n$ ancillas. There is a non-adaptive protocol requiring no quantum memory that solves the same problem, using $O(n/\epsilon^2)$ queries to $\mathcal{O}_{\mathcal{E}, O}$.

    \item Furthermore, if we allow $0 \leq n_{\text{anc}} \leq n$ ancillas but disallow adaptivity and quantum memory, then the query complexity of this problem is $\Theta(n2^{n-n_{\text{anc}}}/\epsilon^2)$.
    
\end{enumerate}
\end{theorem}

\begin{proofsketch}
A detailed proof is provided in Appendices~\ref{app:proof_ancilla} and \ref{app:ancilla_non_adapt}. The proof of the lower bound is based on the observation that the task of estimating all diagonal OTOCs can be reduced to the task of discriminating between channels of the form $\mathcal{O}_{\mathcal{E},O} \coloneqq \mathcal{E}^\dagger \circ \widetilde{O} \circ \mathcal{E}$, and that hard instances of this distinguishing problem coincide with the hard instances of the Pauli eigenvalue estimation problem (defined in e.g. \cite{chen2022quantum,chen2024tight,chen2025efficient}). This allows the present problem to inherit the lower bounds of the Pauli eigenvalue estimation problem, which are proved using the learning tree formalism \cite{chen2022exponential} and Le Cam’s two-point method for the adaptive case, and by encoding the problem as a communication protocol in the non-adaptive case. The algorithm of \cite{chiew2026quantum}, Theorem~1 based on Bell/Pauli operator sampling is optimal for the $n_{\text{anc}} = n$ case, while an adaptation of the algorithm of \cite{chen2022quantum} for the Pauli channel eigenvalue problem is optimal for the non-adaptive $0 \leq n_{\text{anc}} \leq n$ case.
\end{proofsketch}

The above result implies that the problem of estimating diagonal OTOCs hinges on the availability of ancilla qubits -- neither adaptivity, nor quantum memory, i.e. the ability to repeatedly apply $\mathcal{O}_{\mathcal{E}, O}$, is sufficient to overcome the exponential gap. This mirrors the gap present in the Pauli channel eigenvalue estimation problem \cite{chen2022quantum,chen2024tight,chen2025efficient}.

Similar to Corollary~\ref{theorem:ptm_tight}, the above result also implies tight bounds for the problem of learning all $4^n-1$ diagonal PTM elements of general unknown channels (Theorem~\ref{theorem:diag_ptm_optimal}, Appendix~\ref{app:family_diag_ptm}).

\vspace{5pt}

Secondly, we consider the (harder) task of estimating all OTOCs $\left \{\tr(O \mathcal{E}(Q) O \mathcal{E}(Q'))/2^n \right \}_{Q, Q' \in {\mathcal{P}}_n}$, including off-diagonal ones. We show that this task exhibits an exponential separation in query complexity between protocols with quantum memory versus those without:
\begin{theorem}[Exponential separation in query complexity for estimating all OTOCs with and without quantum memory] \label{theorem:exp_sep_all_otocs_query}
Let $\mathcal{E}$ be an arbitrary doubly-stochastic map and $O$ be a Pauli operator.
\begin{itemize}
    \item Consider a learning protocol that queries $\mathcal{O}_{\mathcal{E}, O}$ on arbitrary initial density matrices (of any size), and can be adaptive, but does not possess quantum memory.
    Such a protocol requires $\Omega(4^n/\epsilon^2)$ runs/measurements to solve the All OTOC estimation problem.

    \item In contrast, consider the same protocol, but with quantum memory. There is such a protocol that solves the same problem, using $O(n/\epsilon^4)$ queries of $\mathcal{O}_{\mathcal{E}, O}$.
\end{itemize}
\end{theorem}

\begin{proofsketch}
A detailed proof is provided in Appendix~\ref{app:proof_memory}. For the lower bound, similar to the proof of Theorem~\ref{theorem:expsep_diag}, it is based on reducing the task of estimating all OTOCs to the task of discriminating between channels of the form $\mathcal{E}^\dagger \circ \widetilde{O} \circ \mathcal{E}$. The hard instances of this problem coincide with the hard instances of the problem of estimating the Pauli transfer matrix elements of general channels of \cite{caro2024learning}, allowing the present problem to inherit their lower bounds, which are proved via the learning tree formalism \cite{chen2022exponential} and Le Cam’s two-point method.
The upper bound is achieved by applying the Pauli shadow tomography algorithm of \cite{huang2021information} on the Choi states of the operator $\mathcal{O}_{\mathcal{E}, O}$.
\end{proofsketch}
This result thus shows that the difficulty of estimating all OTOCs is characterized by the presence of quantum memory, mirroring the problem of estimating all PTM entries of a channel \cite{caro2024learning} and the estimation of all Pauli expectation values of quantum states \cite{huang2021information}.

\section{Discussion}

Vectorizing time-evolved Heisenberg operators allows many correlators and OTOCs to be treated as a \emph{single} estimation problem: one operator-shadow data set can be post-processed into large families of observables, rather than rerunning bespoke circuits for each target. In this sense, a key advantage of the doubled-Hilbert-space setting is a sort of \emph{parallelism} of the operator estimation process. Another advantage is that it can reduce circuit \emph{depth} when $U$ is deep, since the ricochet identity permits forward/backward pieces to be implemented in parallel across the two halves. Of course, this comes at the cost of doubling the number of qubits (and, depending on the variant, requiring access to $U^{T}$ and/or entangling measurements). We therefore expect the approach to be most attractive when measurement overhead dominates and additional qubits are available.

Simultaneous access to many correlators/OTOCs is particularly useful when one aims to reconstruct structured dynamical objects (e.g., response landscapes, dynamical susceptibilities/structure factors, or disorder-averaged spreading profiles) where the naive ``one observable at a time'' paradigm can be shot-prohibitive. One reason such multi-observable dynamical characterization has been less explored is the presumption that it is inherently hard; our framework makes it routine once the experimenter can prepare and measure in the vectorized picture.

Our lower bounds and learning separations establish fundamental limitations on the difficulty of learning OTOCs in a practical setting, complementing sparse existing results on their complexity \cite{cotler2023information,schuster2023learning}.
Open questions include: (i) which separations persist when we limit only to unitary dynamics and other more restrictive physically motivated access models; (ii) how the ``positioning'' of noise $\mathcal{E}$ (forward-only/backward-only/interleaved) changes achievable protocols and lower bounds; and (iii) finer tradeoffs among Bell/operator sampling and operator-shadows, including hybrids that optimize qubit overhead versus depth versus sample complexity.

\section{Acknowledgements}
We thank Giuseppe Carleo and Richard Küng for helpful conversations.
SC acknowledges the funding support from NCCR SPIN, a National Centre of Competence in Research, funded by the Swiss National Science Foundation (grant number 225153). 
AA and ZH acknowledge support from the Sandoz Family Foundation-Monique de Meuron program for Academic Promotion.

\bibliography{quantum}

@article{lawrence2002mutually,
   title={Mutually unbiased binary observable sets on<i>N</i>qubits},
   volume={65},
   ISSN={1094-1622},
   url={http://dx.doi.org/10.1103/PhysRevA.65.032320},
   DOI={10.1103/physreva.65.032320},
   number={3},
   journal={Physical Review A},
   publisher={American Physical Society (APS)},
   author={Lawrence, Jay and Brukner, {\v{C}}aslav and Zeilinger, Anton},
   year={2002},
   month=Feb }

@article{algorithmiq,
	title = {Loschmidt echo for probing operator hydrodynamics in heterogeneous structures},
	author = {Algorithmiq},
	journal = {Algorithmiq},
	url = {https://algorithmiq.fi/files/model-information-flow-complex-material-document.pdf},
	year = {2025}
}

@inproceedings{chen2021exponential,
	title = {Exponential separations between learning with and without quantum memory},
	author = {Chen, Sitan and Cotler, Jordan and Huang, Hsin-Yuan and Li, Jerry},
	booktitle = {2021 IEEE 62nd Annual Symposium on Foundations of Computer Science (FOCS)},
	pages = {574--585},
	year = {2022},
	organization = {IEEE},
	doi = {10.1109/FOCS52979.2021.00063},
	url = {https://ieeexplore.ieee.org/document/9719827}
}

@article{mele2024introduction,
  title={Introduction to Haar measure tools in quantum information: A beginner's tutorial},
  author={Mele, Antonio Anna},
  journal={Quantum},
  volume={8},
  pages={1340},
  year={2024},
  publisher={Verein zur F{\"o}rderung des Open Access Publizierens in den Quantenwissenschaften}
}

@article{jalabert2001environment,
  title={Environment-independent decoherence rate in classically chaotic systems},
  author={Jalabert, Rodolfo A and Pastawski, Horacio M},
  journal={Physical review letters},
  volume={86},
  number={12},
  pages={2490},
  year={2001},
  publisher={APS}
}

@article{peres1984stability,
  title={Stability of quantum motion in chaotic and regular systems},
  author={Peres, Asher},
  journal={Physical Review A},
  volume={30},
  number={4},
  pages={1610},
  year={1984},
  publisher={APS}
}

@article{montanaro2017learning,
	title = {Learning stabilizer states by Bell sampling},
	author = {Montanaro, Ashley},
	journal = {arXiv preprint arXiv:1707.04012},
	year = {2017},
	url = {https://arxiv.org/abs/1707.04012},
	doi = {10.48550/arXiv.1707.04012}
}

@article{gross2021schur,
	title = {Schur--Weyl duality for the Clifford group with applications: Property testing, a robust Hudson theorem, and de Finetti representations},
	author = {Gross, David and Nezami, Sepehr and Walter, Michael},
	journal = {Communications in Mathematical Physics},
	volume = {385},
	number = {3},
	pages = {1325--1393},
	year = {2021},
	publisher = {Springer},
	doi = {10.1007/s00220-021-04118-7},
	url = {https://link.springer.com/article/10.1007%2Fs00220-021-04118-7}
}

@article{caro2024learning,
   title={Learning Quantum Processes and Hamiltonians via the Pauli Transfer Matrix},
   volume={5},
   ISSN={2643-6817},
   url={http://dx.doi.org/10.1145/3670418},
   DOI={10.1145/3670418},
   number={2},
   journal={ACM Transactions on Quantum Computing},
   publisher={Association for Computing Machinery (ACM)},
   author={Caro, Matthias C.},
   year={2024},
   month=jun, pages={1–53} }

@article{swingle2016measuring,
  title = {Measuring the scrambling of quantum information},
  author = {Swingle, Brian and Bentsen, Gregory and Schleier-Smith, Monika and Hayden, Patrick},
  journal = {Phys. Rev. A},
  volume = {94},
  issue = {4},
  pages = {040302},
  numpages = {6},
  year = {2016},
  month = {Oct},
  publisher = {American Physical Society},
  doi = {10.1103/PhysRevA.94.040302},
  url = {https://link.aps.org/doi/10.1103/PhysRevA.94.040302}
}

@article{von2018operator,
	title = {Operator hydrodynamics, OTOCs, and entanglement growth in systems without conservation laws},
	author = {von Keyserlingk, Curt W and Rakovszky, Tibor and Pollmann, Frank and Sondhi, Shivaji Lal},
	journal = {Physical Review X},
	volume = {8},
	number = {2},
	pages = {021013},
	year = {2018},
	publisher = {APS},
	url = {https://journals.aps.org/prx/abstract/10.1103/PhysRevX.8.021013},
	doi = {10.1103/PhysRevX.8.021013}
}

@article{mi2021information,
	title = {Information scrambling in quantum circuits},
	author = {Mi, Xiao and Roushan, Pedram and Quintana, Chris and Mandra, Salvatore and Marshall, Jeffrey and Neill, Charles and Arute, Frank and Arya, Kunal and Atalaya, Juan and Babbush, Ryan and others},
	journal = {Science},
	volume = {374},
	number = {6574},
	pages = {1479--1483},
	year = {2021},
	publisher = {American Association for the Advancement of Science},
	url = {https://www.science.org/doi/10.1126/science.abg5029},
	doi = {10.1126/science.abg5029}
}

@article{landsman2019verified,
	title = {Verified quantum information scrambling},
	volume = {567},
	issn = {1476-4687},
	url = {https://doi.org/10.1038/s41586-019-0952-6},
	doi = {10.1038/s41586-019-0952-6},
	pages = {61--65},
	number = {7746},
	journal = {Nature},
	shortjournal = {Nature},
	author = {Landsman, K. A. and Figgatt, C. and Schuster, T. and Linke, N. M. and Yoshida, B. and Yao, N. Y. and Monroe, C.},
	date = {2019-03-01},
}

@article{yoshida2019disentangling,
  title = {Disentangling Scrambling and Decoherence via Quantum Teleportation},
  author = {Yoshida, Beni and Yao, Norman Y.},
  journal = {Phys. Rev. X},
  volume = {9},
  issue = {1},
  pages = {011006},
  numpages = {17},
  year = {2019},
  month = {Jan},
  publisher = {American Physical Society},
  doi = {10.1103/PhysRevX.9.011006},
  url = {https://link.aps.org/doi/10.1103/PhysRevX.9.011006}
}

@article{huang2021quantum,
	title = {Quantum advantage in learning from experiments},
	author = {{Huang}, Hsin-Yuan and {Broughton}, Michael and {Cotler}, Jordan and {Chen}, Sitan and {Li}, Jerry and {Mohseni}, Masoud and {Neven}, Hartmut and {Babbush}, Ryan and {Kueng}, Richard and {Preskill}, John and {McClean}, Jarrod R.},
	journal = {Science},
	volume = {376},
	number = {6598},
	pages = {1182--1186},
	year = {2022},
	publisher = {American Association for the Advancement of Science},
	url = {https://www.science.org/doi/10.1126/science.abn7293},
	doi = {10.1126/science.abn7293}
}

@article{chen2022quantum,
	title = {Quantum advantages for Pauli channel estimation},
	author = {Chen, Senrui and Zhou, Sisi and Seif, Alireza and Jiang, Liang},
	journal = {Physical Review A},
	volume = {105},
	number = {3},
	pages = {032435},
	year = {2022},
	publisher = {APS},
	url = {https://journals.aps.org/pra/abstract/10.1103/PhysRevA.105.032435},
	doi = {10.1103/PhysRevA.105.032435}
}

@article{fisher2023random,
	title = {Random quantum circuits},
	author = {Fisher, Matthew PA and Khemani, Vedika and Nahum, Adam and Vijay, Sagar},
	journal = {Annual Review of Condensed Matter Physics},
	volume = {14},
	pages = {335--379},
	year = {2023},
	publisher = {Annual Reviews},
	url = {https://www.annualreviews.org/doi/10.1146/annurev-conmatphys-031720-030658},
	doi = {10.1146/annurev-conmatphys-031720-030658}
}

@article{nahum2018operator,
	title = {Operator spreading in random unitary circuits},
	author = {Nahum, Adam and Vijay, Sagar and Haah, Jeongwan},
	journal = {Physical Review X},
	volume = {8},
	number = {2},
	pages = {021014},
	year = {2018},
	publisher = {APS},
	url = {https://journals.aps.org/prx/abstract/10.1103/PhysRevX.8.021014},
	doi = {10.1103/PhysRevX.8.021014}
}

@article{huang2020predicting,
	title = {Predicting many properties of a quantum system from very few measurements},
	author = {Huang, Hsin-Yuan and Kueng, Richard and Preskill, John},
	journal = {Nature Physics},
	volume = {16},
	number = {10},
	pages = {1050--1057},
	year = {2020},
	publisher = {Nature Publishing Group},
	doi = {10.1038/s41567-020-0932-7},
	url = {https://www.nature.com/articles/s41567-020-0932-7}
}

@article{huang2021information,
	title = {Information-Theoretic Bounds on Quantum Advantage in Machine Learning},
	author = {Huang, Hsin-Yuan and Kueng, Richard and Preskill, John},
	journal = {Phys. Rev. Lett.},
	volume = {126},
	issue = {19},
	pages = {190505},
	numpages = {7},
	year = {2021},
	month = {May},
	publisher = {American Physical Society},
	doi = {10.1103/PhysRevLett.126.190505},
	url = {https://link.aps.org/doi/10.1103/PhysRevLett.126.190505}
}

@article{fuchs1999cryptographic,
	title = {Cryptographic distinguishability measures for quantum-mechanical states},
	author = {Fuchs, Christopher A and Van De Graaf, Jeroen},
	journal = {IEEE Transactions on Information Theory},
	volume = {45},
	number = {4},
	pages = {1216--1227},
	year = {1999},
	publisher = {IEEE},
	doi = {10.1109/18.761271},
	url = {https://ieeexplore.ieee.org/document/761271}
}

@article{elben2022randomized,
	title = {The randomized measurement toolbox},
	author = {Elben, Andreas and Flammia, Steven T and Huang, Hsin-Yuan and Kueng, Richard and Preskill, John and Vermersch, Beno{\^\i}t and Zoller, Peter},
	journal = {Nature Review Physics},
	year = {2022},
	doi = {10.1038/s42254-022-00535-2},
	url = {https://www.nature.com/articles/s42254-022-00535-2}
}

@article{zhao2024learning,
	title = {Learning quantum states and unitaries of bounded gate complexity},
	author = {Zhao, Haimeng and Lewis, Laura and Kannan, Ishaan and Quek, Yihui and Huang, Hsin-Yuan and Caro, Matthias C},
	journal = {PRX Quantum},
	volume = {5},
	number = {4},
	pages = {040306},
	year = {2024},
	publisher = {APS},
	doi = {10.1103/PRXQuantum.5.040306},
	url = {https://link.aps.org/doi/10.1103/PRXQuantum.5.040306}
}

@article{google2025observation,
	title = {Observation of constructive interference at the edge of quantum ergodicity},
	journal = {Nature},
	volume = {646},
	number = {8086},
	pages = {825--830},
	year = {2025},
	publisher = {Nature Publishing Group UK London},
	url = {https://doi.org/10.1038/s41586-025-09526-6},
	doi = {https://doi.org/10.1038/s41586-025-09526-6}
}

@article{castaneda2025hamiltonian,
	title = {Hamiltonian learning via shadow tomography of pseudo-choi states},
	author = {Castaneda, Juan and Wiebe, Nathan},
	journal = {Quantum},
	volume = {9},
	pages = {1700},
	year = {2025},
	publisher = {Verein zur F{\"o}rderung des Open Access Publizierens in den Quantenwissenschaften},
	url = {https://quantum-journal.org/papers/q-2025-04-09-1700/},
	doi = {10.22331/q-2025-04-09-1700}
}

@article{levy2024classical,
	title = {Classical shadows for quantum process tomography on near-term quantum computers},
	author = {Levy, Ryan and Luo, Di and Clark, Bryan K},
	journal = {Physical Review Research},
	volume = {6},
	number = {1},
	pages = {013029},
	year = {2024},
	publisher = {APS},
	url = {https://journals.aps.org/prresearch/abstract/10.1103/PhysRevResearch.6.013029},
	doi = {10.1103/PhysRevResearch.6.013029}
}

@article{flammia2011direct,
	title = {Direct fidelity estimation from few Pauli measurements},
	author = {Flammia, Steven T and Liu, Yi-Kai},
	journal = {Physical review letters},
	volume = {106},
	number = {23},
	pages = {230501},
	year = {2011},
	publisher = {APS},
	url = {https://journals.aps.org/prl/abstract/10.1103/PhysRevLett.106.230501},
	doi = {10.1103/PhysRevLett.106.230501}
}

@article{baez2020dynamical,
	title = {Dynamical structure factors of dynamical quantum simulators},
	author = {Baez, Maria Laura and Goihl, Marcel and Haferkamp, Jonas and Bermejo-Vega, Juani and Gluza, Marek and Eisert, Jens},
	journal = {Proceedings of the National Academy of Sciences},
	volume = {117},
	number = {42},
	pages = {26123--26134},
	year = {2020},
	publisher = {National Acad Sciences},
	url = {https://www.pnas.org/doi/10.1073/pnas.2006103117},
	doi = {10.1073/pnas.2006103117}
}

@article{flammia2020efficient,
	title = {Efficient estimation of Pauli channels},
	author = {Flammia, Steven T and Wallman, Joel J},
	journal = {ACM Transactions on Quantum Computing},
	volume = {1},
	number = {1},
	pages = {1--32},
	year = {2020},
	publisher = {ACM New York, NY, USA},
	doi = {10.1145/3408039}
}

@article{kunjummen2023shadow,
	title = {Shadow process tomography of quantum channels},
	author = {Kunjummen, Jonathan and Tran, Minh C and Carney, Daniel and Taylor, Jacob M},
	journal = {Physical Review A},
	volume = {107},
	number = {4},
	pages = {042403},
	year = {2023},
	publisher = {APS},
	url = {https://journals.aps.org/pra/abstract/10.1103/PhysRevA.107.042403},
	doi = {10.1103/PhysRevA.107.042403}
}

@article{jamiolkowski1972linear,
	title = {Linear transformations which preserve trace and positive semidefiniteness of operators},
	author = {Jamio{\l}kowski, Andrzej},
	journal = {Reports on mathematical physics},
	volume = {3},
	number = {4},
	pages = {275--278},
	year = {1972},
	publisher = {Elsevier},
	url = {https://www.sciencedirect.com/science/article/abs/pii/0034487772900110},
	doi = {10.1016/0034-4877(72)90011-0}
}

@article{choi1975completely,
	title = {Completely positive linear maps on complex matrices},
	author = {Choi, Man-Duen},
	journal = {Linear algebra and its applications},
	volume = {10},
	number = {3},
	pages = {285--290},
	year = {1975},
	publisher = {Elsevier},
	doi = {10.1016/0024-3795(75)90075-0},
	url = {https://doi.org/10.1016/0024-3795(75)90075-0}
}

@article{von2022operator,
	title = {Operator backflow and the classical simulation of quantum transport},
	author = {Von Keyserlingk, Curt and Pollmann, Frank and Rakovszky, Tibor},
	journal = {Physical Review B},
	volume = {105},
	number = {24},
	pages = {245101},
	year = {2022},
	publisher = {APS},
	url = {https://journals.aps.org/prb/abstract/10.1103/PhysRevB.105.245101},
	doi = {10.1103/PhysRevB.105.245101}
}

@article{kim2026fundamental,
  title={On the fundamental resource for exponential advantage in quantum channel learning},
  author={Kim, Minsoo and Oh, Changhun},
  journal={Nature Communications},
  year={2026},
  publisher={Nature Publishing Group UK London}
}

@article{chiew2026quantum,
  title={Quantum simulation in the Heisenberg picture via vectorization},
  author={Chiew, Shao-Hen and Angrisani, Armando and Holmes, Zo{\"e} and Carleo, Giuseppe},
  journal={arXiv preprint arXiv:2602.20154},
  year={2026}
}

@article{nahum2022real,
  title={Real-time correlators in chaotic quantum many-body systems},
  author={Nahum, Adam and Roy, Sthitadhi and Vijay, Sagar and Zhou, Tianci},
  journal={Physical Review B},
  volume={106},
  number={22},
  pages={224310},
  year={2022},
  publisher={APS}
}

@article{cotler2017chaos,
	title = {Chaos, complexity, and random matrices},
	volume = {2017},
	issn = {1029-8479},
	url = {https://doi.org/10.1007/JHEP11(2017)048},
	doi = {10.1007/JHEP11(2017)048},
	pages = {48},
	number = {11},
	journal = {Journal of High Energy Physics},
	shortjournal = {Journal of High Energy Physics},
	author = {Cotler, Jordan and Hunter-Jones, Nicholas and Liu, Junyu and Yoshida, Beni},
	date = {2017-11-09},
}

@article{moudgalya2024symmetries,
  title = {Symmetries as Ground States of Local Superoperators: Hydrodynamic Implications},
  author = {Moudgalya, Sanjay and Motrunich, Olexei I.},
  journal = {PRX Quantum},
  volume = {5},
  issue = {4},
  pages = {040330},
  numpages = {41},
  year = {2024},
  month = {Nov},
  publisher = {American Physical Society},
  doi = {10.1103/PRXQuantum.5.040330},
  url = {https://link.aps.org/doi/10.1103/PRXQuantum.5.040330}
}

@INPROCEEDINGS{gleinig2021efficient,
  author={Gleinig, Niels and Hoefler, Torsten},
  booktitle={2021 58th ACM/IEEE Design Automation Conference (DAC)}, 
  title={An Efficient Algorithm for Sparse Quantum State Preparation}, 
  year={2021},
  volume={},
  number={},
  pages={433-438},
  doi={10.1109/DAC18074.2021.9586240}}

@article{zhang2022quantum,
  title = {Quantum State Preparation with Optimal Circuit Depth: Implementations and Applications},
  author = {Zhang, Xiao-Ming and Li, Tongyang and Yuan, Xiao},
  journal = {Phys. Rev. Lett.},
  volume = {129},
  issue = {23},
  pages = {230504},
  numpages = {6},
  year = {2022},
  month = {Nov},
  publisher = {American Physical Society},
  doi = {10.1103/PhysRevLett.129.230504},
  url = {https://link.aps.org/doi/10.1103/PhysRevLett.129.230504}
}

@article{dowling2025magic,
  title = {Magic Resources of the Heisenberg Picture},
  author = {Dowling, Neil and Kos, Pavel and Turkeshi, Xhek},
  journal = {Phys. Rev. Lett.},
  volume = {135},
  issue = {5},
  pages = {050401},
  numpages = {10},
  year = {2025},
  month = {Jul},
  publisher = {American Physical Society},
  doi = {10.1103/p7xt-s9nz},
  url = {https://link.aps.org/doi/10.1103/p7xt-s9nz}
}

@article{xu2024scrambling,
  title = {Scrambling Dynamics and Out-of-Time-Ordered Correlators in Quantum Many-Body Systems},
  author = {Xu, Shenglong and Swingle, Brian},
  journal = {PRX Quantum},
  volume = {5},
  issue = {1},
  pages = {010201},
  numpages = {46},
  year = {2024},
  month = {Jan},
  publisher = {American Physical Society},
  doi = {10.1103/PRXQuantum.5.010201},
  url = {https://link.aps.org/doi/10.1103/PRXQuantum.5.010201}
}

@article{rakovszky2018diffusive,
	title = {Diffusive hydrodynamics of out-of-time-ordered correlators with charge conservation},
	author = {Rakovszky, Tibor and Pollmann, Frank and von Keyserlingk, Curt W},
	journal = {Physical Review X},
	volume = {8},
	number = {3},
	pages = {031058},
	year = {2018},
	publisher = {APS}
}

@article{chen2024tight,
	title = {Tight bounds on Pauli channel learning without entanglement},
	author = {Chen, Senrui and Oh, Changhun and Zhou, Sisi and Huang, Hsin-Yuan and Jiang, Liang},
	journal = {Physical Review Letters},
	volume = {132},
	number = {18},
	pages = {180805},
	year = {2024},
	publisher = {APS}
}

@article{schuster2023operator,
  title = {Operator Growth in Open Quantum Systems},
  author = {Schuster, Thomas and Yao, Norman Y.},
  journal = {Phys. Rev. Lett.},
  volume = {131},
  issue = {16},
  pages = {160402},
  numpages = {6},
  year = {2023},
  month = {Oct},
  publisher = {American Physical Society},
  doi = {10.1103/PhysRevLett.131.160402},
  url = {https://link.aps.org/doi/10.1103/PhysRevLett.131.160402}
}

@article{hangleiter2024bell,
  title = {Bell Sampling from Quantum Circuits},
  author = {Hangleiter, Dominik and Gullans, Michael J.},
  journal = {Phys. Rev. Lett.},
  volume = {133},
  issue = {2},
  pages = {020601},
  numpages = {7},
  year = {2024},
  month = {Jul},
  publisher = {American Physical Society},
  doi = {10.1103/PhysRevLett.133.020601},
  url = {https://link.aps.org/doi/10.1103/PhysRevLett.133.020601}
}

@article{qi2019quantum,
	title = {Quantum epidemiology: operator growth, thermal effects, and {SYK}},
	volume = {2019},
	issn = {1029-8479},
	url = {https://doi.org/10.1007/JHEP08(2019)012},
	doi = {10.1007/JHEP08(2019)012},
	pages = {12},
	number = {8},
	journal = {Journal of High Energy Physics},
	author = {Qi, Xiao-Liang and Streicher, Alexandre},
	date = {2019-08-02},
}

@inproceedings{li2024nearly,
  doi = {10.4230/LIPICS.ICALP.2025.113},
  url = {https://drops.dagstuhl.de/entities/document/10.4230/LIPIcs.ICALP.2025.113},
  author = {Li, Lvzhou and Luo, Jingquan},
  title = {Nearly Optimal Circuit Size for Sparse Quantum State Preparation},
  journal = {LIPIcs, Volume 334, ICALP 2025},
  volume = {334},
  pages = {113:1-113:19},
  publisher = {Schloss Dagstuhl – Leibniz-Zentrum für Informatik},
  year = {2025}
}

@misc{vilmart2025resource,
      title={Resource-Efficient Synthesis of Sparse Quantum States}, 
      author={Renaud Vilmart and Sunheang Ty and Chetra Mang},
      year={2025},
      archivePrefix={arXiv},
      primaryClass={quant-ph},
      url={https://arxiv.org/abs/2508.05386}, 
}

@article{cotler2023information,
  title = {Information-theoretic hardness of out-of-time-order correlators},
  author = {Cotler, Jordan and Schuster, Thomas and Mohseni, Masoud},
  journal = {Phys. Rev. A},
  volume = {108},
  issue = {6},
  pages = {062608},
  numpages = {15},
  year = {2023},
  month = {Dec},
  publisher = {American Physical Society},
  doi = {10.1103/PhysRevA.108.062608},
  url = {https://link.aps.org/doi/10.1103/PhysRevA.108.062608}
}

@inproceedings{chen2022exponential,
	title = {Exponential separations between learning with and without quantum memory},
	author = {Chen, Sitan and Cotler, Jordan and Huang, Hsin-Yuan and Li, Jerry},
	booktitle = {2021 IEEE 62nd Annual Symposium on Foundations of Computer Science (FOCS)},
	pages = {574--585},
	year = {2022},
	organization = {IEEE}
}

@article{parker2019universal,
  title = {A Universal Operator Growth Hypothesis},
  author = {Parker, Daniel E. and Cao, Xiangyu and Avdoshkin, Alexander and Scaffidi, Thomas and Altman, Ehud},
  journal = {Phys. Rev. X},
  volume = {9},
  issue = {4},
  pages = {041017},
  numpages = {29},
  year = {2019},
  month = {Oct},
  publisher = {American Physical Society},
  doi = {10.1103/PhysRevX.9.041017},
  url = {https://link.aps.org/doi/10.1103/PhysRevX.9.041017}
}

@article{sundar2022proposal,
	title = {Proposal for measuring out-of-time-ordered correlators at finite temperature with coupled spin chains},
	volume = {24},
	url = {https://doi.org/10.1088/1367-2630/ac5002},
	doi = {10.1088/1367-2630/ac5002},
	pages = {023037},
	number = {2},
	journal = {New Journal of Physics},
	publisher = {{IOP} Publishing},
	author = {Sundar, Bhuvanesh and Elben, Andreas and Joshi, Lata Kh and Zache, Torsten V},
	date = {2022-02},
}

@article{green2022experimental,
  title = {Experimental Measurement of Out-of-Time-Ordered Correlators at Finite Temperature},
  author = {Green, Alaina M. and Elben, A. and Alderete, C. Huerta and Joshi, Lata Kh and Nguyen, Nhung H. and Zache, Torsten V. and Zhu, Yingyue and Sundar, Bhuvanesh and Linke, Norbert M.},
  journal = {Phys. Rev. Lett.},
  volume = {128},
  issue = {14},
  pages = {140601},
  numpages = {6},
  year = {2022},
  month = {Apr},
  publisher = {American Physical Society},
  doi = {10.1103/PhysRevLett.128.140601},
  url = {https://link.aps.org/doi/10.1103/PhysRevLett.128.140601}
}

@article{chen2025efficient,
	title = {Efficient Pauli channel estimation with logarithmic quantum memory},
	author = {Chen, Sitan and Gong, Weiyuan},
	journal = {PRX Quantum},
	volume = {6},
	number = {2},
	pages = {020323},
	year = {2025},
	publisher = {APS}
}

@article{schuster2022many,
  title = {Many-Body Quantum Teleportation via Operator Spreading in the Traversable Wormhole Protocol},
  author = {Schuster, Thomas and Kobrin, Bryce and Gao, Ping and Cong, Iris and Khabiboulline, Emil T. and Linke, Norbert M. and Lukin, Mikhail D. and Monroe, Christopher and Yoshida, Beni and Yao, Norman Y.},
  journal = {Phys. Rev. X},
  volume = {12},
  issue = {3},
  pages = {031013},
  numpages = {61},
  year = {2022},
  month = {Jul},
  publisher = {American Physical Society},
  doi = {10.1103/PhysRevX.12.031013},
  url = {https://link.aps.org/doi/10.1103/PhysRevX.12.031013}
}

@article{Pineda_Jim_nez_2026,
   title={Operator delocalization in disordered spin chains via exact MPO marginals},
   volume={28},
   ISSN={1367-2630},
   url={http://dx.doi.org/10.1088/1367-2630/ae774b},
   DOI={10.1088/1367-2630/ae774b},
   number={6},
   journal={New Journal of Physics},
   publisher={IOP Publishing},
   author={Pineda-Jiménez, Jonnathan and Collura, Mario and Passarelli, Gianluca and Lucignano, Procolo and Rossini, Davide and Russomanno, Angelo},
   year={2026},
   month=June, pages={064512} }

@article{khemani2018operator,
  title = {Operator Spreading and the Emergence of Dissipative Hydrodynamics under Unitary Evolution with Conservation Laws},
  author = {Khemani, Vedika and Vishwanath, Ashvin and Huse, David A.},
  journal = {Phys. Rev. X},
  volume = {8},
  issue = {3},
  pages = {031057},
  numpages = {25},
  year = {2018},
  month = {Sep},
  publisher = {American Physical Society},
  doi = {10.1103/PhysRevX.8.031057},
  url = {https://link.aps.org/doi/10.1103/PhysRevX.8.031057}
}

@article{roberts2018operator,
	title = {Operator growth in the {SYK} model},
	volume = {2018},
	issn = {1029-8479},
	url = {https://doi.org/10.1007/JHEP06(2018)122},
	doi = {10.1007/JHEP06(2018)122},
	pages = {122},
	number = {6},
	journal = {Journal of High Energy Physics},
	author = {Roberts, Daniel A. and Stanford, Douglas and Streicher, Alexandre},
	date = {2018-06-22},
}

@article{schuster2023learning,
	title = {Learning quantum systems via out-of-time-order correlators},
	author = {Schuster, Thomas and Niu, Murphy and Cotler, Jordan and O'Brien, Thomas and McClean, Jarrod R and Mohseni, Masoud},
	journal = {Physical Review Research},
	volume = {5},
	number = {4},
	pages = {043284},
	year = {2023},
	publisher = {APS}
}

@article{mcginley2022quantifying,
  title = {Quantifying information scrambling via classical shadow tomography on programmable quantum simulators},
  author = {McGinley, Max and Leontica, Sebastian and Garratt, Samuel J. and Jovanovic, Jovan and Simon, Steven H.},
  journal = {Phys. Rev. A},
  volume = {106},
  issue = {1},
  pages = {012441},
  numpages = {14},
  year = {2022},
  month = {Jul},
  publisher = {American Physical Society},
  doi = {10.1103/PhysRevA.106.012441},
  url = {https://link.aps.org/doi/10.1103/PhysRevA.106.012441}
}

@article{yoshimura2025theory,
  title = {Theory of irreversibility in quantum many-body systems},
  author = {Yoshimura, Takato and S\'a, Lucas},
  journal = {Phys. Rev. E},
  volume = {111},
  issue = {6},
  pages = {064135},
  numpages = {22},
  year = {2025},
  month = {Jun},
  publisher = {American Physical Society},
  doi = {10.1103/82f6-qdyd},
  url = {https://link.aps.org/doi/10.1103/82f6-qdyd}
}

@article{zhang2025quantum,
      title={Quantum computation of molecular geometry via many-body nuclear spin echoes}, 
      author={C. Zhang and R. G. Cortiñas and A. H. Karamlou and N. Noll and J. Provazza and J. Bausch and S. Shirobokov and A. White and M. Claassen and S. H. Kang and others},
      year={2025},
      journal = {arXiv:2510.19550},
      archivePrefix={arXiv},
      primaryClass={quant-ph},
      url={https://arxiv.org/abs/2510.19550}
}

@online{wolf2012quantumchannels,
	author = {Michael M. Wolf},
	title = {Quantum Channels and Operations - Guided Tour},
	year = {2012},
	month = {07},
	doi = {},
	pages = {},
    url={https://mediatum.ub.tum.de/doc/1701036/1701036.pdf}
}


\clearpage

\onecolumngrid

\counterwithin*{equation}{section}
\renewcommand\theequation{\thesection\arabic{equation}}

\renewcommand\partname{} 
\appendix
\begin{center}
 {\Large \textbf{Appendix} }   
\end{center}

\part{}
\parttoc 

\vspace{10pt}

\clearpage

\section{Preliminaries}

\subsection{The vectorization and Choi-Jamiolkowski maps} \label{app:prelim_vect_choi}
We begin by briefly reviewing relevant properties of the vectorization and Choi-Jamiolkowski maps, also setting up notation along the way. Further relevant details on these topics will be introduced as needed, and can be found e.g. in \cite{chiew2026quantum} and \cite{wolf2012quantumchannels} respectively.

The normalized vectorization map with respect to an $n$-qubit orthogonal operator basis $\mathcal{Q} = \{Q_k\}_{k=1}^{2^n}$ and a $2n$-qubit orthonormal state basis $\{\ket{k}\}_{k=1}^{4^n}$ is defined as:
\begin{equation}
    \kett{O}_\mathcal{Q} \equiv \sum_{k=1}^{4^n} \frac{\tr(Q_k O)}{\sqrt{\sum_i \abs{\tr(Q_i O)}^2}} \ket{k},
\end{equation}
and maps between $n$-qubit linear operators $O \in \mathcal{L}(\mathcal{H})$ and $2n$-qubit pure states $\kett{O} \in \mathcal{H}'$, where $\mathcal{H}$ and $\mathcal{H}'$ denote the $n$- and $2n$-qubit Hilbert spaces respectively. For convenience, order the doubled space $\mathcal{H}'$ spatially as $\mathcal{H}' =\mathcal{H}_L \otimes \mathcal{H}_R = (\bigotimes_{i=1}^n \mathcal{H}_L^i ) \otimes (\bigotimes_{i=1}^n \mathcal{H}_R^i )$, where $\mathcal{H}_L$ and $\mathcal{H}_R$ are each of dimension $2^n$, and each $\mathcal{H}_L^i$ or $\mathcal{H}_R^i$ has local dimension 2. In particular, when $\mathcal{Q} = \mathcal{C} = \{\ketbra{i}{j}\}_{i,j=1}^{2^n}$, the computational operator basis, the vectorization map corresponds to ‘row-stacking’ a $2^n \times 2^n$ dimensional matrix to form a $4^n$ dimensional vector. For brevity, we will also omit the subscript for this choice, and write $\kett{O}_{\mathcal{C}} = \kett{O}$. Furthermore, the `ricochet identity':
\begin{equation} \label{eq:ricochet_id}
    A \otimes C^T \kett{B} = \kett{A B C}
\end{equation} 
holds for this basis.

A useful fact is that $\kett{\mathbb{I}} \in \mathcal{H}'$, the vectorization of the normalized identity operator $\mathbb{I}$ in $\mathcal{C}$, is the quantum state corresponding to $n$ Bell pairs $(\ket{00} + \ket{11})/\sqrt{2})^{\otimes n}$ (up to qubit permutation), and can be conveniently expressed as:
\begin{equation} \label{eq:bell_state}
    \kettbbra{\mathbb{I}}{\mathbb{I}} = \frac{1}{4^n} \sum_i P_i \otimes P_i^T.
\end{equation}

The Choi-Jamiolkowski isomorphism \cite{jamiolkowski1972linear,choi1975completely} states that any $n$-qubit linear map $\mathcal{N} : \mathcal{L}(\mathcal{H}) \rightarrow \mathcal{L}(\mathcal{H})$ can be mapped to a unique $2n$-qubit operator (the Choi operator of $\mathcal{N}$):
\begin{equation} \label{eq:choi_def}
    \rho_{\mathcal{N}} \equiv (\mathcal{N} \otimes \mathbb{I})(\kettbbra{\mathbb{I}}{\mathbb{I}}) \in \mathcal{L}(\mathcal{H}'),
\end{equation}
and vice versa. Hereafter, $\rho_{\mathcal{N}}$ refers to the Choi state of the channel $\mathcal{N}$ whenever the subscript is a quantum channel. Furthermore, the following (directly verifiable via Eqs.~(\ref{eq:ricochet_id}),(\ref{eq:bell_state}), (\ref{eq:choi_def})) holds:
\begin{enumerate}
    \item $\mathcal{N}$ is a quantum channel if and only if $\rho_{\mathcal{N}}$ is a quantum state with maximally mixed right marginal, i.e. $\tr_{\mathcal{H}_L}(\rho_{\mathcal{N}}) = \mathbb{I}/2^n$.

    \item $\mathcal{N}$ is unital if and only if $\rho_{\mathcal{N}}$ has a maximally mixed left marginal, i.e. $\tr_{\mathcal{H}_R}(\rho_{\mathcal{N}}) = \mathbb{I}/2^n$.

    \item $\mathcal{N}$ is a doubly-stochastic quantum channel if and only if $\rho_{\mathcal{N}}$ is a quantum state with maximally mixed left and right marginals.
\end{enumerate}

In the case where $\mathcal{N}$ corresponds to conjugation by a linear operator, Eq.~(\ref{eq:choi_def}) reduces to the vectorization map. For instance, if $\mathcal{N}(\cdot) = U(\cdot)U^\dagger$, where $U$ is a unitary operator, then $\rho_{\mathcal{N}} = (U \otimes \mathbb{I})\kettbbra{\mathbb{I}}{\mathbb{I}}(U^\dagger \otimes \mathbb{I}) = \kettbbra{U}{U}$, where we used Eq.~(\ref{eq:ricochet_id}) in the second equality.

An alternative characterization of $\mathcal{N}$ is in terms of its Transfer Matrix $M_\mathcal{Q}^\mathcal{N}$, i.e. its matrix representation in an orthonormal operator basis $\mathcal{Q} = \{Q_i\}$, with matrix elements defined by:
\begin{equation}
    \left[M_\mathcal{Q}^\mathcal{N} \right]_{ij} \equiv \frac{1}{2^n} \tr(Q_i^\dagger \mathcal{N}(Q_j)).
\end{equation}
Its matrix representation in the Pauli basis $\mathcal{P}_n$, $M_\mathcal{P}^\mathcal{N}$, is simply referred to as the Pauli Transfer Matrix (PTM) of $\mathcal{N}$. Elements of the transfer matrix of $\mathcal{N}$ can also be related to expectation values on its Choi state: $\forall O_i, O_j \in \mathcal{L}(\mathcal{H})$, we have:
\begin{equation} \label{eq:ricochet_channel}
    \tr(O_i \otimes O_j^T \rho_{\mathcal{N}}) = \frac{1}{2^n}\tr(O_i^\dagger \mathcal{N}(O_j)).
\end{equation}
To show this, it suffices to consider $P_i, P_j \in \mathcal{P}_n$ (by linearity and completeness of $\mathcal{P}_n$), and make use of Eq.~(\ref{eq:bell_state}) to yield:
\begin{equation*}
    \tr(P_i \otimes P_j^T \rho_{\mathcal{N}}) = \frac{1}{4^n} \sum_P \tr(P_i \mathcal{N}(P)) \tr(P_j^T P) = \frac{1}{2^n} \tr(P_i \mathcal{N}(P_j)).
\end{equation*}
Notably, when $O_i, O_j \in \mathcal{P}_n$, the quantity above corresponds to the PTM elements of $\mathcal{N}$.

An object of frequent study in our work is the map $\mathcal{O}_{\mathcal{E}, O} \equiv \mathcal{E}^\dagger \circ \tilde{O} \circ \mathcal{E}$, where $\mathcal{E}$ is a general quantum channel, and $\tilde{O}(\cdot) \equiv O^\dagger(\cdot)O$ dentoes conjugation by the operator $O$. It can be thought of as the generalization of the ``Loschmidt echo" unitary map $\kett{\mathbb{I}} \rightarrow \kett{O(t)}$ carried out by the unitary $O(t) \otimes \mathbb{I} = U^\dagger OU \otimes \mathbb{I}$, to the case where Heisenberg evolution $\mathcal{E}$ is non-unitary. The following properties about $\mathcal{O}_{\mathcal{E}, O}$ hold:
\begin{lemma}\label{lemma:prop_o}
Define the map $\mathcal{O}_{\mathcal{E}, O} \equiv \mathcal{E}^\dagger \circ \tilde{O} \circ \mathcal{E}$, where $\mathcal{E}$ is a quantum channel and the conjugation map $\tilde{O}(\cdot) = O(\cdot)O^\dagger$, with $O \in \mathcal{P}_n$ unless stated.
\begin{enumerate}
    \item $\mathcal{O}_{\mathcal{E}, O}$ is Hermitian-preserving.

    \item $\mathcal{O}_{\mathcal{E}, O}$ is self-adjoint.

    \item $\mathcal{O}_{\mathcal{E}, O}$ is completely-positive.

    \item If $\mathcal{E}$ is unital, then $\mathcal{O}_{\mathcal{E}, O}$ is both trace-preserving and unital (and is therefore a doubly-stochastic quantum channel).
\end{enumerate}
\end{lemma}
\begin{proof}
(1) and (2) can be directly verified, e.g. by expanding $\mathcal{O}_{\mathcal{E}, O}$ in terms of the Kraus operators of $\mathcal{E}$. (3) is true because $\tilde{P}$ and $\mathcal{E}^\dagger$ are CP, and the composition of CP maps is still CP. To show (4), observe that:
\begin{equation*}
    \mathcal{O}_{\mathcal{E}, O} \text{~is TP~} \iff \mathcal{O}_{\mathcal{E}, O}^\dagger \text{~is unital~} \varoverset{(2)}{\iff} \mathcal{O}_{\mathcal{E}, O} \text{~is unital~} \iff \mathcal{O}_{\mathcal{E}, O}(\mathbb{I}) = (\mathcal{E}^\dagger \circ \tilde{P} \circ \mathcal{E})(\mathbb{I}) = \mathbb{I}.
\end{equation*}
Therefore, if $\mathcal{E}$ is unital, the final equality holds, and $\mathcal{O}_{\mathcal{E}, O}$ is both TP and unital.
\end{proof}

The above Lemma implies that $\mathcal{O}_{\mathcal{E}, O}$ is a quantum channel if the quantum channel $\mathcal{E}$ is also unital (i.e. it is doubly-stochastic). In what follows, we will thus take $\mathcal{E}$ to be unital, to guarantee that $\mathcal{O}_{\mathcal{E}, O}$ is a bona fide physically queryable/implementable quantum channel. Note that this assumption need not be present for any of our lower bounds to hold (since lower bounds for protocols with this assumption imply lower bounds for protocols without), though for completeness we state all lower bounds with this assumption.

The Choi state of $\mathcal{O}_{\mathcal{E}, O}$ is the state $\rho_{\mathcal{O}_{\mathcal{E}, O}} \equiv (\mathcal{O}_{\mathcal{E}, O} \otimes \mathbb{I})(\kettbbra{\mathbb{I}}{\mathbb{I}})$. When $\mathcal{E}(\cdot) = U^\dagger (\cdot) U$ and $\tilde{O}(\cdot)=O(\cdot)O$, $\rho_{\mathcal{O}_{\mathcal{E}, O}}$ reduces to $\kettbbra{O(t)}{O(t)}$ (by again applying Eq.~(\ref{eq:ricochet_id})), the pure state obtained by vectorizing $O(t)=U^\dagger O U$ in the computational basis. By Lemma~\ref{lemma:prop_o}, since $\mathcal{E}$ is unital, $\rho_{\mathcal{O}_{\mathcal{E}, O}}$ is a quantum state that has maximally mixed first and second marginals. Furthermore, for the channel $\mathcal{O}_{\mathcal{E}, O} = \mathcal{E}^\dagger \circ \tilde{P} \circ \mathcal{E}$, the expectation value $\tr(O_i \otimes O_j^T \rho_{\mathcal{O}_{\mathcal{E}, O}})$ evaluates to $\tr(P_i \otimes P_j^T \rho_{\mathcal{O}_{\mathcal{E}, O}}) = \tr(P_i \mathcal{O}_{\mathcal{E}, O}(P_j))/2^n = \tr(\mathcal{E}(P_i) O \mathcal{E}(P_j) O)/2^n$, which is an OTOC. When $\mathcal{E}$ is unitary, Eq.~(\ref{eq:ricochet_channel}) reduces to the expression of OTOC for vectorized operators:
\begin{equation}
    \tr(P_i \otimes P_j^T \rho_{\mathcal{O}_{\mathcal{E}, O}}) = \bbra{O(t)} P_i \otimes P_j^T \kett{O(t)} = \frac{1}{2^n} \tr(P_i O(t) P_j O(t)),
\end{equation}
where we applied the usual `ricochet identity' in the second equality. Finally, we remark that the distribution resulting from sampling from $\rho_{\mathcal{O}_{\mathcal{E}, O}}$ in the Bell basis reduces to the Pauli distribution of $O(t)$ when Heisenberg evolution is unitary, which recovers the case studied in \cite{chiew2026quantum}.

The following Lemma shows that some quantum channels can be expressed in the ``Loschmidt-echo form" $\mathcal{E}^\dagger \circ \tilde{P} \circ \mathcal{E}$, for suitably chosen doubly stochastic channels $\mathcal{E}$ and Pauli operator $P$. As we will discuss, these channels appear as hard instances used in our lower-bound proofs.
\begin{lemma} \label{lemma:hard_instances_echo_form}
Let $\widetilde P(\cdot)=P(\cdot)P^\dagger$ denote conjugation by the Pauli operator $P \in \mathcal{P}_n$.
\begin{enumerate}
    \item The completely depolarizing channel can be written as:
    \begin{equation}
         \Lambda_0(X) \equiv \frac{\mathbb I}{2^n}\Tr(X) = (\Lambda_0 ^*\circ\widetilde P \circ \Lambda_0) (X)
    \end{equation}
    for any $P\in\mathcal P_n$. Its Choi state is:
    \begin{equation}
        \rho_{\Lambda_0} = \frac{\mathbb{I}}{4^n}.
    \end{equation}

    \item Let $Q_i\in\mathcal P_n\setminus\{\mathbb I\}$,
    $s\in\{-1,1\}$, and $0\leq\alpha\leq1$. Then the Pauli channel:
    \begin{equation} \label{eq:instance_ei}
        \Lambda_{i,s}(X)
        \equiv
        \frac{1}{2^n}
        \left[
            \mathbb I\Tr(X)
            +s\alpha Q_i\Tr(Q_iX)
        \right]
        =
        \mathcal (\mathcal{E}_i^\dagger\circ\widetilde R_{i,s} \circ\mathcal E_i)(X),
    \end{equation}
    where $\mathcal E_i$ is the doubly-stochastic channel with the action:
    \begin{equation} \label{eq:instance_action_ei}
        \mathcal E_i(\mathbb I)=\mathbb I,\qquad
        \mathcal E_i(Q_i)=\sqrt{\alpha}\,Q_i,\qquad
        \mathcal E_i(Q)=0 ~~\text{for all $Q\in\mathcal P_n\setminus\{\mathbb I,Q_i\}$},
    \end{equation}
    and $R_{i,s}$ is any Pauli operator that commutes with $Q_i$ when $s=1$ and anticommutes with $Q_i$ when $s=-1$. Its Choi state is:
    \begin{equation} \label{eq:choi_diag}
        \rho_{\Lambda_{i,s}}
        =
        \frac{1}{4^n}
        \left(
            \mathbb I^{\otimes 2n}
            + \alpha s \, Q_i\otimes Q_i^T
        \right).
    \end{equation}

    \item Let $Q_i,Q_j\in\mathcal P_n\setminus\{\mathbb I\}$ be
    distinct, $s \in \{-1,1\}$, and $0\leq\alpha\leq1/2$.
    Then the channel that exchanges $Q_i, Q_j$ while annihilating all other Paulis and keeping the identity unchanged:
    \begin{equation} \label{eq:instance_eij}
        \Lambda_{i,j,s}(X)
        \equiv
        \frac{1}{2^n}
        \left[
            \mathbb I\Tr(X)
            +s\alpha
            \left(
                Q_i\Tr(Q_jX)+Q_j\Tr(Q_iX)
            \right)
        \right]
        =
        (\mathcal E_{i,j,s}^*
        \circ\widetilde R
        \circ\mathcal E_{i,j,s})(X),
    \end{equation}
    where $\mathcal{E}_{i,j,s}$ is a doubly-stochastic channel chosen as follows. Choose non-identity Paulis $A,B,R$ satisfying
    \begin{equation}
        [A,B]=0,\qquad [R,A]=0,\qquad \{R,B\}=0,
    \end{equation}
    so that
    \begin{equation} \label{eq:instance_action_eij_r}
        R(A+B)R^\dagger=A-B,\qquad
        R(A-B)R^\dagger=A+B.
    \end{equation}
    For $n\geq2$, one may take
    $A=Z_1$, $B=Z_2$, and $R=X_2$. Setting
    $a=\sqrt{\alpha/2}$, define
    \[
    \begin{aligned}
        \mathcal E_{i,j,s}(\mathbb I)&=\mathbb I,\\
        \mathcal E_{i,j,s}(Q_i)&=a(A+B),\\
        \mathcal E_{i,j,s}(Q_j)&=sa(A-B),\\
        \mathcal E_{i,j,s}(Q)&=0 ~~\text{for all $Q\in\mathcal P_n\setminus\{\mathbb I,Q_i,Q_j\}$}.
    \end{aligned}
    \]
    Its Choi state is:
    \begin{equation} \label{eq:choi_hard_instances_gen}
        \rho_{\Lambda_{i,j,s}}
        =
        \frac{1}{4^n}
        \left(
            \mathbb I^{\otimes 2n}
            + \alpha s \, (Q_i\otimes Q_j^T + Q_j\otimes Q_i^T)
        \right).
    \end{equation}
\end{enumerate}
\end{lemma}

\begin{proof}
The claim for $\Lambda_0$ is straightforward. For the remaining claims, first observe that the channels $\mathcal{E}_i$ and $\mathcal{E}_{i,j,s}$ have the Choi states:
\begin{equation}
    \rho_{\mathcal E_i}
    =
    \frac{1}{4^n}
    \left(
        \mathbb I^{\otimes 2n}
        +\sqrt{\alpha}\,Q_i\otimes Q_i^T
    \right), \qquad \rho_{\mathcal E_{i,j,s}}
    =
    \frac{1}{4^n}
    \left(
        \mathbb I^{\otimes 2n} + a((A+B)\otimes Q_i^T
    +s(A-B)\otimes Q_j^T)
    \right),
\end{equation}
which both have maximally mixed left and right marginals. $\rho_{\mathcal E_i}$ is positive semidefinite for $\alpha \leq 1$, since $\norm{\sqrt{\alpha}\,Q_i\otimes Q_i^T}_{\infty} = \sqrt{\alpha}$. For $\rho_{\mathcal{E}_{i,j,s}}$, it can be checked that $\left((A+B)\otimes Q_i^T +s(A-B)\otimes Q_j^T \right)^2 = 4 \mathbb{I}^{\otimes 2n}$, and therefore $\norm{(A+B)\otimes Q_i^T +s(A-B)\otimes Q_j^T}_{\infty} = 2$, implying that $\rho_{\mathcal{E}_{i,j,s}}$ is positive semidefinite for $a \leq 1/2$ (equivalently $\alpha \leq 1/2$). Therefore, $\mathcal{E}_i$ and $\mathcal{E}_{i,j,s}$ are doubly-stochastic quantum channels.

Next, it suffices to check that for any input Pauli operator $X$, the right equalities of Eqs.~(\ref{eq:instance_ei}) and (\ref{eq:instance_eij}) hold. The case of $\Lambda_{i,s}$ can be verified straightforwardly via Eqs.~(\ref{eq:instance_action_ei}). The case of $\Lambda_{i,j,s}$ can be verified by noting that $\mathcal E_{i,j,s}^*(A+B)=2aQ_i$ and $\mathcal E_{i,j,s}^*(A-B)=2saQ_j$, and applying Eqs.~(\ref{eq:instance_action_eij_r}).

Finally, the form of the Choi states $\rho_{\Lambda_{i,s}}$ and $\rho_{\Lambda_{i,j,s}}$ can be obtained by e.g. applying the channels to the Bell state via Eq.~(\ref{eq:bell_state}).

\end{proof}

\subsection{Models of learning protocols} \label{app:models}

Here, we describe the models of OTOC learning protocols considered in our work. We begin by describing the closely related models of protocols for learning quantum channels \cite{chen2021exponential,chen2022quantum,caro2024learning,chen2024tight,chen2025efficient}. 

In such protocols, the user is given samples of a queryable $n$-qubit quantum channel $\mathcal{N}$, and access to a quantum computer with at least $n$ qubits. A generic coherent execution/run (each labeled with $i=1,...,N$) involves initializing the quantum computer as a quantum state $\rho_i$, querying $\mathcal{N}$, and measuring all qubits of the quantum computer (using a POVM $M_i = \{M^b_i\}_b$) to obtain the bitstring $b_i$ (with probability $\tr(M_b^i (\mathcal{N} \otimes \mathbb{I}^{\otimes n_{\text{anc}}})(\rho_i))$). After $N$ repetitions, the acquired classical data $\{b_1,...,b_N\}$ are then processed on a classical computer; we (generously) assume unlimited classical processing resources when deriving our lower bounds, while our algorithms/upper bounds require only classical processing resources that scale polynomially with problem size.

Additionally, the learning protocol may possess the following additional features\footnote{We follow terminologies of \cite{chen2021exponential,huang2021quantum,caro2024learning,chen2024tight}, etc. Note that \cite{chen2025efficient} instead refers to what we call `quantum memory' \textit{concatenation}, and what we call `having access to ancillas' \textit{quantum memory}.}:
\begin{enumerate}

    \item The protocol is \textit{adaptive} if the choice of initial state and measurement (and intermediate processing channels, if present) of later runs can depend on those of previous runs.

    \item The protocol has access to ($n_{\text{anc}}$ qubits of) ancillas if each coherent run operates on a total $n + n_{\text{anc}}$ qubits. Protocols without ancillas can only operate on $n$ qubits at each coherent run.

    \item The protocol has \textit{quantum memory} if each coherent run involves multiple queries of $\mathcal{N}$ before measurement, together with arbitrary interleaved channels. Protocols without quantum memory only query $\mathcal{N}$ once in each coherent run, and is always preceded by initialization and followed by measurement.

\end{enumerate}

Next, differing to the models above for channel learning, we restrict the set of queryable channels to those that take the form $\mathcal{O}_{\mathcal{E}, O} = \mathcal{E}^\dagger \circ \tilde{O} \circ \mathcal{E}$, as described in Appendix~\ref{app:prelim_vect_choi} above. These models capture the vast majority of existing protocols for learning OTOCs, such as \cite{swingle2016measuring,mi2021information,schuster2022many,schuster2023operator,schuster2023learning,cotler2023information,google2025observation,algorithmiq} (without ancillas) and \cite{yoshida2019disentangling,landsman2019verified,schuster2022many,sundar2022proposal,green2022experimental,xu2024scrambling} (with ancillas). Schematics of the main classes of protocols considered in this work are illustrated in Fig.~\ref{fig:protocols}:
\begin{enumerate}[label=(\alph*)]
    \item displays the class of protocols with no ancillas, but with quantum memory;

    \item displays a protocol with access to $n$ ancillas, but no quantum memory;

    \item displays the class of protocols with access to arbitrary number of ancillas, but no quantum memory;

    \item displays a protocol with access to $n$ ancillas and quantum memory.
\end{enumerate}

Finally, for a minority of results, we distinguish between protocols that can initialize as arbitrary states, versus those that can only initialize as product states. We also distinguish between protocols having \textit{channel access}, where the user is given samples of queryable channels $\mathcal{O}_{\mathcal{E}, O}$ as in the above cases, versus those having only \textit{Choi access} \cite{caro2024learning}, where the user is only given samples of $\rho_{\mathcal{O}_{\mathcal{E}, O}}$, the Choi state of $\mathcal{O}_{\mathcal{E}, O}$. In models that have Choi access, quantum memory thus translates to the ability to entangle multiple copies of $\rho_{\mathcal{O}_{\mathcal{E}, O}}$ in each coherent run, before collective measurement. Protocols with channel access are more general than those having Choi access (lower bounds on the former imply lower bounds on the latter), and are our main subject of study.

\subsection{Classical shadows and Pauli tomography}

Here, we collect several known results for learning properties of unknown quantum states and channels, to be invoked later.

\begin{theorem}[Pauli classical shadow, \cite{huang2020predicting}] \label{theorem:huang_predicting_pauli}
Given $N$ copies of the quantum state $\rho$, consider the task of simultaneously estimating $M$ (a priori unknown) expectation values $\{ \tr(\rho Q_i)\}_{i=1}^M$ each to additive error $\epsilon$, with success probability $1-\delta$, where $Q_i$ are tensor product operators of weight $w$. There is an algorithm that achieves this using $N = \bigo{3^w \log(M/\delta)/\epsilon^2}$ copies of $\rho$.
\end{theorem}

\begin{theorem}[Clifford classical shadow, \cite{huang2020predicting}] \label{theorem:huang_predicting_clifford}
Given $N$ copies of the quantum state $\rho$, consider the task of simultaneously estimating $M$ (a priori unknown) Pauli expectation values $\{ \tr(\rho O_i)\}_{i=1}^M$ each to additive error $\epsilon$, with success probability $1-\delta$, where $O_i$ has Hilbert-Schmidt norm $\tr(O_i^2)$. There is an algorithm that achieves this using $N = \bigo{\max_i(\tr(O_i^2)) \log(M/\delta)/\epsilon^2}$ copies of $\rho$.

In particular, if the observables are $O_i = \ketbra{\psi_i}{\psi_i}$ so that $\tr(\rho O_i)$ correspond to the fidelity with target pure states $\ket{\psi_i}$, $\tr(O_i^2) = 1$ and $N = \bigo{\log(M/\delta)/\epsilon^2}$.
\end{theorem}

\begin{theorem}[Pauli tomography with quantum memory, \cite{huang2021information}, Theorem 2] \label{result:pauli_tomo}
Given $N$ copies of the quantum state $\rho$, consider the task of simultaneously estimating $M$ (a priori unknown) Pauli expectation values $\{ \tr(\rho P_i)\}_{i=1}^M$ each to additive error $\epsilon$, with success probability $1-\delta$. There is an algorithm requiring quantum memory that achieves this using $N=\bigo{\log(M/\delta)/\epsilon^4}$ copies of $\rho$.
\end{theorem}

When the algorithm of Theorem~\ref{result:pauli_tomo} is applied with the Choi state of the channel $\mathcal{N}$ as inputs, we immediately obtain:
\begin{theorem}[Pauli Transfer Matrix Learning, \cite{caro2024learning}, Theorem 4.1]
Given $N$ copies of the queryable channel $\mathcal{N}$, consider the task of simultaneously estimating $M$ (a priori unknown) Pauli Transfer Matrix elements $\{ \tr(P_i \mathcal{N}(Q_i))\}_{i=1}^M$ each to additive error $\epsilon$, with success probability $1-\delta$. There is a an algorithm using quantum memory that achieves this using $N=\bigo{\log(M/\delta)/\epsilon^4}$ runs.
\end{theorem}

\section{Pauli operator shadow for learning arbitrary weight-$w$ OTOCs} \label{app:pauli_shadow}

Here, we prove Theorem~\ref{theorem:simul_otoc_shadows_nq} on $n$- and $2n$-qubit algorithms for simultaneously estimating OTOCs (for both the general and diagonal cases, as Theorems~\ref{theorem:general_op_shad} and \ref{theorem:general_op_shad_diag} respectively.). They are direct applications of the approach introduced in \cite{chiew2026quantum} based on vectorizing $O(t)$, and regarding OTOCs as the expectation value of observables over $\kett{O(t)}$, which can subsequently be estimated via the classical shadow formalism. We phrase our results in full generality here, by considering Heisenberg evolution under general unital channels $\mathcal{E}^\dagger$ and assuming the ability to query $\mathcal{E}$ and its adjoint $\mathcal{E}^\dagger$. The (weaker) model that only has access to copies of the state $\kett{O(t)}$ or $\rho_{{\mathcal{O}_{\mathcal{E}, O}}}$ is also discussed later in this section.

\subsection{Pauli operator shadow for general local OTOCs} \label{app:pauli_shad_general}
The procedures for obtaining the classical shadow of Heisenberg operators, as briefly described in the main text, are summarized as Algorithms~\ref{algo:pauli_shad_2n} and \ref{algo:pauli_shad_n} below.

\vspace{5pt}
\noindent
\begin{minipage}[c]{0.48\textwidth}
\begin{algorithm}[H] \label{algo:pauli_shad_2n}
\linespread{1.25}\selectfont
\caption{Pauli operator shadow, $2n$ qubits}

\KwData{$N$ copies of the quantum channel ${\mathcal{O}_{\mathcal{E}, O}} = \mathcal{E}^\dagger \circ \tilde{O} \circ \mathcal{E}$}
\KwResult{$N$ classical snapshots $\{\hat{\rho}^{(1)}_{{\mathcal{O}_{\mathcal{E}, O}}}, ..., \hat{\rho}^{(N)}_{{\mathcal{O}_{\mathcal{E}, O}}} \}$.}

For $i \in [1,\ldots,N]$:\\
\Indp
Prepare $\rho_{{\mathcal{O}_{\mathcal{E}, O}}} = ({\mathcal{O}_{\mathcal{E}, O}} \otimes \mathbb{I})(\kettbbra{\mathbb{I}}{\mathbb{I}})$\;
Independently sample $V^{(i)}_1,...,V^{(i)}_{2n} \sim \text{Cl}(2)$ and apply $V^{(i)} \equiv \bigotimes_{j=1}^{2n} V^{(i)}_j$\;
Measure in computational basis to obtain outcome $\ket{b^{(i)}}$\;
Record $\hat{\rho}^{(i)}_{{\mathcal{O}_{\mathcal{E}, O}}} = M^{-1}(V^{(i)\dagger} \ketbra{b^{(i)}}{b^{(i)}} V^{(i)})$\;
\Indm
\textbf{Return} $\{\hat{\rho}^{(1)}_{{\mathcal{O}_{\mathcal{E}, O}}}, ..., \hat{\rho}^{(N)}_{{\mathcal{O}_{\mathcal{E}, O}}} \}$\;
\end{algorithm}
\end{minipage}
\hfill
\begin{minipage}[c]{0.48\textwidth}
\begin{algorithm}[H] \label{algo:pauli_shad_n}
\linespread{1.15}\selectfont
\caption{Pauli operator shadow, $n$ qubits}

\KwData{$N$ copies of the quantum channel ${\mathcal{O}_{\mathcal{E}, O}} = \mathcal{E}^\dagger \circ \tilde{O} \circ \mathcal{E}$}
\KwResult{$N$ classical snapshots $\{\hat{\rho}^{(1)}_{{\mathcal{O}_{\mathcal{E}, O}}}, ..., \hat{\rho}^{(N)}_{{\mathcal{O}_{\mathcal{E}, O}}} \}$.}

For $i \in [1,\ldots,N]$:\\
\Indp

Prepare random basis state $| b^{(i)}_L \rangle$\;
Independently sample $V^{(i)}_1,...,V^{(i)}_{n} \sim \text{Cl}(2)$ and apply $(V_L^{(i)})^T \equiv \bigotimes_{j=1}^{n} (V^{(i)}_j)^T$\;
Apply ${\mathcal{O}_{\mathcal{E}, O}}$\;
Independently sample $V^{(i)}_{n+1},...,V^{(i)}_{2n} \sim \text{Cl}(2)$ and apply $V_R^{(i)} \equiv \bigotimes_{j=n+1}^{2n} V^{(i)}_j$\;
Measure in computational basis to obtain outcome $| b^{(i)}_R \rangle$\;
Record $\hat{\rho}^{(i)}_{{\mathcal{O}_{\mathcal{E}, O}}} = M^{-1}((V_L^{(i)} \otimes V_R^{(i)})^\dagger \ketbraa{b_L^{(i)}}{b_L^{(i)}} \otimes \ketbraa{b_R^{(i)}}{b_R^{(i)}} (V_L^{(i)} \otimes V_R^{(i)}))$\;
\Indm
\textbf{Return} $\{\hat{\rho}^{(1)}_{{\mathcal{O}_{\mathcal{E}, O}}}, ..., \hat{\rho}^{(N)}_{{\mathcal{O}_{\mathcal{E}, O}}} \}$\;
\end{algorithm}
\end{minipage}
\vspace{5pt}

The snapshots $\{\hat{\rho}^{(1)}_{{\mathcal{O}_{\mathcal{E}, O}}}, ..., \hat{\rho}^{(N)}_{{\mathcal{O}_{\mathcal{E}, O}}} \}$ can subsequently be used to construct estimators for OTOCs, yielding:

\begin{theorem}[General OTOCs] \label{theorem:general_op_shad}
Consider the task of simultaneously estimating $M$ arbitrary OTOCs $\{\tr(\mathcal{E}(Q_i) O \mathcal{E}(Q_j) O)/2^n \}_{i,j}$, where $Q_i, Q_j$ are weight-$w$ tensor product operators, to additive error $\epsilon$ and success probability at least $1-\delta$. Algorithms~\ref{algo:pauli_shad_2n} and \ref{algo:pauli_shad_n} achieves this using $\bigo{9^w \log(M/\delta)/\epsilon^2}$ queries of the channel ${\mathcal{O}_{\mathcal{E}, O}} = \mathcal{E}^\dagger \circ \tilde{O} \circ \mathcal{E}$.
\end{theorem}

\begin{proof}
Algorithm~\ref{algo:pauli_shad_2n} prepares $N$ copies of the Choi state $\rho_{{\mathcal{O}_{\mathcal{E}, O}}}$ by querying ${\mathcal{O}_{\mathcal{E}, O}}$, which yields the snapshots $\{\hat{\rho}^{(1)}_{{\mathcal{O}_{\mathcal{E}, O}}}, ..., \hat{\rho}^{(N)}_{{\mathcal{O}_{\mathcal{E}, O}}} \}$. Eq.~(\ref{eq:ricochet_channel}) implies that the OTOC $\tr(\mathcal{E}(Q_i) O \mathcal{E}(Q_j) O)/2^n$ correspond to the expectation value $\tr(Q_i \otimes Q_j^T ~\rho_{{\mathcal{O}_{\mathcal{E}, O}}})$ on the Choi state $\rho_{{\mathcal{O}_{\mathcal{E}, O}}}$, for the weight-$2w$ tensor product operator $Q_i \otimes Q_j^T$. These two facts, taken together with Theorem~\ref{theorem:huang_predicting_pauli} (from \cite{huang2020predicting}), imply that $N = \bigo{3^{2w} \log(M/\delta)/\epsilon^2}$ snapshots suffices to predict the $M$ quantities $\{\tr(\mathcal{E}(Q_i) O \mathcal{E}(Q_j) O)/2^n\}_{i,j}$ via a median-of-means estimator, proving our claim for the $2n$-qubit approach.

Next, we show that the $n$-qubit Algorithm~\ref{algo:pauli_shad_n} also yields the classical shadow of $\rho_{{\mathcal{O}_{\mathcal{E}, O}}}$. In Algorithm~\ref{algo:pauli_shad_2n}, each experimental run is labeled by the tuple $(V,b)$, consisting of the sampled random unitary $V$ and the final measured outcome $\ket{b}$. The probability of measuring $\ket{b}$, given that $V$ is sampled, is given by $p(b) = \tr( V^\dagger\ketbra{b}{b}V \rho_{{\mathcal{O}_{\mathcal{E}, O}}})$.

In Algorithm~\ref{algo:pauli_shad_n}, each experimental run is instead labeled by $((V_L, V_R), (b_L, b_R))$, where $b_L \in \{0,1\}^n$ and $V_L$ labels the initial state (randomly chosen with probability $1/2^n$) and initial $n$-qubit random unitary respectively, and $b_R \in \{0,1\}^n$ and $V_R$ labels the measurement outcome and final $n$-qubit random unitary respectively. The probability that the tuple $(b_L, b_R)$ is obtained is then:
\begin{align} \label{eq:probability_equiv}
    p(b_L, b_R) &= p(b_R ~\text{measured} ~ | ~ \text{initial state is }~ b_L) *p(\text{initial state is }~ b_L) \\
    &= \tr( V_R^\dagger\ketbra{b_R}{b_R} V_R ~ \mathcal{O}((V_L^\dagger\ketbra{b_L}{b_L}V_L )^T))* (1/2^n) \\
    &= 2^n \tr( V_R^\dagger\ketbra{b_R}{b_R}V_R \otimes V_L^\dagger\ketbra{b_L}{b_L}V_L ~\rho_{{\mathcal{O}_{\mathcal{E}, O}}}) *(1/2^n) \\
    &= \tr( V^\dagger\ketbra{b}{b}V \rho_{{\mathcal{O}_{\mathcal{E}, O}}}),
\end{align}
where in the third equality we used Eq.~(\ref{eq:ricochet_channel}) ($\tr( A \otimes B^T \rho_{\mathcal{N}}) = \tr(A \mathcal{N} (B))/2^n$ for any channel $\mathcal{N}$), and in the final equality we identified $V = V_L \otimes V_R$ and $\ket{b} = \ket{b_L} \otimes \ket{b_R}$. As expected, $p(b_L, b_R)$ is identical to the probability distribution $p(b)$ observed from the $2n$-qubit approach above, yielding the identical classical snapshot $\hat{\rho}_{{\mathcal{O}_{\mathcal{E}, O}}}$, and completing the proof.
\end{proof}

Note that when $\mathcal{E}^\dagger$ is unitary, we recover the more conventional case where $\rho_{{\mathcal{O}_{\mathcal{E}, O}}} = \kettbbra{O(t)}{O(t)}$ is the pure state corresponding to the vectorization of the Heisenberg operator $O(t)$, and due to the cyclicity of trace, that:
\begin{equation}
    \frac{1}{2^n}\tr(\tr(\mathcal{E}(Q_i) O \mathcal{E}(Q_j) O)) = \frac{1}{2^n} \tr(U Q_i U^\dagger O U Q_j U^\dagger O) = \frac{1}{2^n} \tr(Q_i O(t) Q_j O(t)).
\end{equation}

Alternatively and more precisely, following the proof of Theorem~\ref{theorem:huang_predicting_pauli} in \cite{huang2020predicting}, the above algorithms can be analyzed by regarding them as implementing a quantum channel $\mathcal{M}$ acting on $\rho_{\mathcal{O}_{\mathcal{E}, O}}$, and analyzing the resulting shadow norm. In the $2n$-qubit approach, the action of repeatedly drawing $V_1, ..., V_{2n} \sim \text{Cl}(2)$, applying $V \equiv \bigotimes_{i=1}^{2n} V_i$ to the state $\rho_{{\mathcal{O}_{\mathcal{E}, O}}}$, and measuring it in the computational basis defines the quantum channel:
\begin{equation} \label{eq:general_shadows_channel}
    \mathcal{M}(\rho_{\mathcal{O}_{\mathcal{E}, O}}) = \mathop{\mathbb{E}}_{V \sim \mathcal{U}} \left[ \sum_{b \in \{0,1\}^{2n}} \tr(\rho_{\mathcal{O}_{\mathcal{E}, O}} ~V^\dagger \ketbra{b}{b}V) ~ V^\dagger \ketbra{b}{b}V \right],
\end{equation}
with the inverse channel $\mathcal{M}^{-1}$, and $\mathcal{U}=\text{Cl}(2)^{\otimes 2n}$. Each experimental run labeled by $(V,b)$ then defines a classical snapshot of $\rho_{\mathcal{O}_{\mathcal{E}, O}}$, given by $\hat{\rho}_{\mathcal{O}_{\mathcal{E}, O}} \equiv \mathcal{M}^{-1}(V^\dagger \ketbra{b}{b}V)$, which is an unbiased estimator of $\rho_{\mathcal{O}_{\mathcal{E}, O}}$. Similarly for the $n$-qubit approach, each experimental run labeled by $((V_L, V_R), (b_L, b_R))$ yields the identical classical snapshots $\hat{\rho}_{{\mathcal{O}_{\mathcal{E}, O}}}$ via Eq.~(\ref{eq:general_shadows_channel}).

Subsequently, $\hat{\rho}_{{\mathcal{O}_{\mathcal{E}, O}}}$ can be used to construct unbiased estimators for OTOCs via $\widehat{\text{OTOC}}(Q_i, Q_j) \equiv \tr(Q_i \otimes Q_j^T ~\hat{\rho}_{{\mathcal{O}_{\mathcal{E}, O}}})$. Adopting a medians-of-means approach for estimating the $M$ weight-$w$ OTOCs $\{\tr(Q_i O(t)Q_j O(t))/2^n\}$ yields the upper-bound $N = \bigo{B \log(M/\delta)/\epsilon^2}$, where $B$ is the maximum shadow norm:
\begin{equation}
    B = \max_{(i,j)} \left\Vert Q_i \otimes Q_j^T - \frac{\tr(Q_i \otimes Q_j^T)}{4^n} \mathbb{I} \otimes \mathbb{I} \right\Vert ^2_{\text{sh}},
\end{equation}
with the shadow norm defined as:
\begin{equation} \label{eq:shadow_norm}
    \left\Vert \cdot \right\Vert^2_{\text{sh}} \equiv \max_{\sigma \in \mathcal{H}'} \mathop{\mathbb{E}}_{V \sim \mathcal{U}} \left[ \sum_{b \in \{0,1\}^{2n}} \bra{b} V \sigma V^\dagger \ket{b} \left( \bra{b} V M^{-1}(\cdot) V^\dagger \ket{b} \right)^2 \right].
\end{equation}
Invoking Lemma 3 of \cite{huang2020predicting} which states $\norm{P}^2_{\text{sh}} = 3^w$ for weight $w$ Pauli observables $P$, and applying it to the weight-$2w$ observables $Q_i \otimes Q_j^T$ then proves the upper bounds.

\subsection{Pauli operator shadow for diagonal local OTOCs} \label{app:pauli_shad_diag}

When one only seeks to compute diagonal OTOCs -- i.e. quantities of the form $\{\tr(O(t) P O(t) P)/2^n\}_P$ -- such as when studying local operator geometrical properties, a simple modification of the previous operator shadows protocol yields a $\bigo{3^w}$ scaling in sample complexity, compared to the $\bigo{9^w}$ scaling if Theorem~\ref{theorem:general_op_shad} is naively applied. These algorithms are summarized as Algorithms~\ref{algo:diag_pauli_shad_2n} and \ref{algo:diag_pauli_shad_n} below:

\vspace{5pt}
\noindent
\begin{minipage}[c]{0.48\textwidth}
\begin{algorithm}[H] \label{algo:diag_pauli_shad_2n}
\linespread{1.25}\selectfont
\caption{Correlated Pauli operator shadow, $2n$ qubits}

\KwData{$N$ copies of the quantum channel ${\mathcal{O}_{\mathcal{E}, O}} = \mathcal{E}^\dagger \circ \tilde{O} \circ \mathcal{E}$}
\KwResult{$N$ classical snapshots $\{\hat{\sigma}^{(1)}_{{\mathcal{O}_{\mathcal{E}, O}}}, ..., \hat{\sigma}^{(N)}_{{\mathcal{O}_{\mathcal{E}, O}}} \}$.}

For $i \in [1,\ldots,N]$:\\
\Indp
Prepare $\rho_{{\mathcal{O}_{\mathcal{E}, O}}} = ({\mathcal{O}_{\mathcal{E}, O}} \otimes \mathbb{I})(\kettbbra{\mathbb{I}}{\mathbb{I}})$\;
Independently sample $V^{(i)}_1,...,V^{(i)}_{n} \sim \text{Cl}(2)$ and apply $V^{(i)} = \bigotimes_{j=1}^{n} V^{(i)}_j \otimes V^{(i)}_j$\;
Measure in computational basis to obtain outcome $\ket{b^{(i)}}$\;
Record $\hat{\sigma}^{(i)}_{{\mathcal{O}_{\mathcal{E}, O}}} = M^{-1}((V^{(i)} \otimes V^{(i)}_j)^\dagger \ketbra{b^{(i)}}{b^{(i)}} (V^{(i)} \otimes V^{(i)}_j))$\;
\Indm
\textbf{Return} $\{\hat{\sigma}^{(1)}_{{\mathcal{O}_{\mathcal{E}, O}}}, ..., \hat{\sigma}^{(N)}_{{\mathcal{O}_{\mathcal{E}, O}}} \}$\;
\end{algorithm}
\end{minipage}
\hfill
\begin{minipage}[c]{0.48\textwidth}
\begin{algorithm}[H] \label{algo:diag_pauli_shad_n}
\linespread{1.15}\selectfont
\caption{Correlated Pauli operator shadow, $n$ qubits}

\KwData{$N$ copies of the quantum channel ${\mathcal{O}_{\mathcal{E}, O}} = \mathcal{E}^\dagger \circ \tilde{O} \circ \mathcal{E}$}
\KwResult{$N$ classical snapshots $\{\hat{\sigma}^{(1)}_{{\mathcal{O}_{\mathcal{E}, O}}}, ..., \hat{\sigma}^{(N)}_{{\mathcal{O}_{\mathcal{E}, O}}} \}$.}

For $i \in [1,\ldots,N]$:\\
\Indp

Prepare random basis state $| b^{(i)}_L \rangle$\;
Independently sample $V^{(i)}_1,...,V^{(i)}_{n} \sim \text{Cl}(2)$ and apply $(V^{(i)})^T = \bigotimes_{j=1}^{n} (V^{(i)}_j)^T$\;
Apply ${\mathcal{O}_{\mathcal{E}, O}}$\;
Apply $V^{(i)} = \bigotimes_{j=1}^{n} V^{(i)}_j$\;
Measure in computational basis to obtain outcome $| b^{(i)}_R \rangle$\;
Record $\hat{\sigma}^{(i)}_{{\mathcal{O}_{\mathcal{E}, O}}} = M^{-1}((V^{(i)} \otimes V^{(i)})^\dagger \ketbraa{b_L^{(i)}}{b_L^{(i)}} \otimes \ketbraa{b_R^{(i)}}{b_R^{(i)}} (V^{(i)} \otimes V^{(i)}))$\;
\Indm
\textbf{Return} $\{\hat{\sigma}^{(1)}_{{\mathcal{O}_{\mathcal{E}, O}}}, ..., \hat{\sigma}^{(N)}_{{\mathcal{O}_{\mathcal{E}, O}}} \}$\;
\end{algorithm}
\end{minipage}
\vspace{5pt}

The snapshots $\{\hat{\sigma}^{(1)}_{{\mathcal{O}_{\mathcal{E}, O}}}, ..., \hat{\sigma}^{(N)}_{{\mathcal{O}_{\mathcal{E}, O}}} \}$ can subsequently be used to construct estimators for diagonal OTOCs, yielding:

\begin{theorem}[Diagonal OTOCs] \label{theorem:general_op_shad_diag}
Consider the task of simultaneously estimating $M$ arbitrary diagonal OTOCs $\{\tr(\mathcal{E}(P) O \mathcal{E}(P) O)/2^n\}_{i,j}$, where $P$ are weight-$w$ tensor product operators, to additive error $\epsilon$ and success probability at least $1-\delta$. Algorithms~\ref{algo:diag_pauli_shad_2n} and \ref{algo:diag_pauli_shad_n} achieves this using $\bigo{3^w \log(M/\delta)/\epsilon^2}$ queries of the channel ${\mathcal{O}_{\mathcal{E}, O}} = \mathcal{E}^\dagger \circ \tilde{O} \circ \mathcal{E}$.
\end{theorem}

Before formally proving Theorem 1, let us start by explaining the intuition behind it and set up a few helpful lemmas for the proof. \\

The $\bigo{3^w}$ scaling can be intuitively understood as follows. Let $P \in \mathcal{P}_n$ be weight-$w$ Pauli operators, with corresponding diagonal OTOCs $\tr(O(t) P O(t) P)/2^n = \bbra{O(t)} P \otimes P^T \kett{O(t)}$, corresponding to expectation values of $2n$-qubit Pauli operators of the form $P \otimes P^T$. The strategy in Appendix~\ref{app:pauli_shad_general} of performing random single-qubit Pauli measurements to the state $\kett{O(t)}$ leads to a $\bigo{(1/3)^{2w}}$ probability that a particular realization of measurement setting (in one of the eigenbases of $\{X,Y,Z\}$) contributes to the computation of the weight-$2w$ Pauli operator $P \otimes P^T$. Instead, we can exploit the fact that only operators of the form $P \otimes P^T$ are needed, by performing random measurements in one of the three Pauli product bases $\{XX, YY, ZZ\}$ on qubits in $\mathcal{H}^i_L \otimes \mathcal{H}^i_R$, so that the probability is $\bigo{(1/3)^{w}}$ (instead of measuring in $\{XX, YY, ZZ, XY, XZ, YX, YZ, ZX, ZY\}$). Since this procedure still only involves local (but classically correlated) measurements, the above $2n$-qubit algorithm can also be unfolded to a $n$-qubit version, by (classically) correlating the unitary applied after initialization and before measurements. The $2n$-qubit approach was also observed by \cite{Pineda_Jim_nez_2026} (without analysis). \\

In the following, we will derive the above scaling by upper-bounding the relevant shadow norm, resulting in Theorem~\ref{theorem:general_op_shad_diag}. Later in Theorem~\ref{theorem:shadow_diag_lowerbound_formal}, we will also see that this scaling is unavoidable -- the above strategy for computing diagonal OTOCs is optimal among the class of non-adaptive protocols without quantum memory (i.e. single-copy protocols), capable of local measurements and arbitrary initializations.

Similar to the general case of the previous section, consider the Choi state of ${\mathcal{O}_{\mathcal{E}, O}}$, $\rho_{{\mathcal{O}_{\mathcal{E}, O}}} \equiv ({\mathcal{O}_{\mathcal{E}, O}} \otimes \mathbb{I})(\kettbbra{\mathbb{I}}{\mathbb{I}})$. In the $2n$-qubit approach, the action of repeatedly drawing $n$ random single-qubit Cliffords $V_1, ..., V_{n} \sim \text{Cl}(2)$, applying $V \equiv \bigotimes_{i=1}^{n} (V_i \otimes V_i)$ to the state $\rho_{\mathcal{O}_{\mathcal{E}, O}}$, and measuring it in the computational basis implements the quantum channel:
\begin{equation} \label{eq:shadow_channel_diagonal}
    \mathcal{M}_d(\rho_{\mathcal{O}_{\mathcal{E}, O}}) = \mathop{\mathbb{E}}_{V \sim \mathcal{U}} \left[ \sum_{b \in \{0,1\}^{2n}} \tr(\rho_{\mathcal{O}_{\mathcal{E}, O}} ~(V \otimes V)^\dagger \ketbra{b}{b} (V \otimes V)) ~ (V \otimes V)^\dagger \ketbra{b}{b} (V \otimes V) \right],
\end{equation}
where $\mathcal{U}=\text{Cl}(2)^{\otimes n}$. Each experimental run labeled by $(V, b)$ then defines a classical snapshot $\hat{\sigma}_{\mathcal{O}_{\mathcal{E}, O}} \equiv \mathcal{M}_d^{-1}(V^\dagger \ketbra{b}{b}V)$. 

Note that $\mathcal{M}_d$ is not tomographically complete, and therefore not invertible in general. The map $\mathcal{M}_d^{-1}$ denotes instead the inverse of $\mathcal{M}_d$ when its kernel is removed, so that the domain of $\mathcal{M}_d^{-1}$ is restricted to the image of $\mathcal{M}_d$. Since $V^\dagger \ketbra{b}{b}V$ always lies in the domain of $\mathcal{M}_d^{-1}$ (due to Lemma~\ref{lemma:m_inverse} below), $\hat{\sigma}_{\mathcal{O}_{\mathcal{E}, O}}$ is always well defined. The snapshots $\hat{\sigma}_{\mathcal{O}_{\mathcal{E}, O}}$ can only be used to estimate diagonal OTOCs, which is sufficient for a large class of physically relevant tasks, e.g. applications involving operator growth and hydrodynamics \cite{nahum2018operator,von2018operator,rakovszky2018diffusive,khemani2018operator}.

From Eq.~(\ref{eq:shadow_channel_diagonal}), it can be straightforwardly verified that $\mathcal{M}_d$ factorizes as $\mathcal{M}_d = \bigotimes_{i=1}^n \mathcal{M}_{d,1}$. The two-qubit channels $\mathcal{M}_{d,1}, \mathcal{M}_{d,1}^{-1}$ satisfy similar properties as the factorized single-qubit channels $\mathcal{M}_1, \mathcal{M}_1^{-1}$ (from $\mathcal{M} = \bigotimes_{i=1}^n \mathcal{M}_1$) resulting from the usual classical shadow protocol: 
\begin{lemma} \label{lemma:m_inverse}
The two-qubit channels $\mathcal{M}_{d,1}$ and $\mathcal{M}_{d,1}^{-1}$ are self-adjoint and unital. In addition, $\tr(\hat{\sigma}_{\mathcal{O}_{\mathcal{E}, O}}) = 1$, and $\mathcal{M}_{d,1}$ is trace-preserving. Furthermore, $\mathcal{M}_{d,1}$ takes the following form for $P \in \{X,Y,Z\}$:
\begin{align}
    \mathcal{M}_{d,1}(\mathbb{I} \otimes \mathbb{I}) = \mathbb{I} \otimes \mathbb{I}, &\quad\quad \mathcal{M}_{d,1}(P \otimes \mathbb{I}) = (P \otimes \mathbb{I})/3, \\
    \mathcal{M}_{d,1} (\mathbb{I} \otimes P) = (\mathbb{I} \otimes P)/3, &\quad\quad \mathcal{M}_{d,1}(P \otimes P) = (P \otimes P)/3,
\end{align}
while mapping all other Pauli operators to 0. $\mathcal{M}_{d,1}^{-1}$, the inverse of $\mathcal{M}_{d,1}$ in the image of $\mathcal{M}_{d,1}$, has the action:
\begin{align} \label{eq:inverse_diag_channel}
    \mathcal{M}^{-1}_{d,1}(\mathbb{I} \otimes \mathbb{I}) = \mathbb{I} \otimes \mathbb{I}, &\quad\quad \mathcal{M}^{-1}_{d,1}(P \otimes \mathbb{I}) = 3(P \otimes \mathbb{I}), \\
    \mathcal{M}^{-1}_{d,1} (\mathbb{I} \otimes P) = 3(\mathbb{I} \otimes P), &\quad\quad \mathcal{M}^{-1}_{d,1}(P \otimes P) = 3(P \otimes P), \label{eq:inverse_diag_channel2}
\end{align}
and is undefined elsewhere.
\end{lemma}

\begin{proof}
The properties of $\mathcal{M}_{d,1}, \mathcal{M}_{d,1}^{-1}$, and $\tr(\hat{\sigma}_{\mathcal{O}_{\mathcal{E}, O}})$ can be straightforwardly shown via the form of Eq.~(\ref{eq:shadow_channel_diagonal}), fully mirroring the derivations of the same properties for the usual (state) shadow channel \cite{huang2020predicting}. 

To obtain the explicit form of the two-qubit channel $\mathcal{M}_{d,1}$, write, $\forall A,B \in \mathcal{L}(\mathcal{H}_1)$:
\begin{equation} \label{eq:m_ab}
    \mathcal{M}_{d,1}(A \otimes B) = \mathop{\mathbb{E}}_{V \sim \text{Cl}(2)} \left[ \sum_{l,r \in \{0,1\}} \langle l | VAV^\dagger | l \rangle \langle r | VBV^\dagger | r \rangle ~ V^\dagger \ketbra{l}{l} V \otimes V^\dagger \ketbra{r}{r} V \right].
\end{equation}
Note also that the expectation value $\langle l | VAV^\dagger | l \rangle$ takes value $\{-1,1\}$ if $VAV^\dagger \in \{Z, \mathbb{I}\}$ and $0$ if $ VAV^\dagger \in \{X,Y\}$ (and similarly for $\langle r | VBV^\dagger | r \rangle$). Possible cases are then as follows:
\begin{itemize}
    \item If $A=B=\mathbb{I}$: 
    \begin{equation*}
        \mathcal{M}_{d,1}(\mathbb{I} \otimes \mathbb{I}) = \mathop{\mathbb{E}}_{V \sim \text{Cl}(2)} \left[ \sum_{l,r \in \{0,1\}} V^\dagger \ketbra{l}{l} V \otimes V^\dagger \ketbra{r}{r} V \right] = \mathop{\mathbb{E}}_{V \sim \text{Cl}(2)} \left[ V^\dagger V \otimes V^\dagger V \right] = \mathbb{I} \otimes \mathbb{I},
    \end{equation*}
    by completeness $\sum_l \ketbra{l}{l} = \sum_r \ketbra{r}{r}$ and the unitarity of $V$.

    \item If one of $A,B$ is $\mathbb{I}$, and the other is $X,Y,$ or $Z$:
    \begin{equation*}
    \begin{aligned}
    \mathcal{M}_{d,1}(A \otimes \mathbb{I})
    &=
    \mathop{\mathbb{E}}_{V \sim \operatorname{Cl}(2)}
    \left[
        \sum_{l,r \in \{0,1\}}
        \langle l | V A V^\dagger | l \rangle
        \,
        V^\dagger \ketbra{l}{l} V
        \otimes
        V^\dagger \ketbra{r}{r} V
    \right]
    \\
    &=
    \mathop{\mathbb{E}}_{V \sim \operatorname{Cl}(2)}
    \left[
        \sum_{l \in \{0,1\}}
        \langle l | V A V^\dagger | l \rangle
        \,
        \left(V^\dagger \ketbra{l}{l} V\right)
        \otimes \mathbb{I}
    \right].
    \end{aligned}
    \end{equation*}
    Then, noting that conjugating an arbitrary non-identity Pauli with random single-qubit Cliffords $V \sim \text{Cl}(2)$ maps it to one of $\{X,Y,Z\}$, each with probability 1/3, and denoting by $V_z$ the Clifford satisfying $V_z A V_z^\dagger = Z$, we have:
    \begin{align*}
        \mathop{\mathbb{E}}_{V \sim \text{Cl}(2)} \left[ \sum_{l \in \{0,1\}} \langle l | VAV^\dagger | l \rangle ~~ (V^\dagger \ketbra{l}{l} V) \otimes \mathbb{I} \right] &= \frac{1}{3}\left( \sum_{l \in \{0,1\}} \langle l | V_z A V_z^\dagger | l \rangle ~~V^\dagger \ketbra{l}{l}V \right) \otimes \mathbb{I} + 0 + 0 \\
        &=  \frac{1}{3}\left(V_z^\dagger \ketbra{0}{0}V_z - V_z^\dagger \ketbra{1}{1}V_z \right) \otimes \mathbb{I} \\
        &= \frac{1}{3}\left(V_z^\dagger Z V_z\right) \otimes \mathbb{I} \\
        &= \frac{1}{3} A \otimes \mathbb{I}.
    \end{align*}
    Following the same steps as above, we also have $\mathcal{M}_{d,1}(\mathbb{I} \otimes B) = \frac{1}{3} \left( \mathbb{I} \otimes B \right)$.

    \item If $A \neq B$ and neither is $\mathbb{I}$: At least one of $\langle l | VAV^\dagger | l \rangle$ or $\langle r | VBV^\dagger | r \rangle$ must be zero, so $\mathcal{M}_{d,1}(A \otimes B) = 0$.

    \item If $A = B \neq \mathbb{I}$: Again denoting by $V_z$ the Clifford satisfying $V_z A V_z^\dagger = Z$, Eq.~(\ref{eq:m_ab}) evaluates to:
    \begin{align*}
        \mathcal{M}_{d,1}(A \otimes B) &= \mathop{\mathbb{E}}_{V \sim \text{Cl}(2)} \left[ \sum_{l,r \in \{0,1\}} \langle l | VAV^\dagger | l \rangle \langle r | VAV^\dagger | r \rangle ~ V^\dagger \ketbra{l}{l} V \otimes V^\dagger \ketbra{r}{r} V \right] \\
        &= \frac{1}{3} \sum_{l,r \in \{0,1\}} \langle l | V_z A V_z^\dagger | l \rangle \langle r | V_z A V_z^\dagger | r \rangle ~ V_z^\dagger \ketbra{l}{l} V_z \otimes V_z^\dagger \ketbra{r}{r} V_z  \\
        &= \frac{1}{3}  V_z^\dagger \otimes V_z^\dagger \left( \ketbra{0}{0} \otimes \ketbra{0}{0} - \ketbra{0}{0} \otimes \ketbra{1}{1} - \ketbra{1}{1} \otimes \ketbra{0}{0} + \ketbra{1}{1} \otimes \ketbra{1}{1} \right) V_z \otimes V_z \\
        &= \frac{1}{3}  V_z^\dagger \otimes V_z^\dagger \left( Z \otimes Z \right) V_z \otimes V_z \\
        &= \frac{1}{3} A \otimes A.
    \end{align*}
\end{itemize}

Both the domain and image of $\mathcal{M}_{d,1}$ is $\{\mathbb{I} \otimes \mathbb{I}, P \otimes \mathbb{I}, \mathbb{I} \otimes P, P \otimes P\}_{P \in \{X,Y,Z\}}$. Defining $\mathcal{M}_{d,1}^{-1}$ as the inverse of $\mathcal{M}_{d,1}$ restricted to this domain, the form of $\mathcal{M}_{d,1}^{-1}$ can be determined by inspection.
\end{proof}

\begin{proofof}[Theorem~\ref{theorem:general_op_shad_diag}]

$\hat{\sigma}_{{\mathcal{O}_{\mathcal{E}, O}}}$ can be used to construct unbiased estimators for diagonal OTOCs via $\widehat{\text{OTOC}}(P, P) \equiv \tr(P \otimes P^T ~\hat{\sigma}_{{\mathcal{O}_{\mathcal{E}, O}}})$. Using the properties of $\mathcal{M}_{d,1}$ in Lemma~\ref{lemma:m_inverse}, it can then be straightforwardly verified that the variance of this estimator satisfies:
\begin{equation} \label{eq:diag_variance}
    \text{Var}(\widehat{\text{OTOC}}(P, P)) \leq \left\Vert P \otimes P^T - \frac{\tr(P \otimes P^T)}{4^n} \mathbb{I} \otimes \mathbb{I} \right\Vert ^2_{\text{sh}, \text{diag}},
\end{equation}
where the shadow norm is now $\left\Vert X \right\Vert^2_{\text{sh}, \text{diag}} \equiv \max_{\sigma \in \mathcal{H}'} \tr(\sigma L) ~~ \forall X \in \mathcal{L}(\mathcal{H}_{2n})$ with the operator $L_X$ defined as:
\begin{equation} \label{eq:diag_shadow_norm}
    L_X \equiv \mathop{\mathbb{E}}_{V \sim \mathcal{U}} \left[ \sum_{b \in \{0,1\}^{2n}} (V \otimes V)^\dagger \ket{b} \bra{b} (V \otimes V)   \left( \bra{b} (V \otimes V) \mathcal{M}_d^{-1}(X) (V \otimes V)^\dagger \ket{b} \right)^2 \right]
\end{equation}
with $\mathcal{U} = \text{Cl}(2)^{\otimes n}$. 

Subsequently, again adopting a medians-of-means approach for estimating $M$ arbitrary diagonal OTOCs $\{\tr(P_i O(t) P_i O(t))/2^n\}$, we have the upper bound $N=\bigo{B \log(M/\delta)/\epsilon^2}$, with $B$ the right-hand-side of Eq.~(\ref{eq:diag_variance}) maximized over $\{P_i\}_{i=1}^M$. It remains to upper-bound $B$ to prove the scaling of the $2n$-qubit Algorithm~\ref{algo:diag_pauli_shad_2n}. 

If $P = \bigotimes_{j=1}^n P^{(j)}$ is a weight-$w$ Pauli operator, Eqs.~(\ref{eq:inverse_diag_channel}) and (\ref{eq:inverse_diag_channel2}) imply that $\mathcal{M}_d^{-1}(P \otimes P^T) = \bigotimes_{j=1}^n \mathcal{M}_{d,1}^{-1}(P^{(j)} \otimes (P^{(j)})^T) = 3^w (P \otimes P^T)$. This yields:
\begin{align}
   L_{P \otimes P^T} &= 9^w \mathop{\mathbb{E}}_{V \sim \mathcal{U}}\left[ \sum_{b \in \{0,1\}^{2n}}  (V \otimes V)^\dagger \ket{b} \bra{b} (V \otimes V) \left( \bra{b} (V \otimes V) (P \otimes P) (V \otimes V)^\dagger \ket{b} \right)^2 \right] \\
   &= 9^w \mathop{\mathbb{E}}_{V \sim \mathcal{U}} \left[ \Big( \sum_{b' \in \{0,1\}^{n}}  \left(V^\dagger \ket{b'} \bra{b'} V \right) \bra{b'} V P V^\dagger \ket{b'}^2 \Big)^{\otimes 2} \right] \\
   &= 9^w \bigotimes_{j=1}^n \mathop{\mathbb{E}}_{V_j \sim \text{Cl}(2)} \left[ \Big( \sum_{b'' \in \{0,1\}}  \left(V_j^\dagger \ket{b''} \bra{b''} V_j \right) \bra{b''} V_j P^{(j)} V_j^\dagger \ket{b''}^2 \Big)^{\otimes 2} \right], \label{eq:l_expression}
\end{align}
where we dropped the transpose in the first equality because it only contributes a phase that squares to 1, and refactored terms in the others. To evaluate the tensor product factors, observe that:
\begin{itemize}
    \item If $P^{j} = \mathbb{I}$:
    \begin{equation*}
        \mathop{\mathbb{E}}_{V_j \sim \text{Cl}(2)} \left[ \Big( \sum_{b'' \in \{0,1\}}  \left(V_j^\dagger \ket{b''} \bra{b''} V_j \right) \bra{b''} V_j P^{(j)} V_j^\dagger \ket{b''}^2 \Big)^{\otimes 2} \right] = \mathop{\mathbb{E}}_{V_j \sim \text{Cl}(2)} \left[ \Big( \sum_{b'' \in \{0,1\}}  \left(V_j^\dagger \ket{b''} \bra{b''} V_j \right) \Big)^{\otimes 2} \right] = \mathbb{I} \otimes \mathbb{I}.
    \end{equation*}

    \item If $P^{j} \in \{X,Y,Z\}$: Firstly, conjugating an arbitrary non-identity Pauli $P^{(j)}$ with a random single-qubit Cliffords $V_j \sim \text{Cl}(2)$ maps it to one of $\{X,Y,Z\}$, each with probability 1/3.. Secondly, the expectation value $\bra{b''} V_j P^{(j)} V_j^\dagger \ket{b''}^2 = 0$ if $V_j P^{(j)} V_j^\dagger \in \{X,Y\}$, and is 1 if $V_j P^{(j)} V_j^\dagger = Z$. Denoting by $V_{z,j}$ the Clifford satisfying $V_{z,j} P^{(j)} V_{z,j}^\dagger = Z$, this implies that:
    \begin{align*}
        \mathop{\mathbb{E}}_{V_j \sim \text{Cl}(2)} \left[ \Big( \sum_{b'' \in \{0,1\}}  \left(V_j^\dagger \ket{b''} \bra{b''} V_j \right) \bra{b''} V_j P^{(j)} V_j^\dagger \ket{b''}^2 \Big)^{\otimes 2} \right] &= \frac{1}{3} \left( \sum_{b'' \in \{0,1\}}  \left(V_{z,j}^\dagger \ket{b''} \bra{b''} V_{z,j} \right) \right)^{\otimes 2} + 0 + 0 \\
        &= \frac{1}{3}(\mathbb{I} \otimes \mathbb{I}).
    \end{align*}
\end{itemize}
Putting all together, Eq.~(\ref{eq:l_expression}) evaluates to $L_{P \otimes P^T} = 3^w ~\mathbb{I}^{\otimes 2n}$ for any weight-$w$ Pauli operator $P$, implying that $\left\Vert P \otimes P^T \right\Vert^2_{\text{sh}, \text{diag}} \equiv \max_{\sigma \in \mathcal{H}'} \tr(\sigma L) = 3^w \max_{\sigma \in \mathcal{H}'} \tr(\sigma) = 3^w$, and therefore the bound $N=\bigo{3^w \log(M/\delta)/\epsilon^2}$. This completes the proof for the $2n$-qubit algorithm.

An identical scaling can be proven for the $n$-qubit Algorithm~\ref{algo:diag_pauli_shad_n}, following the same steps as the previous section. That is, in the expression Eq.~(\ref{eq:probability_equiv}), it suffices to impose $V_L = V_R$, which reproduces the channel $\mathcal{M}$ (Eq.~(\ref{eq:shadow_channel_diagonal})) of the above procedure, yielding the snapshots $\hat{\sigma}_{{\mathcal{O}_{\mathcal{E}, O}}}$. This completes the proof for the $n$-qubit algorithm.

\end{proofof}

\section{Tight lower bounds for learning arbitrary weight-$w$ OTOCs} \label{app:otoc_lower_shadow}

In this appendix, we will prove the formal version of Theorem~\ref{theorem:shadow_lowerbound} of the main text (as Theorems~\ref{theorem:shadow_diag_lowerbound_formal}, \ref{theorem:shadow_gen_lowerbound_formal}). We begin by outlining the proof technique, largely based on the proofs of Theorem~6 of \cite{huang2020predicting} and Theorem~4 of \cite{chen2022quantum}.

\subsection{Proof strategy} \label{app:shadow_proof_strat}
Consider a three-party communication protocol. Alice wishes to communicate an integer $X$ in $\{1,...,2M\}$ to Bob; they achieve this by encoding them as quantum channels expressible in the form $\mathcal{E}^\dagger \circ \tilde{O} \circ \mathcal{E}$. Alice then chooses one of the $2M$ channels randomly, and sends $N$ copies of it to Bob.  

Alice's instances are chosen to take the form: 
\begin{equation} \label{eq:off_diag_instances}
    \Lambda_{a,b,s} (\cdot) \equiv \frac{1}{2^n}(\mathbb{I} \tr(\cdot) + s \epsilon_0 \left( P_a \tr(P_b ~\cdot) + P_b \tr(P_a ~\cdot) \right),
\end{equation}
where $ P_a, P_b \in {\mathcal{P}}_n \backslash \mathbb{I}^{\otimes n}$, $s\in\{-1,1\}$, and $\epsilon_0 \geq 0$, resulting in the $2M$ channels $\{\Lambda_{a,b,s=1}\} \cup \{\Lambda_{a,b,s=-1}\}$. By Lemma~\ref{lemma:hard_instances_echo_form}, $\Lambda_{a,b,s}$ can be written in the form $\mathcal{E}^\dagger \circ \tilde{O} \circ \mathcal{E}$, and therefore accepted as instances of any OTOC learning protocol satisfying the assumptions of the theorem. To eventually obtain a strong lower bound, we choose these instances to be as hard to distinguish as possible; the exact choice of the $2M$ instances will be specified later, depending on the statement we wish to prove.

Ideally, Bob uses the protocol for predicting OTOCs that satisfies the assumptions of the theorem to identify Alice's chosen integer. Bob can achieve this with high success probability by consulting the output of the protocol, which yield the OTOCs:
\begin{equation} \label{eq:otoc_to_distinguish}
    \frac{1}{2^n} \tr(\mathcal{E}(P_i) O \mathcal{E}(P_j) O) = \frac{1}{2^n} \tr(P_i \Lambda_{a,b,s}(P_j)) = s \epsilon_0 (\delta_{i,a} \delta_{j,b} + \delta_{i,b} \delta_{j,a}),
\end{equation}
which takes value $s \epsilon_0$ if $(i,j) \in \{(a,b), (b,a)\}$ and $a \neq b$, $2s \epsilon_0$ if $i=j=a=b$, and $0$ otherwise. Setting $\epsilon \leq \epsilon_0/2$, it suffices for Bob to output the integer corresponding to the index $(i,j)$ with $\frac{1}{2^n} \abs{\tr(\mathcal{E}(P_i) O \mathcal{E}(P_j) O)} \geq \epsilon$, and with the sign of $\tr(\mathcal{E}(P_i) O \mathcal{E}(P_j) O)/2^n$ determining $s$.

However, to further reflect the fact that Bob has no knowledge on the identities of the $2M$ Pauli operators $\{(P_a,P_b)\}_{a,b}$ (which forces Bob to utilize strategies that are agnostic of the set of $2M$ Pauli operators), a third party (Loki) is introduced, who intercepts the $N$ copies of $\Lambda_{a,b,s}$ sent by Alice. Loki scrambles the Pauli operators $P_a \rightarrow Q_{a} \equiv U P_a U^\dagger$, $P_b \rightarrow Q_{b} \equiv U P_b U^\dagger$ without changing their weights, by sampling a random $n$-qubit product unitary $U = \bigotimes_{\alpha=1}^n U^{(\alpha)}$, where $U^{(\alpha)} \in \mathcal{U}$ are chosen from a suitable unitary ensemble (specified later to either be random single-qubit unitaries or Clifford gates) and sandwiching $ \Lambda_{a,b,s}$ between $U$, yielding:
\begin{equation} \label{eq:loki_scrambling_channel}
    \Lambda_{a,b,s} \longrightarrow ~\mathcal{U} \circ \Lambda_{a,b,s} \circ \mathcal{U} = \frac{1}{2^n}(\mathbb{I} \tr(\cdot) + s \epsilon_0 \left( Q_a \tr(Q_b ~\cdot) + Q_b \tr(Q_a ~\cdot) \right) \equiv {\tilde{\Lambda}}_{a,b,s},
\end{equation}
which is still of the form $\mathcal{E}^\dagger \circ \tilde{O} \circ \mathcal{E}$.

Bob subsequently applies the learning protocol to the $N$ given copies of the quantum channel. Finally, Loki reveals the choice of $U$ applied, allowing Bob to post-process their measurement results to identify the true label of Alice's message despite Loki's interference. 

Summarizing, the above protocol proceeds as follows:
\begin{enumerate}
    \item Alice selects an integer $X$ randomly from $\{1,...,2M\}$, and prepares $N$ copies of the corresponding channel $\Lambda_{a,b,s}$.

    \item Alice attempts to send them to Bob, but the copies are intercepted by Loki, who modifies all $N$ copies to $\Lambda_{\tilde{a},\tilde{b},s}$, with $U = \bigotimes_{\alpha=1}^n U^{(\alpha)}$, with $U^{(\alpha)} \in \mathcal{U}$ sampled from a suitable ensemble.

    \item Bob, upon receiving the channels, applies the OTOC learning protocol. In each of the $N$ runs $j=1,...,N$, they prepare $(n+k)$-qubit initial states $\ket{u_j}$, query the channel, and perform local measurements with the rank-1 POVM:
    \begin{equation}
        F_{j} = \{ w_i d \ketbra{v_{ij}}{v_{ij}} \},  
    \end{equation}
    with $\ket{v_{ij}} \equiv \bigotimes_{\alpha = 1}^{n+k} | v_{ij}^{(\alpha)}\rangle$, $d = 2^{n+k}$, and $\sum_i w_i = 1$.
    
    Bob receives classical outcomes $Y = (Y_1,...,Y_N)$, which are random variables; they observe outcome $i$ in the $j$-th run with probability $p_{ij}$.

    \item Loki reveals the identity of $U$ to Bob. With this knowledge, Bob uses $Y$ to infer Alice's selected integer $X$, identifying the integer to be $\overline{X}$, which succeeds with high probability.
\end{enumerate}
Note that it suffices to consider pure input states and rank-1 POVMs without loss of generality; the former because a mixed input state can always be regarded as part of a larger pure state via purification, and the latter because higher-rank POVMs can always be viewed as rank-1 POVMs by decomposing in their eigenbases.

Given the above protocol, Fano's inequality implies that the mutual information $I(X:\overline{X})$ between Alice's choice $X$ and Bob's inference result $\overline{X}$ is bounded as:
\begin{equation} \label{eq:fano}
    I(X:\overline{X}) \geq \Omega(\log(2M)).
\end{equation}
By the fact that $I(X:U) = 0$ and the data processing inequality $I(X:Y | U) \geq I(X:\overline{X} | U)$, we have:
\begin{equation}
    I(X:Y | U) \geq I(X:\overline{X} | U) \geq I(X:\overline{X}) \geq \Omega(\log(2M)).
\end{equation}
Subsequently, making use of properties of the conditional entropy and the independence of the random variables $Y_1,...,Y_N$, we have:
\begin{align}
    I(X:Y | U) &= H(Y|U) - H(Y|X,U) \\
    & \leq \sum_{j=1}^N (H(Y_j | U) - H(Y_j|X,U)) \\
    & = \sum_{j=1}^N I(X:Y_j | U). \label{eq:mutual_inf_ineq}
\end{align}
Following \cite{huang2020predicting} (S.M. Eq.~(S112)), applications of the concavities of the Shannon entropy $H$ and the logarithm function yield:
\begin{equation} \label{eq:moment_frac}
    I(X:Y_j | U) \leq \mathlarger{\sum}_i ~\frac{\mathop{\mathbb{E}}_{X,U}[p_{ij}^2] - \mathop{\mathbb{E}}_{X,U}[p_{ij}]^2}{\mathop{\mathbb{E}}_{X,U}[p_{ij}]}.
\end{equation}

\medskip

In order to compute the terms that appear in Eq.~\eqref{eq:moment_frac} explicitly, and thereby bound $I(X:Y_j | U)$ it will be helpful to first introduce several definitions and identities for convenience. Due to the one-to-one mapping between $(n+k)$-qubit pure states and Hilbert-Schmidt normalized matrices in $\mathbb{C}^{2^n \times 2^k}$, we can define $A_j, B_{ij} \in \mathbb{C}^{2^n \times 2^k}$ to be the unique `matricizations' of the input and POVM states $\ket{u_j}$ and $\ket{v_{ij}}$ respectively, i.e. the matrices $A_j = \mathrm{vec}^{-1}(\ket{u_j})$ and $B_{ij} = \mathrm{vec}^{-1}(\ket{v_{ij}})$ with $\tr(A_j^\dagger A_j) = \tr(B_{ij}^\dagger B_{ij}) = 1$. Define also:
\begin{equation}
    C_{ij} \equiv B_{ij}A_j^\dagger \in \mathbb{C}^{2^n \times 2^n}.
\end{equation}
Writing $| v^{(1 \sim m)}_{ij} \rangle \equiv \bigotimes_{\alpha = 1}^{m} | v_{ij}^{(\alpha)}\rangle$, the following identities then hold:
\begin{align}
    B_{ij} &= |v_{ij}^{(1 \sim n)} \rangle \langle v^{(n+1 \sim n+k)}_{ij} |,   \\ 
    \tr(C_{ij}^\dagger C_{ij}) &=  \langle v_{ij}^{(n+1 \sim n+k)}| A^\dagger_j A_j | v_{ij}^{(n+1 \sim n+k)} \rangle, \label{eq:cc_vaav} \\
    \tr(P_a C_{ij}^\dagger P_b C_{ij}) &=  \langle v_{ij}^{(1 \sim n)}| P_a | v_{ij}^{(1 \sim n)} \rangle ~ \langle v_{ij}^{(n+1 \sim n+k)}| A^\dagger_j P_b A_j | v_{ij}^{(n+1 \sim n+k)} \rangle \\
    &= \langle P_a \rangle_{v_{ij}^{(1 \sim n)}} \langle A^\dagger_j P_b A_j \rangle_{v_{ij}^{(n+1 \sim n+k)}}. \label{eq:pcpc}
\end{align}
We also have:
\begin{equation} \label{eq:apa_inequality}
    \langle A^\dagger_j P_b A_j \rangle_{v_{ij}^{(n+1 \sim n+k)}}^2 \leq \langle A^\dagger_j A_j \rangle_{v_{ij}^{(n+1 \sim n+k)}}^2 = \tr(C_{ij}^\dagger C_{ij})^2,
\end{equation}
where the first inequality is due to the Cauchy–Schwarz inequality and the second equality is due to Eq.~(\ref{eq:cc_vaav}).

\begin{lemma}\label{lemma_mi_expression}
The mutual information upper-bound Eq.~(\ref{eq:moment_frac}) can be rewritten as:
\begin{equation}
    I(X:Y_j | U) \leq \mathlarger{\sum}_i ~ \frac{2\epsilon_0^2~ w_{ij} 2^k}{\tr(C_{ij}^\dagger C_{ij})} ~ \left( \mathop{\mathbb{E}}_{X,U} \left[ \langle Q_a \rangle_{v_{ij}^{(1 \sim n)}}^2 \langle A^\dagger_j Q_b A_j \rangle^2_{v_{ij}^{(n+1 \sim n+k)}}\right] + \mathop{\mathbb{E}}_{X,U} \left[ \langle Q_b \rangle_{v_{ij}^{(1 \sim n)}}^2 \langle A^\dagger_j Q_a A_j \rangle^2_{v_{ij}^{(n+1 \sim n+k)}}\right] \right).\label{eq:mutual_info_bound}
\end{equation}
\end{lemma}

\begin{proof}
Developing $p_{ij}$ yields:
\begin{align*}
    p_{ij} &= w_{ij} 2^{n+k} \bra{v_{ij}} (\tilde{\Lambda}_{a,b,s} \otimes \mathbb{I}^{\otimes k})(\ketbra{u_j}{u_j}) \ket{v_{ij}} \\
    & = w_{ij} 2^{n+k} \bbra{B_{ij}} (\tilde{\Lambda}_{a,b,s} \otimes \mathbb{I}^{\otimes k})(\kett{A_j}\!\bbra{A_j}) \kett{B_{ij}} \\
    &= w_{ij} \sum_{P \in \mathcal{P}_{k}} \tr(A^\dagger_j A_j P) \tr(B_{ij}^\dagger B_{ij} P) + s \epsilon_0 \left( \tr(A_j^\dagger Q_b A_j P) \tr(B_{ij}^\dagger Q_a B_{ij} P) + \tr(A_j^\dagger Q_a A_j P) \tr(B_{ij}^\dagger Q_b B_{ij} P) \right)  \\
    &= w_{ij} 2^k \left[ \tr(B_{ij}^\dagger B_{ij} A^\dagger_j A_j) + s \epsilon_0 \left( \tr(B_{ij}^\dagger Q_a B_{ij} A^\dagger_j Q_b A_j) + \tr(B_{ij}^\dagger Q_b B_{ij} A^\dagger_j Q_a A_j) \right) \right] \\
    &= w_{ij} 2^k \left[ \tr(C_{ij}^\dagger C_{ij}) + s \epsilon_0 \left( \tr(Q_a C_{ij} Q_b C_{ij}^\dagger) + \tr(Q_b C_{ij} Q_a C_{ij}^\dagger) \right) \right],
\end{align*}
where we made use of the previous definitions in the second and final equalities, $\kett{M}\bbra{M} = {(\tr(M^\dagger M) \tr(\mathbb{I}^{\otimes k}))}^{-1} \sum_{P} MPM^\dagger \otimes P^T$ in the third equality, and $\sum_P P\otimes P = 2^k ~\mathrm{SWAP}$ and $\tr(MN) = \tr((M \otimes N) \mathrm{SWAP})$ in the fourth equality.

\vspace{5pt}

The first moment in Eq.~(\ref{eq:moment_frac}) is given by:
\begin{align}
    \mathop{\mathbb{E}}_{X,U}[p_{ij}] &= w_{ij} 2^k \left( \mathop{\mathbb{E}}_{X,U} \left[\tr(C_{ij}^\dagger C_{ij}) \right] + \epsilon_0 ~ \mathop{\mathbb{E}}_{X,U}\left[ s \tr(Q_a C_{ij} Q_b C_{ij}^\dagger) + s \tr(Q_b C_{ij} Q_a C_{ij}^\dagger) \right]  \right) \\
    &= w_{ij} 2^k \tr(C_{ij}^\dagger C_{ij}), 
\end{align}
where the second term vanishes because $\mathbb{E}[s]=0$, and the first term is independent of $X,U$. Together with $\sum_i \mathop{\mathbb{E}}_{X,U}[p_{ij}] = \mathop{\mathbb{E}}_{X,U} \left[\sum_i p_{ij} \right] = 1$, this also implies that:
 \begin{equation} \label{eq:sum_unity}
     \sum_i w_{ij} 2^k \tr(C_{ij}^\dagger C_{ij}) = 1.
 \end{equation}
Similarly, for the second moment, we have:
\begin{align}
    \mathop{\mathbb{E}}_{X,U}[p_{ij}^2] &= w^2_{ij} 4^k \left( \tr(C_{ij}^\dagger C_{ij})^2 + \epsilon_0^2 ~ \mathop{\mathbb{E}}_{X,U} \left[ \left( \tr(Q_a C_{ij} Q_b C_{ij}^\dagger) + \tr(Q_b C_{ij} Q_a C_{ij}^\dagger) \right)^2 \right] \right).
\end{align}
Putting all together, Eq.~(\ref{eq:moment_frac}) becomes:
\begin{align}
    I(X:Y_j | U) &\leq \mathlarger{\sum}_i ~ \frac{\epsilon_0^2 ~w_{ij} 2^k}{\tr(C_{ij}^\dagger C_{ij})} ~ \mathop{\mathbb{E}}_{X,U} \left[ \left( \tr(Q_a C_{ij} Q_b C_{ij}^\dagger) + \tr(Q_b C_{ij} Q_a C_{ij}^\dagger) \right)^2 \right] \\
    &\leq \mathlarger{\sum}_i ~ \frac{2\epsilon_0^2~w_{ij} 2^k}{\tr(C_{ij}^\dagger C_{ij})} ~ \left( \mathop{\mathbb{E}}_{X,U} \left[ \left( \tr(Q_a C_{ij} Q_b C_{ij}^\dagger) \right)^2\right] + \mathop{\mathbb{E}}_{X,U} \left[ \left( \tr(Q_b C_{ij} Q_a C_{ij}^\dagger) \right)^2\right] \right) \\
    &= \mathlarger{\sum}_i ~ \frac{2\epsilon_0^2~ w_{ij} 2^k}{\tr(C_{ij}^\dagger C_{ij})} ~ \left( \mathop{\mathbb{E}}_{X,U} \left[ \langle Q_a \rangle_{v_{ij}^{(1 \sim n)}}^2 \langle A^\dagger_j Q_b A_j \rangle^2_{v_{ij}^{(n+1 \sim n+k)}}\right] + \mathop{\mathbb{E}}_{X,U} \left[ \langle Q_b \rangle_{v_{ij}^{(1 \sim n)}}^2 \langle A^\dagger_j Q_a A_j \rangle^2_{v_{ij}^{(n+1 \sim n+k)}}\right] \right),
\end{align}
where we used $(x +y)^2 \leq 2(x^2 + y^2)$ in the second inequality, and Eq.~(\ref{eq:pcpc}) in the last equality.
\end{proof}

\vspace{5pt}

Finally, we list several standard identities on expectation values over random unitaries, to be evoked in subsequent proofs (see e.g. \cite{mele2024introduction}). Denoting with $\mathcal{U}_H(2^n)$ the ensemble of Haar-random $n$-qubit unitaries, the following hold:
\begin{align}
    \mathop{\mathbb{E}}_{U \sim \mathcal{U}_H(2)} \left[ \left( U \ketbra{\phi}{\phi} U^\dagger \right) \right] &= \frac{\mathbb{I} \otimes \mathbb{I}}{2} \\
   \mathop{\mathbb{E}}_{U \sim \mathcal{U}_H(2)} \left[ \left( U \ketbra{\phi}{\phi} U^\dagger \right)^{\otimes 2} \right] &= \frac{\mathbb{I} \otimes \mathbb{I} + \mathrm{SWAP}}{6} \\
   \mathop{\mathbb{E}}_{U \sim \mathcal{U}_H(2)} \left[ \left( U P U^\dagger \right)^{\otimes 2} \right] &= \frac{1}{3}(XX + YY + ZZ ) \label{eq:haar_upu}
\end{align}
with $\ket{\phi}$ an arbitrary single-qubit state, and $P \in \{X,Y,Z\}$. Further denoting with $\ket{\psi}$ an arbitrary $n$-qubit product state and $P^{(w)} \in \mathcal{P}_n$ an arbitrary weight-$w$ Pauli operator, the above identities straightforwardly imply:
\begin{align}
    \mathop{\mathbb{E}}_{U \sim \mathcal{U}_H(2)^{\otimes n}} \left[ \langle \psi |U P^{(w)} U^\dagger | \psi \rangle^2 \right] &= \frac{1}{3^w} \label{eq:haar_eq1} \\ 
    \mathop{\mathbb{E}}_{U \sim \mathcal{U}_H(2)^{\otimes n}} \left[ \left( U P^{(w)} U^\dagger \right)^{\otimes 2} \right] &= \frac{1}{3^w} \sum_{Q \in \{X,Y,Z\}^{\otimes w}} Q \otimes Q. \label{eq:haar_eq2}
\end{align}
With some further algebra and using Eq.~(\ref{eq:haar_eq2}), one can also show that for arbitrary $n$-qubit quantum states $\ket{\psi'}$ we have:
\begin{equation}
    \mathop{\mathbb{E}}_{U \sim \mathcal{U}_H(2)^{\otimes n}} \left[ \langle \psi' |U P^{(w)} U^\dagger | \psi' \rangle^2 \right] \leq \left(\frac{2}{3}\right)^w. \label{eq:haar_eq3}
\end{equation}

\subsection{Lower bounds for estimating diagonal OTOCs}

We begin with the (simpler) case where the OTOCs to be estimated are promised to be diagonal, which is the formal version of the diagonal OTOC part of Theorem~\ref{theorem:shadow_lowerbound} of the main text:
\begin{theorem}[Sample complexity lower bound on predicting local diagonal OTOCs; arbitrary initialization] \label{theorem:shadow_diag_lowerbound_formal}
Consider any protocol that is able to predict any set of $M$ weight-$w$ diagonal OTOCs $\{\tr(\mathcal{E}(P_i) O \mathcal{E}(P_i) O)/2^n\}_{i=1}^M$, i.e. $\forall P_i \in \mathcal{P}_n, ~w(P_i)=w$, based on initializing as arbitrary states of $n$ or more qubits, querying the channel ${\mathcal{O}_{\mathcal{E}, O}} = \mathcal{E}^\dagger \circ \tilde{O} \circ \mathcal{E}$, and performing local measurements in each of the $N$ runs. Any such protocol requires at least $N \geq \Omega\left(3^w \log(M)/\epsilon^2\right)$ runs to solve this problem with constant success probability, $\forall w, M$.
\end{theorem}

\begin{proof}
When the OTOCs are promised to be diagonal, the hard instances Eq.~(\ref{eq:off_diag_instances}) take the form:
\begin{equation} \label{eq:hard_instances_diag}
    \Lambda_{a,s} (\cdot) \equiv \frac{1}{2^n}(\mathbb{I} \tr(\cdot) + 2s \epsilon_0 P_a \tr(P_a ~\cdot),
\end{equation}
which are Pauli channels. There are $3^w {n \choose w}$ possible choices of weight-$w$ Pauli operators, and their ordering does not matter. We also take Loki's random unitaries to be $\mathcal{U}=\text{Cl}(2)^{\otimes n}$, which has the effect of permuting the labels of Alice's channels. 

It remains to upper-bound the summand of Eq.~(\ref{eq:mutual_info_bound}) as follows:
\begin{align}
    \mathop{\mathbb{E}}_{X,U} [ \langle P_{\tilde{a}} \rangle_{v_{ij}^{(1 \sim n)}}^2 \langle A^\dagger_j P_{\tilde{a}} A_j \rangle^2_{v_{ij}^{(n+1 \sim n+k)}}] & \leq  \mathop{\mathbb{E}}_{X,U} [ \langle U P_a U^\dagger \rangle_{v_{ij}^{(1 \sim n)}}^2 ] \tr(C_{ij}^\dagger C_{ij})^2 \label{eq:diag_proof1} \\
    &= \frac{1}{3^w} \tr(C_{ij}^\dagger C_{ij})^2, \label{eq:diag_proof2}
\end{align}
where we made use of Eq.~(\ref{eq:apa_inequality}) in the first inequality, and Eq.~(\ref{eq:haar_eq1}) in the second equality together with the fact that $\mathcal{U}=\text{Cl}(2)^{\otimes n}$ form a 3-design. Finally, substituting into Eq.~(\ref{eq:mutual_info_bound}) and making use of Eq.~(\ref{eq:sum_unity}), we have:
\begin{equation*}
    I(X:Y_j | U) \leq \frac{4\epsilon_0^2}{3^w},
\end{equation*}
and therefore via Eqs.~(\ref{eq:fano}) and (\ref{eq:mutual_inf_ineq}) the claimed lower bound $N \geq \Omega\left(3^w \log(M)/\epsilon^2\right)$.
\end{proof}

This lower bound is saturated by the correlated Pauli operator shadow protocols of Algorithms~\ref{algo:diag_pauli_shad_2n} and~\ref{algo:diag_pauli_shad_n}, which estimate $M$ weight-$w$ diagonal OTOCs using $\bigo{3^w\log(M/\delta)/\epsilon^2}$ queries. Hence, for constant $\delta$, Algorithms~\ref{algo:diag_pauli_shad_2n} and~\ref{algo:diag_pauli_shad_n} are optimal in their dependence on $w$, $M$, and $\epsilon$ within the no-quantum-memory, local-measurement model considered here. Notably, the lower bound permits arbitrary, possibly entangled, input states on the system and ancillary qubits, whereas the matching upper bound of Algorithm~\ref{algo:diag_pauli_shad_n} requires only product-state preparation, and is therefore also optimal in the number of qubits needed. 

\subsection{Lower bounds for estimating general OTOCs}
Next, we prove the lower bound for learning general OTOCs, the formal version of Theorem~\ref{theorem:shadow_lowerbound} of the main text:
\begin{theorem}[Sample complexity lower bound on predicting local OTOCs] \label{theorem:shadow_gen_lowerbound_formal}
Consider any protocol that is able to predict any set of $M$ weight-$w$ OTOCs $\{\tr(\mathcal{E}(Q_i) O \mathcal{E}(Q_j) O)/2^n\}_{i,j}$ where $Q_i, Q_j$ are weight-$w$ tensor product operators, i.e. $\forall ~ (i,j), ~w(Q_i)=w(Q_j)=w$, that uses $n$ or more qubits, based on querying the channel ${\mathcal{O}_{\mathcal{E}, O}} = \mathcal{E}^\dagger \circ \tilde{O} \circ \mathcal{E}$, and performing local measurements in each of the $N$ runs.
\begin{enumerate}
    \item Any such protocol that can initialize as arbitrary product states requires at least $N \geq \Omega\left(9^w \log(M)/\epsilon^2\right)$ runs to solve this problem, if $w \leq \lfloor n/2 \rfloor$ and $M \leq \frac{3^{2w}}{2} {n \choose w}{n-w \choose w}$.

    \item Any such protocol that can initialize as arbitrary states requires at least $N \geq \Omega\left((9/2)^w \log(M)/\epsilon^2\right)$ runs to solve this problem, if $w \leq \lfloor n/2 \rfloor$ and $M \leq \frac{3^{2w}}{2} {n \choose w}{n-w \choose w}$.

    \item If $w > \lfloor n/2 \rfloor$ and $M \leq \frac{3^{2w}}{2} {2n \choose 2w}$, we have the weaker $N \geq \Omega\left(3^w \log(M)/\epsilon^2\right)$ bound instead, for any such protocol that can initialize as arbitrary states.
\end{enumerate}
\end{theorem}

\begin{proof} We prove each case in turn. In all cases, we take Loki's random unitaries to be products of independently sampled Haar-random single-qubit unitaries $U \sim \mathcal{U}_H(2)$, which has the effect of transforming the weight-$w$ Pauli operators $P_a, P_b$ of Alice's channels to weight-$w$ tensor product operators $Q_a, Q_b$.

\vspace{5pt}

\noindent \textit{1. Product state initialization, weights satisfy $w \leq \lfloor n/2 \rfloor$ and $M \leq \frac{3^{2w}}{2} {n \choose w}{n-w \choose w}$}: We choose Alice's $2M$ instances Eq.~(\ref{eq:off_diag_instances}) such that $\forall \Lambda_{a,b,s}$, the supports of $P_a$ and $P_b$ are completely disjoint, i.e. $\mathrm{supp}(P_a) \cap \mathrm{supp}(P_b) = \emptyset$. This can only be done if $w \leq \lfloor n/2 \rfloor$, and there are $M \leq \frac{3^{2w}}{2} {n \choose w}{n-w \choose w}$ such choices. The ordering of $a,b$ within these $M$ choices are unimportant.

As a consequence of disjointness, the expectation value in the summand of Eq.~(\ref{eq:mutual_info_bound}) factorize as follows:
\begin{equation} \label{eq:disjoint_factorize}
    \mathop{\mathbb{E}}_{X,U} \left[ \langle Q_a \rangle_{v_{ij}^{(1 \sim n)}}^2 \langle A^\dagger_j Q_b A_j \rangle^2_{v_{ij}^{(n+1 \sim n+k)}}\right] = \mathop{\mathbb{E}}_{X,U} \left[ \langle Q_a \rangle_{v_{ij}^{(1 \sim n)}}^2\right] \mathop{\mathbb{E}}_{X,U} \left[ \langle A^\dagger_j Q_b A_j \rangle^2_{v_{ij}^{(n+1 \sim n+k)}}\right].
\end{equation}
The first factor evaluates to $1/3^w$ using Eq.~(\ref{eq:haar_eq1}). To evaluate the second factor, define the normalized state $\ket{v} \equiv A_j | v_{ij}^{(n+1 \sim n+k)} \rangle/ \norm{A_j | v_{ij}^{(n+1 \sim n+k)} \rangle}$ with $\sqrt{\tr(C_{ij}^\dagger C_{ij})} = \norm{A_j | v_{ij}^{(n+1 \sim n+k)} \rangle}$ (via Eq.~(\ref{eq:cc_vaav})), and note that the initial state $\ket{u_j}$ is a product state and therefore $A_{j} = |u_{j}^{(1 \sim n)} \rangle \langle u^{(n+1 \sim n+k)}_{j} |$, which implies that $\ket{v}$ is also a product state. Subsequently, it can be written as:
\begin{align}
\mathop{\mathbb{E}}_{X,U}
\left[
    \bigl\langle A_j^\dagger Q_b A_j \bigr\rangle_
    {v_{ij}^{(n+1\sim n+k)}}^{2}
\right]
&=
\mathop{\mathbb{E}}_{X,U}
\left[
    \langle v \rvert U P_b U^\dagger \lvert v\rangle
\right]
\,
\norm{
    A_j \lvert v_{ij}^{(n+1\sim n+k)}\rangle
}^{4}
\label{eq:lowerbound_eq1}
\\
&=
\mathop{\mathbb{E}}_{X,U}
\left[
    \langle v \rvert U P_b U^\dagger \lvert v\rangle
\right]
\tr\!\left(C_{ij}^\dagger C_{ij}\right)^{2}
\label{eq:lowerbound_eq2}
\\
&=
\frac{1}{3^{w}}
\tr\!\left(C_{ij}^\dagger C_{ij}\right)^{2}.
\end{align}
where we again used Eq.~(\ref{eq:haar_eq1}) in the final equality for product states.

Finally, substituting into Eq.~(\ref{eq:mutual_info_bound}) and making use of Eq.~(\ref{eq:sum_unity}), we have $I(X:Y_j | U) \leq {4\epsilon_0^2}/{9^w}$ and therefore via Eqs.~(\ref{eq:fano}) and (\ref{eq:mutual_inf_ineq}) the claimed lower bound $N \geq \Omega\left(9^w \log(M)/\epsilon^2\right)$.

\vspace{5pt}

\noindent \textit{2. Arbitrary state initialization, weights satisfy $w \leq \lfloor n/2 \rfloor$ and $M \leq \frac{3^{2w}}{2} {n \choose w}{n-w \choose w}$}: The proof mirrors the above case closely. We choose the same hard instances as the above case with disjoint supports, which again factorize to yield Eq.~(\ref{eq:disjoint_factorize}). The first factor again evaluates to $1/3^w$ using Eq.~(\ref{eq:haar_eq1}), and we similarly have Eqs.~(\ref{eq:lowerbound_eq1}) and (\ref{eq:lowerbound_eq2}). However, because $\ket{v}$ is no longer a product state, $\mathop{\mathbb{E}}_{X,U} \left[ \langle v | U P_b U^\dagger | v  \rangle \right]$ can take values up to $(2/3)^w$ via Eq.~(\ref{eq:haar_eq3}).

Finally, substituting into Eq.~(\ref{eq:mutual_info_bound}) and making use of Eq.~(\ref{eq:sum_unity}), we have $I(X:Y_j | U) \leq {4\epsilon_0^2}{(2/9)^w}$ and therefore via Eqs.~(\ref{eq:fano}) and (\ref{eq:mutual_inf_ineq}) the claimed lower bound $N \geq \Omega\left((9/2)^w \log(M)/\epsilon^2\right)$.

\vspace{5pt}
\noindent \textit{3. Arbitrary state initialization, any weight, and $M \leq \frac{3^{2w}}{2} {2n \choose 2w}$}: In this case, we can no longer make use of the hard instances of the previous cases, and have to make do with the weaker class of instances, where $P_a \neq P_b$ are arbitrary weight-$w$ Pauli operators, of which there are $\frac{3^{2w}}{2} {2n \choose 2w}$ many.

To upper bound the summand of Eq.~(\ref{eq:mutual_info_bound}), it suffices to observe that Eqs.~(\ref{eq:diag_proof1}), (\ref{eq:diag_proof2}) remain true for $U \sim \mathcal{U}_H(2)^{\otimes n}$, which yields the claimed lower bound $N \geq \Omega\left(3^w \log(M)/\epsilon^2\right)$ after following similar steps.

\end{proof}

The first lower bound is saturated by Algorithm~\ref{algo:pauli_shad_n}, the $n$-qubit Pauli operator shadow protocol. Allowing arbitrary, possibly entangled, input states weakens the proven lower bound to $\Omega((9/2)^w\log(M)/\epsilon^2)$, leaving a $2^w$ gap with Algorithm~\ref{algo:pauli_shad_2n}. Whether this gap can be closed is left as an open problem. 

Nonetheless, we find that OTOC learning models with Choi access -- i.e. models that are given access to copies of the state $\rho_{{\mathcal{O}_{\mathcal{E}, O}}}$, instead of the ability to query the channel ${\mathcal{O}_{\mathcal{E}, O}}$ -- have the sample complexity $\Theta(9^w)$, which is saturated by Algorithm~\ref{algo:pauli_shad_2n} (with the preparation of the states $\rho_{\mathcal{O}_{\mathcal{E},O}}$ skipped, since they are assumed to be given):
\begin{theorem}[Sample complexity lower bound on predicting local OTOCs with Choi access] \label{theorem:shadow_gen_lowerbound_formal_choi}
Consider any protocol that is able to predict any set of $M$ weight-$w$ OTOCs $\{\tr(\mathcal{E}(Q_i) O \mathcal{E}(Q_j) O)/2^n\}_{i,j}$ where $Q_i, Q_j \in \mathcal{P}_n$ are weight-$w$ Pauli operators, i.e. $\forall ~ (i,j), ~w(Q_i)=w(Q_j)=w$. The learner is given access to copies of the state $\rho_{\mathcal{O}_{\mathcal{E},O}}$, and the ability to perform local measurements in each of the $N$ independent runs. Any such protocol requires at least $N \geq \Omega\left(9^w \log(M)/\epsilon^2\right)$ runs to solve this problem, if $w \leq \lfloor n/2 \rfloor$ and $M \leq \frac{3^{2w}}{2} {n \choose w}{n-w \choose w}$.
\end{theorem}

\begin{proof}
The proof is a simplified version of the proof strategy detailed in Appendix~\ref{app:shadow_proof_strat}, which describes the more general case where the access model involves queries to the channel $\mathcal{O}_{\mathcal{E},O}$. We can therefore ignore the discussion on the initial states, and work with the Choi states $\rho_{\mathcal{O}_{\mathcal{E},O}}$ directly. The proof proceeds almost verbatim as the proof of Theorem~6 of \cite{huang2020predicting}, except that the hard instances are taken to be the states $\rho_{\Lambda_{i,j,s}}$ defined in Eq.~(\ref{eq:choi_hard_instances_gen}). Subsequently, it suffices to evaluate $\mathop{\mathbb{E}}_{X,U}[p_{ij}]$ and $\mathop{\mathbb{E}}_{X,U}[p^2_{ij}]$ to bound Eq.~(\ref{eq:moment_frac}).

As in Theorem~\ref{theorem:shadow_gen_lowerbound_formal}, choose weight-$w$ Paulis $P_a,P_b$ with disjoint supports and let $Q_a=UP_aU^\dagger$, $Q_b=UP_bU^\dagger$, where $U\sim\mathcal{U}_H(2)^{\otimes n}$. There are $M\leq \frac{3^{2w}}{2}{n\choose w}{n-w\choose w}$ such unordered pairs when $w\leq\lfloor n/2\rfloor$.

In the notation of Lemma~\ref{lemma_mi_expression}, Choi access corresponds to $k=n$ and
$\ket{u_j}=\kett{\mathbb I}$, so $A_j=\frac{\mathbb I}{\sqrt{2^n}}$. Since the POVM states are product states
$\ket{v_{ij}}=\ket{v^L_{ij}}\otimes\ket{v^R_{ij}}$,
Eqs.~(\ref{eq:cc_vaav}) and~(\ref{eq:pcpc}) yield:
\[
\tr(C_{ij}^\dagger C_{ij})=\frac1{2^n},\qquad
\tr(Q_aC_{ij}Q_bC_{ij}^\dagger)
=\frac1{2^n}\langle Q_a\rangle_{v^L_{ij}}
          \langle Q_b\rangle_{v^R_{ij}}.
\]
Disjointness of the supports and Eq.~(\ref{eq:haar_eq1}) therefore imply:
\begin{equation*}
\mathbb E_{X,U}\!\left[
\tr(Q_aC_{ij}Q_bC_{ij}^\dagger)^2
\right]
=
\mathbb E_{X,U}\!\left[
\tr(Q_bC_{ij}Q_aC_{ij}^\dagger)^2
\right]
=\frac{1}{4^n ~9^w}.
\end{equation*}
Using $(x+y)^2\leq2(x^2+y^2)$ and the moment formulas derived in the
proof of Lemma~\ref{lemma_mi_expression}, we obtain:
\[
\mathbb E_{X,U}[p_{ij}]=w_{ij},
\qquad
\mathbb E_{X,U}[p_{ij}^2]
\leq
w_{ij}^2\left(1+\frac{4\epsilon_0^2}{9^w}\right).
\]
Substitution into Eq.~(\ref{eq:moment_frac}), together with
$\sum_iw_{ij}=1$, yields the upper bound:
\begin{equation*}
    I(X:Y_j | U) \leq \frac{4\epsilon_0^2}{9^w},
\end{equation*}
and therefore via Eqs.~(\ref{eq:fano}) and (\ref{eq:mutual_inf_ineq}) the claimed lower bound $N \geq \Omega\left(9^w \log(M)/\epsilon^2\right)$.
\end{proof}

\section{Details on algorithm for two-point correlators and Clifford operator shadows} \label{app:2pc}

The vectorization map Eq.~(\ref{eq:vect_map}) maps wavefunction amplitudes to two-point correlators, and more generally inner products between vectorized operators $\kett{O(t)} = \kett{U O U^\dagger}$, $\kett{O'(t')} = \kett{U' O' U'^\dagger}$ to operator inner products:
\begin{equation}
    c(O ,O') \equiv \frac{1}{2^n} \tr(O^\dagger(t) O'(t')) = ~_{\mathcal{Q}} \langle \! \langle O(t) \kett{O'(t')}_\mathcal{Q}.
\end{equation}
A naive approach for their computation on quantum computers involve applying protocols for evaluating amplitudes and inner products to the vectorized states, such as the Hadamard test, which may require ancillas or queries to controlled unitaries. In the following, we describe an alternative approach which exploits the unitality of Heisenberg time-evolution to reduce the overlap estimation problem to the easier fidelity estimation problem, via the shifted encoding $\kett{\sigma(O(t))}_\mathcal{Q}$. This ultimately allows for their efficient simultaneous estimation, via algorithms such as Clifford classical shadows.

We begin by discussing the computation of operator overlaps $c(O(t), O'(t')) \equiv \tr(O(t) O'(t'))/2^n$ and amplitudes $c_k \equiv \tr(P_k O(t))/2^n$. Subsequently, we discuss the shifted-vectorized encoding $\kett{\sigma(O)}_\mathcal{Q}$ and its use in the computation of two-point correlators via fidelity estimation protocols (such as DFE) and Clifford classical shadows.

For simplicity, $O$ and $O'$ are assumed to be Pauli operators; this can be straightforwardly generalized to sparse linear combinations (c.f. Appendix~B of \cite{chiew2026quantum}). We will make use of different operator encodings, each requiring different number of qubits and accesses to $O(t) = U^\dagger OU$ and $O'(t') = U'^\dagger O' U'$, which we highlight. It is worth remarking that many results obtained in this section exploit the unitality of Heisenberg time evolution (translating to $U^\dagger \otimes U^T \kett{\mathbb{I}}_\mathcal{C} = \kett{\mathbb{I}}_\mathcal{C}$ in vectorized form), distinct from generic Schrodinger-picture situations.

\subsection{The shifted-vectorization encoding $\kett{\sigma(O}$}

We begin by showing that the state $\kett{\sigma(O(t))} \equiv \frac{1}{\sqrt{2}}(\kett{\mathbb{I}} + i\kett{O(t)})$ can be prepared under essentially the same conditions as the unshifted state $\kett{O(t)}_\mathcal{Q}$. This is a consequence of the unitality of Heisenberg time-evolution, which is generally absent from Schrodinger-picture time-evolution (which are only required to be CPTP). 

For simplicity, we work with the more general state parametrized by $\theta \in \mathbb{R}$:
\begin{equation}
    \kett{\sigma(O(t), \theta)}_\mathcal{Q} \equiv \kett{\cos(\theta) \mathbb{I} + i \sin(\theta) O(t)}_\mathcal{Q}
\end{equation}
which reduces to $\kett{\sigma(O(t))}_\mathcal{Q}$ when $\theta = \pi/4$ and $\kett{O(t)}_\mathcal{Q}$ when $\theta = \pi/2$ (up to a global phase of $i$). The following is a direct consequence of Proposition 1 of \cite{chiew2026quantum}, on the computational complexity of preparing the states $\kett{O(t)}_\mathcal{Q}$:

\begin{proposition} \label{theorem:shift_vect_prep}
$\kett{\sigma(O(t))}_\mathcal{Q}$ can be prepared to $\epsilon$ error in trace distance with $L_U + \bigo{\textup{poly(n)} \textup{polylog}(n/\epsilon)}$ gates if:
    \begin{enumerate}
    \item $O(t=0)$ can be expanded as a sum of $\bigo{\textup{poly}(n)}$-many terms in a $k$-producible, orthogonal operator basis $\mathcal{Q}_i$ with $k=\bigo{\log(n)}$, or $R_O(\theta)$ can be implemented using $\bigo{\textup{poly}(n)}$ gates,
    \item $U^\dagger$, and $U^\mathsf{T}$ or $U$ can each be implemented to error $\epsilon/6$ with $L_U$ gates,
    \item and the final basis $\mathcal{Q}$ has a tensor product structure over partitions of at most size $\bigo{\log(n)}$. 
\end{enumerate}
\end{proposition}

\begin{proof}
The following mirrors the proof for the efficient preparation of $\kett{O(t)}_\mathcal{Q}$ of \cite{chiew2026quantum}, with the only replacement being the conjugation by the rotation operator $R_O(\theta) = e^{-i\theta O/2}$ instead of the operator $O(t=0)$.

Making use of the linearity of vectorization, and the fact that the identity term is unchanged by time evolution due to the unitality of Heisenberg evolution, $\kett{\sigma(O(t), \theta)}_\mathcal{Q}$ can be evolved in the same way as $\kett{O(t)}_\mathcal{Q} = M^\mathcal{U}_{\mathcal{Q}} \kett{O}_\mathcal{Q}$:
\begin{align*}
    \kett{\sigma(O(t), \theta)}_\mathcal{Q} &= \kett{\cos(\theta) \mathbb{I} + i \sin(\theta) O(t)}_\mathcal{Q} \\
    &= \kett{\cos(\theta) \mathbb{I} + i \sin(\theta) U^\dagger O U}_\mathcal{Q} \\
    &= M^\mathcal{U}_{\mathcal{Q}} \left( \cos(\theta) \kett{\mathbb{I}}_\mathcal{Q} + i \sin(\theta) \kett{O}_\mathcal{Q} \right) \\
    &= M^\mathcal{U}_{\mathcal{Q}} \kett{\sigma(O, \theta)}_\mathcal{Q}.
\end{align*}
To prepare $\kett{\sigma(O, \theta)}_\mathcal{Q}$, we then make use of the basis change unitary $R_{{\mathcal{P}} \rightarrow \mathcal{Q}}$ to switch to the Pauli basis via $\kett{\sigma(O, \theta)}_\mathcal{Q} = R_{{\mathcal{P}} \rightarrow \mathcal{Q}} \kett{\sigma(O, \theta)}_{\mathcal{P}}$, and make use of the linearity of the vectorization map, yielding:
\begin{align*}
    \kett{\sigma(O, \theta)}_{\mathcal{P}} &= \cos(\theta) \kett{\mathbb{I}}_{\mathcal{P}} + i \sin(\theta) \kett{O}_{\mathcal{P}} \\
    &= \cos(\theta) \ket{0...0} + i \sin(\theta) \kett{O}_{\mathcal{P}}.
\end{align*}
The state $\kett{\sigma(O, \theta)}_P$ is therefore a linear combination of $s+1$ computational basis states if $O$ is a linear combination of $s$ Pauli terms, which can be prepared with $\bigo{ns}$ resources using the general state-preparation methods \cite{gleinig2021efficient,zhang2022quantum,li2024nearly,vilmart2025resource} as described in \cite{chiew2026quantum}. The resources required to implement $M^\mathcal{U}_{\mathcal{Q}}$ and $R_{{\mathcal{P}} \rightarrow \mathcal{Q}}$ (or $R_{\mathcal{Q} \rightarrow {\mathcal{P}}}$) then follow from \cite{chiew2026quantum}.
\end{proof}

When $O^2 = \mathbb{I}$, note that we can further write:
\begin{align*}
    \ket{\sigma(O, \theta)}_\mathcal{C} = \kett{e^{-i O \theta}}_\mathcal{C} &= \left(e^{-i O \theta} \otimes \mathbb{I} \right) \kett{\mathbb{\mathbb{I}}}_\mathcal{C} \\
    &= \left(\mathbb{I} \otimes e^{-i O^T \theta} \right) \kett{\mathbb{\mathbb{I}}}_\mathcal{C},
\end{align*}
allowing preparation by the unitary generated by $O$ to $2n$ Bell pairs, which may be more convenient on certain analog platforms. $\kett{O}_\mathcal{Q}$ is also easier to prepare for certain choices of $O$; for instance, taking it to be the total magnetization in the z-direction $\sum_i Z_i/\sqrt{n}$, it suffices to implement local rotations in the z-axis $(\bigotimes_{i=1}^n e^{i Z_i \theta/\sqrt{n}}) \otimes \mathbb{I}$ in one half of the doubled system, which is easier than preparing the state $\kett{\sum_i Z_i/\sqrt{n}}_P = \sum_i \kett{Z_i}_P/\sqrt{n} = (\ket{10...0} + \ket{001...0} + ... + \ket{00...10})/\sqrt{n}$. We also remark that the probability distribution of $\kett{O}_{\mathcal{P}}$ (sampled in the computational basis) can be recovered from $\kett{\sigma(O)}_{\mathcal{P}}$ by discarding instances of $\kett{\mathbb{I}}_{\mathcal{P}}$, which occur with frequency $\cos^2(\theta)$.

\vspace{5pt}

We further observe that $\ket{\sigma(O, \pi/4)}_{\mathcal{P}}$ and $\ket{\sigma(O, \pi/4)}_\mathcal{C}$ for $O \in {\mathcal{P}}$ is a stabilizer state that can be efficiently prepared with a Clifford unitary decomposable as $\bigo{n}$ single/two-qubit elementary Clifford gates: 
\begin{lemma} \label{lemma:shifted_stab_state}
    $\kett{\sigma(P_k, \pi/4)}_{\mathcal{P}}$ and $\kett{\sigma(P_k, \pi/4)}_\mathcal{C}$, where $P_k \in {\mathcal{P}}$, are stabilizer states, and can be prepared with a unitary with $\bigo{n}$ elementary Clifford gates.
\end{lemma}
\begin{proof}
We show that $\kett{\sigma(P_k, \pi/4)}_{\mathcal{P}} = \kett{\mathbb{I} + i P_k}_\mathcal{Q}/\sqrt{2}$ are stabilizer states, which immediately imply the same for $\kett{\sigma(P_k, \pi/4)}_\mathcal{C}$ since they are related by the Clifford Bell basis transformation $U_{\text{Bell}} = \bigotimes_{i=1}^n (\text{H}_{i_L} \cdot \text{CNOT}_{i_L, i_R})$, where $\text{H}_{i_L}$ is the Hadamard gate acting on $\mathcal{H}_L^i$ while $\text{CNOT}_{i_L, i_R}$ is the CNOT gate with control on the qubit in $\mathcal{H}_L^i$ and target on $\mathcal{H}_R^i$.

Let $\kett{P_k}_{\mathcal{P}} = \ket{b_0...b_{n-1}}$, where $b_i \in \{0,1\}$. The number of $Y$ Paulis present in the Pauli string $P_k$ is then given by $n_Y = \sum_{j=0}^{n-1} b_{2j} b_{2j+1}$. The state takes the form:
\begin{equation}
    \kett{\sigma(P_k, \pi/4)}_{\mathcal{P}} = \frac{1}{\sqrt{2}}\left(\ket{0...0} + \phi_k \ket{b_0...b_{2n-1}} \right),
\end{equation}
where $\phi_k = i (-i)^{n_Y} \in \{1,-i,-1,i\}$.

The Clifford circuit that prepares it from $\ket{0...0}$ is as follows. Let $(x_1,...,x_l)$ denote the indices of the $l \in \{0,..., 2n-1\}$ bits where $b_{x_i} = 1$. Apply the Hadamard gate to the $x_1$-th qubit, followed by a sequence of $l-1$ CNOT gates with controls and targets $(x_1, x_2), (x_2, x_3), ..., (x_{l-1} , x_l)$, yielding the state $\frac{1}{\sqrt{2}}(\ket{0...0} + \ket{b_0 ... b_{2n-1}})$. Finally, to account for the phase $\phi_k$, it suffices to apply a single qubit Clifford to the $x_1$-th qubit, depending on the value of $n_Y$. For $n_Y \! \pmod{4} = 0, 1, 2, 3$, apply $S, \mathbb{I}, S^\dagger, Z$ respectively.
\end{proof}

For example, the circuit that prepares $\kett{\sigma(ZY, \pi/4)}_{\mathcal{P}}$ is:
\begin{center}
    \begin{quantikz}[row sep={0.8cm,between origins}, column sep=0.6cm]
    \lstick{$q_0$} & \gate{H}  & \ctrl{2} & \qw      &  \qw \\
    \lstick{$q_1$} & \qw            & \qw      & \qw       & \qw \\
    \lstick{$q_2$} & \qw            & \targ{}  & \ctrl{1}  & \qw \\
    \lstick{$q_3$} & \qw            & \qw      & \targ{}   & \qw
    \end{quantikz} $~ ~ = \frac{1}{\sqrt{2}}(\ket{0000} + \ket{1011}) = \kett{\sigma(ZY, \pi/4)}_{\mathcal{P}}$.
\end{center}

\subsection{Reduction of operator overlaps to fidelities} \label{app:overlap_fidelity_red}

We begin with the following fact:
\begin{lemma} \label{lemma:fid_to_overlap}
    Let $\hat{p}_\sigma$ be an $\epsilon$-accurate estimator of $p(\sigma(O),\sigma(O'))$ taking $N=\bigo{1/\epsilon^2}$ samples, i.e.:
    \begin{equation} \label{eq:epsilon_p_estimate}
        \left| \hat{p}_\sigma - p(\sigma(O),\sigma(O')) \right| \leq \epsilon.
    \end{equation}
    Then $\hat{c} \equiv 2\sqrt{\hat{p}_\sigma} -1$ is an $\epsilon$-accurate estimator of $c(O,O')$ taking $N=\bigo{1/\epsilon^4}$ samples.
\end{lemma}
\begin{proof}
The overlap $c(O,O')$ can be written as a function of the fidelity between the shifted operators $p(\sigma(O,\theta), \sigma(O',\theta))$, as:
\begin{align*}
p(\sigma(O, \theta), \sigma(O', \theta)) &= \abs{c(\sigma(O, \theta), \sigma(O', \theta))}^2 \\ 
&= \abs{\tr\left((\cos(\theta)\mathbb{I} + i\sin(\theta) O)^\dagger(\cos(\theta)\mathbb{I} + i\sin(\theta) O')\right)/2^n}^2 \\
&= \left( \cos^2(\theta) + \sin^2(\theta) c(O,O') \right)^2,
\end{align*}
since $O$ and $O'$ are traceless. The parameter $\theta$ is chosen so that they are one-to-one functions; a possible choice is $\theta = \pi/4$, which implies that $c(O,O') = -1 \iff p(\sigma(O),\sigma(O')) = 0$ and $c(O,O') = 1 \iff p(\sigma(O),\sigma(O')) = 1$. 

Making use of $\abs{\sqrt{x}-\sqrt{y}} \leq \sqrt{\abs{x-y}}$ yields:
\begin{equation*}
    \abs{\hat{c}-c(O,O')} \leq 2 \abs{\sqrt{\hat{p}_\sigma} - \sqrt{p(\sigma(O),\sigma(O'))}} \leq 2 \sqrt{\abs{\hat{p}_\sigma - p(\sigma(O),\sigma(O'))}} \leq 2 \sqrt{\epsilon},
\end{equation*}
which implies that $\hat{c}$ is an $\epsilon$-accurate estimator of $c(O,O')$ taking $N=\bigo{1/\epsilon^4}$ samples.
\end{proof}

As a result of Lemma~\ref{lemma:fid_to_overlap}, operator overlaps $c(O,O')$ can be estimated through fidelity estimation algorithms such as the SWAP test or (global) Clifford classical shadows (Theorem~\ref{theorem:huang_predicting_clifford}), which generally require less stringent access to $O$ and $O'$ (they only require access to samples of $\kett{\sigma(O)},\kett{\sigma(O')}$), at the expense of a quadratically worse dependence on $\epsilon$. Applying the estimator $\hat{c}$ on the set of fidelities obtained via the Clifford classical shadow of $O$ straightforwardly proves the query complexity upper bound of Theorem~\ref{theorem:simul_2pc_shadows}. \\

In addition, due to Lemma~\ref{lemma:shifted_stab_state}, Direct Fidelity Estimation (DFE) \cite{flammia2011direct} can be employed to compute $c(P_k,O)$, which only requires $N=\bigo{\log(1/\delta)/\epsilon^4}$ samples of the state $\kett{\sigma(O)}_{\mathcal{P}}$. That is, express the fidelity as:
\begin{alignb}
\label{app:dfe_eq}
p(P_k, O)
&= \tr\left(
    \kett{\sigma(P_k)}_{\mathcal{P}}
    \leftindex_{{\mathcal{P}}}{\bbra{\sigma(P_k)}}
    \,
    \kett{\sigma(O)}_{\mathcal{P}}
    \leftindex_{{\mathcal{P}}}{\bbra{\sigma(O)}}
\right)
\\
&= \frac{1}{4^n}
\sum_{i=1}^{4^n}
\langle Q_i \rangle_{\kett{\sigma(P_k)}_{\mathcal{P}}}
\langle Q_i \rangle_{\kett{\sigma(O)}_{\mathcal{P}}}
\\
&= \mathop{\mathbb{E}}_{i \sim \pi(\kett{\sigma(P_k)})}
\left[
    \frac{
        \langle Q_i \rangle_{\kett{\sigma(O)}_{\mathcal{P}}}
    }{
        \langle Q_i \rangle_{\kett{\sigma(P_k)}_{\mathcal{P}}}
    }
\right].
\end{alignb}
where $\pi(\kett{\sigma(P_k)}_{\mathcal{P}}) \equiv \langle Q_i \rangle^2_{\kett{\sigma(P_k)}_{\mathcal{P}}}/4^n$ is the distribution of the state $\kett{\sigma(P_k)}_{\mathcal{P}}$ in the operator basis $\mathcal{Q}$ (usually also taken to be the Pauli basis). It can then be estimated via Monte Carlo by sampling from $\pi(\kett{\sigma(P_k)}_{\mathcal{P}})$, evaluating the fraction in Eq.~\ref{app:dfe_eq}, and taking their empirical average. By virtue of $\kett{\sigma(P_k)}_{\mathcal{P}}$ being a stabilizer state, both sampling and the evaluation of expectation values (of the denominator) can be efficiently achieved classically. The total number of samples of $\pi(\kett{\sigma(O)}_{\mathcal{P}})$ needed to estimate $p(\sigma(P_k), \sigma(O))$ to precision $\epsilon$ and success probability $1-\delta$ is $N=\bigo{\log(1/\delta)/\epsilon^2}$ \cite{flammia2011direct}. This estimate can then be used to compute $c(P_k,O)$ via Lemma~\ref{lemma:fid_to_overlap}.

Finally, return to the task of Theorem~\ref{theorem:simul_2pc_shadows} for simultaneously estimating $M$ two-point correlators $\{\tr(P_kO(t))/2^n\}_{k=1}^{M}$. Naive term-by-term estimation takes $\bigo{M \log(M/\delta)/\epsilon^2}$ samples due to the union bound. Alternatively, classical shadows with random global Cliffords (Theorem~\ref{theorem:huang_predicting_clifford} from \cite{huang2020predicting}) -- which is able to simultaneously estimate fidelities with respect to $M$ target pure states at once using $\bigo{\log(M/\delta)/\epsilon^2}$ runs -- can be employed for this task to improve the dependence on $M$. That is, choose the unitary ensemble $\mathcal{U}$ to be the $2n$-qubit Clifford group $\text{Cl}(4^n)$. Random unitaries $U \in \mathcal{U}$ are then applied to samples of $\kett{\sigma(O(t))}_\mathcal{Q}$ before measurement in the computational basis, resulting in measurement results that can be efficiently post-processed to generate a classical shadow of $\kett{\sigma(O(t))}_\mathcal{Q}$. Applying the reduction discussed in the previous section to the obtained fidelities yields the guarantee of Theorem~\ref{theorem:simul_2pc_shadows}. An analogous procedure holds for general choices of $\theta$, since $\kett{\sigma(P_k, \theta)}_{\mathcal{Q}}$ is simply a linear combination of two stabilizer states, which can be efficiently represented classically.

This result implies that a modest amount of resources (at most $\bigo{n}$ in number of samples and $\bigo{\log(n)}$ in additional depth for random measurements) suffice to construct a classical description of $O(t)$ that contains accurate information on all two-point correlators.

\section{Exponential separation for learning diagonal OTOCs with vs without ancillas} \label{app:proof_ancilla}

In this appendix, we are concerned with the following problem:\\

\noindent \textbf{All diagonal OTOC Estimation Problem:} Estimate $\left \{\tr(O \mathcal{E}(Q) O \mathcal{E}(Q))/2^n \right \}_{Q \in {\mathcal{P}}_n}$ up to additive error $\epsilon$ with success probability $1-\delta$. \\

We will prove Theorem~\ref{theorem:expsep_diag}.1 and \ref{theorem:expsep_diag}.2, which shows an exponential separation in query complexity for this problem, for protocols with vs without ancillas, but can be adaptive and possess quantum memory. Theorem~\ref{theorem:expsep_diag}.3, the case without adaptivity and quantum memory is separately proven in Appendix~\ref{app:ancilla_non_adapt}.

Our (brief) proof of the lower bound is based on the observation that the hard instances for our problem are also hard instances of the Pauli channel eigenvalue estimation problem, defined in e.g. \cite{chen2022quantum,chen2024tight,chen2025efficient}. This immediately allows us to make use of existing bounds for the Pauli channel eigenvalue estimation problem, based on reducing it to a channel discrimination problem. 

\begin{lemma}[Reduction to distinguishing problem]
Fix $0<\alpha\leq 1/3$. Any protocol that can solve the All Diagonal OTOC Estimation Problem to additive error $\epsilon<\alpha/2$ and success probability at least $2/3$ by querying channels of the form ${\mathcal{O}_{\mathcal{E}, O}} \equiv \mathcal{E}^\dagger \circ \tilde{O} \circ \mathcal{E}$ (where $\mathcal{E}$ is a doubly-stochastic channel and $\tilde{O}$ is conjugation by the Pauli operator $O$) can also solve, with success probability at least $2/3$, the two-hypothesis distinguishing problem between the following two scenarios, which occur with equal probability:
\begin{enumerate}
    \item The given channel corresponds to the completely depolarizing channel
    \(
        \Lambda_0(\cdot) \equiv \frac{1}{2^n} \mathbb{I} \tr(\cdot).
    \)
    
    \item The given channel is sampled uniformly from the channels $\{\Lambda_i\}_{i=1}^{4^n-1}$, where
    \(
        \Lambda_i(\cdot)
        \equiv
        \frac{1}{2^n}
        \left(
            \mathbb{I}\tr(\cdot)
            +
            \alpha Q_i \tr(Q_i\,\cdot)
        \right),
        Q_i \in \mathcal{P}_n \backslash \{\mathbb{I}^{\otimes n}\}.
    \)
\end{enumerate}
\end{lemma}
\begin{proof}
By Lemma~\ref{lemma:hard_instances_echo_form}, each channel $\Lambda_i$, for $i=0,\ldots,4^n-1$, can be written in the form
\[
    \Lambda_i
    \equiv
    \mathcal{E}_i^\dagger \circ \tilde{P}_i \circ \mathcal{E}_i.
\]
Subsequently, for the input $\Lambda_i$, the resulting set of diagonal OTOCs associated with the nonidentity Pauli operators $\{Q_j\}_{j=1}^{4^n-1}$ evaluates to
\begin{equation*}
    \frac{1}{2^n}
    \tr\!\left(
        P_i\mathcal{E}_i(Q_j)P_i\mathcal{E}_i(Q_j)
    \right)
    =
    \tr\!\left(
        \rho_{\Lambda_i}Q_j\otimes Q_j^T
    \right)
    =
    \frac{1}{2^n}\tr\!\left(Q_j\Lambda_i(Q_j)\right)
    =
    \alpha\delta_{i,j}.
\end{equation*}
For $\Lambda_0$, these quantities are equal to $0$ for every $j$. On the other hand, for $\Lambda_i$ with $i\neq 0$, it is $\alpha\delta_{i,j}$.

Let $\widehat{C}_j$ denote the estimates returned by the protocol. We output the second hypothesis if
\[
    \max_{1\leq j\leq 4^n-1}|\widehat{C}_j|
    >
    \frac{\alpha}{2},
\]
and the first hypothesis otherwise. Whenever all the estimates have additive error at most $\epsilon$, this rule correctly distinguishes the two cases. Indeed, if the channel is $\Lambda_0$, then
\[
    \max_j|\widehat{C}_j|
    \leq
    \epsilon
    <
    \frac{\alpha}{2},
\]
whereas, if the channel is $\Lambda_i$ with $i\neq 0$, then
\[
    |\widehat{C}_i|
    \geq
    \alpha-\epsilon
    >
    \frac{\alpha}{2}.
\]
Therefore, the distinguishing procedure succeeds whenever the All Diagonal OTOC Estimation protocol succeeds, and hence with probability at least $2/3$. This validates our claim.
\end{proof}

The hardness of the above two-hypothesis distinguishing problem is given by the following known result, which we quote without proof.

\begin{theorem}[adapted from \cite{chen2025efficient}, Theorem 7] \label{theorem:chen_lowerbound}
Let $0 < \epsilon \leq 1$. Consider any protocol without ancillas, but can possibly be adaptive and use quantum memory\footnote{Note that the terminology `quantum memory' is used in \cite{chen2025efficient} to refer to access to ancillas, while we use it to refer to the ability to query the channel multiple times per coherent run (which \cite{chen2025efficient} calls `concatenation').}, based on querying unknown Pauli channels, that can distinguish between the following scenarios to at least large constant probability: 
\begin{enumerate}
    \item The Pauli channel corresponds to the completely depolarizing channel $\Lambda_0(\cdot) \equiv \frac{1}{2^n} \mathbb{I} \tr(\cdot)$.
    
    \item The Pauli channel is sampled uniformly from one of the channels $\{\Lambda_{i}\}_{i=1}^{4^n-1}$, where $\Lambda_{i} (\cdot) \equiv \frac{1}{2^n}(\mathbb{I} \tr(\cdot) + \alpha Q_i \tr(Q_i ~\cdot))$, $Q_i \in {\mathcal{P}}_n \backslash \mathbb{I}^{\otimes n}$, where $\alpha \leq 1/3$.
\end{enumerate}
The protocol must take at least $\Omega(2^n/\epsilon^2)$ runs.
\end{theorem}

We can now prove Theorem~\ref{theorem:expsep_diag}.

\begin{proof}
For the lower bound, consider any protocol that satisfies the assumptions of Theorem~\ref{theorem:expsep_diag}. By the above Lemma, it can also solve the stated distinguishing problem. Invoking Theorem~\ref{theorem:chen_lowerbound} -- which shows a $\Omega(2^n/\epsilon^2)$ measurement lower bound on the stated distinguishing problem -- proves our claim.

On the other hand, the upper bound is obtained by applying our $2n$-qubit algorithm based on querying ${\mathcal{O}_{\mathcal{E}, O}}$ to prepare $\rho_{\mathcal{O}_{\mathcal{E}, O}}$, and measuring it in the Bell basis, which is the simultaneous eigenbasis of the $4^n -1$ observables $\{Q_i \otimes Q_i^T\}$, because:
\begin{equation} \label{eq:expt_otoc}
    \tr(\rho_{\mathcal{O}_{\mathcal{E}, O}} Q_i \otimes Q_i^T) = \frac{1}{2^n} \tr(Q_i^\dagger (\mathcal{E}^\dagger \circ \tilde{O} \circ \mathcal{E})(Q_i)) = \frac{1}{2^n} \tr(O \mathcal{E}(Q_i) O \mathcal{E}(Q_i)).
\end{equation}
By the union bound, this algorithm achieves the $\bigo{n/\epsilon^2}$ upper bound. 
\end{proof}

\section{Family of tight bounds for learning diagonal OTOCs using ancillas} \label{app:ancilla_non_adapt}

In this appendix, we prove the tight $\Theta(n2^{n-n_{\text{anc}}}/\epsilon^2)$ bounds stated in Theorem~\ref{theorem:expsep_diag}.3 on the All Diagonal OTOC Estimation Problem discussed in the main text and Appendix~\ref{app:proof_ancilla}. The tight bounds interpolate between the exponentially separated upper and lower bounds of Theorem~\ref{theorem:expsep_diag}.1 and \ref{theorem:expsep_diag}.2 (up to a factor of $n$), when the number of ancillas $n_{\text{anc}}$ is varied between $0$ and $n$.

\subsection{Lower bound}

\begin{theorem}[Lower bound of Theorem~\ref{theorem:expsep_diag}.3] \label{theorem:expsep_diag_nonad}
Consider OTOC learning protocols based on querying channels of the form ${\mathcal{O}_{\mathcal{E}, O}} = \mathcal{E}^\dagger \circ \tilde{O} \circ \mathcal{E}$ (for arbitrary doubly-stochastic maps $\mathcal{E}$ and Paulis $O$) on arbitrary $(n+n_{\text{anc}})$-qubit density matrices, with $0 \leq n_{\text{anc}} \leq n$, but without adaptivity and quantum memory. The query complexity lower bound $N=\Omega(n 2^{n-n_{\text{anc}}}/\epsilon^2)$ holds. 
\end{theorem}

\begin{proof}
We use the proof techniques of \cite{huang2020predicting,chen2022quantum} to encode the OTOC estimation task as a two-party communication protocol, whose mutual information can be bounded to yield a query-complexity lower bound. We opt for brevity here, since a similar approach was used in Appendix~\ref{app:otoc_lower_shadow}, with the general proof strategy detailed in Appendix~\ref{app:shadow_proof_strat}. It is also simpler here because the learner has knowledge of the exact observables to be estimated (namely, all diagonal OTOCs), which eliminates need for the third party Loki.

In particular, consider the channels:
\begin{equation}
    \Lambda_{i,s} (\cdot) \equiv \mathcal{E}^\dagger_i \circ \tilde{P}_{i,s} \circ \mathcal{E}_i (\cdot) = \frac{1}{2^n}(\mathbb{I} \tr(\cdot) + 3 \epsilon s Q_i \tr(Q_i ~\cdot)),
\end{equation}
where $\mathcal{E}_i, \tilde{P}_{i,s}$ follows the definition of the channels appearing in Eq.~(\ref{eq:instance_ei}) of Lemma~\ref{lemma:hard_instances_echo_form}, whose diagonal PTM elements correspond to OTOCs. The Choi state of $\Lambda_{i,s}$ takes the form of Eq.~(\ref{eq:choi_diag}).

Alice randomly chooses one of the $2(4^n-1)$ possible choices of $ \Lambda_{i,s} (\cdot)$, repeats $N$ times, and sends them to Bob. If there exists a diagonal OTOC estimation protocol satisfying the Theorem's assumptions which uses $N$ queries to $\Lambda_{i,s}$, Bob can use this protocol to uniquely determine $(i,s)$ with high probability.

Notice that $\Lambda_{i,s} (\cdot)$ are Pauli channels; in particular, they are identical to the hard instances defined in the proof of Theorem 5 in \cite{chen2022quantum}, up to the $3\epsilon$ factor. Subsequently, the proof proceeds almost verbatim as \cite{chen2022quantum}, which we sketch below for completeness:
\begin{itemize}
    \item Let $\{\rho_i, E_i\}$ denote Bob's input state and POVM pair for the $i$-th sample, and $I((i,s): \{\rho_1, E_1\}, ..., \{\rho_N, E_N\})$ the mutual information between Alice's distribution and Bob's measurement results.

    \item Applying the chain rule of mutual information and Fano's inequality yields:
    \begin{equation*}
        I((i,s): \{\rho_1, E_1\}, ..., \{\rho_N, E_N\}) = \sum_{i=1}^N I((i,s): \{\rho_i, E_i\}) \geq \Omega(n).
    \end{equation*}

    \item Bounding $I((i,s): \{\rho_i, E_i\})$, we find (Lemma 6, \cite{chen2022quantum}):
    \begin{equation*}
        I((i,s): \{\rho_i, E_i\}) \leq 3 \epsilon^2 \frac{2^{n_{\text{anc}}}}{2^n-1}.
    \end{equation*}
    Substituting into the preceding equation completes the proof.
\end{itemize}
\end{proof}

\subsection{Upper bound \& optimal algorithms for diagonal OTOCs}

Next, we describe a set of algorithms (without adaptivity and quantum memory) that solve the All Diagonal OTOC Estimation Problem. They require between $0$ to $n$ ancilla qubits (for a total of $n$ to $2n$ qubits), and can be viewed as interpolating between the $n$ and $2n$-qubit algorithms. Their query complexities also achieve the lower bound of Theorem~\ref{theorem:expsep_diag_nonad}, and are therefore optimal among the class of protocols satisfying the Theorem's assumptions.

In fact, more generally, the algorithm queries any quantum channel $\mathcal{N}$, and is able to estimate its $4^n-1$ diagonal PTM elements $\{\frac{1}{2^n} \mathrm{tr} \left( P_{(u,s)} \mathcal{N}(P_{(u,s)}) \right)\}_{(u,s)}$ using $n + k$ qubits, without quantum memory or adaptivity. It is direct generalization of existing algorithms for computing the Pauli eigenvalues/diagonal PTM elements of Pauli channels \cite{flammia2020efficient,chen2022quantum}; in particular, see Algorithm 1 of \cite{chen2022quantum}. Inputting the channels $\mathcal{N} = \mathcal{E}^\dagger \circ \tilde{O} \circ \mathcal{E}$ into this algorithm then immediately yields an algorithm for the All Diagonal OTOC Estimation Problem. For this reason, we will first describe the general algorithm for estimating diagonal PTM elements (in Appendix~\ref{app:algo_diag_ptm}), which straightforwardly yields the algorithm for diagonal OTOCs (in Appendix~\ref{app:algo_diag_otoc}).

\subsubsection{Preliminaries}
We begin by setting up some additional notation and results for convenience.
\begin{itemize}
    \item Symplectic representation of Pauli operators: Let $a = (a^x_1,...,a^x_n,a^z_1,...,a^z_n) \in \mathbb{Z}_2^{2n}$ be length-$2n$ binary vectors. Any Pauli operator can be uniquely labelled by $a$ as $P_a = \bigotimes_{i=1}^n i^{a_i^x a_i^z} X^{a_i^x} Z^{a_i^z}$.

    \item Symplectic inner product: $\langle a,b\rangle \equiv \sum_{i=1}^n (a_i^x b_i^z + a_i^z b_i^x)$. We have that $P_a P_b = (-1)^{\langle a,b \rangle} P_b P_a$.

    \item Other useful identities:
    \begin{align}
        P_a P_b P_a &= (-1)^{\langle a,b \rangle} P_b \label{eq:conjug} \\
        (-1)^{\langle a, b \oplus c \rangle} &= (-1)^{\langle a,b \rangle} (-1)^{\langle a,c \rangle} \label{eq:bitadd}
    \end{align}
    where $\oplus$ indicates bitwise addition modulo 2.

    \item Twirling identity:
    \begin{equation} \label{eq:twirling_id}
        \frac{1}{\norm{\mathcal{P}_n}} \sum_{P_r} (-1)^{\langle r,  s \oplus s'\rangle} = \delta_{s, s'}.
    \end{equation}
    \begin{proof}
    When $s=s'$, $\langle r, s\oplus s'\rangle = \langle r, 0\rangle = 0$.
    When $s \neq s'$, $s \oplus s'$ indexes a non-identity Pauli operator. Since non-identity Pauli operators commute and anticommute with the same number of Pauli operators in $\mathcal{P}_n$, half of the terms in the sum with value $-1$ cancel out with the other half with value $1$.
    \end{proof}

    \item Define the maximally entangled $2k$-qubit state / vectorization of $k$-qubit Pauli operator to be:
    \begin{equation}
        \kett{P_v} = (P_v \otimes \mathbb{I}) \kett{\mathbb{I}}.
    \end{equation}
    Using Eq.~(\ref{eq:bell_state}), it can be written in the form:
    \begin{equation}
        \kettbbra{P_v}{P_v} = \frac{1}{4^k} \sum_u (-1)^{\langle u,v \rangle} P_u \otimes P_u^T.
    \end{equation}

    \item A \textit{stabilizer group} $S$ on $k$ qubits is a group of $2^k$ Pauli operators that all pairwise commute. 

    Define the states:
    \begin{equation}
        \ketbra{\phi^S_e}{\phi^S_e} \equiv \frac{1}{2^k} \sum_{s \in S} (-1)^{\langle s,e \rangle} P_s, ~~\forall e \in \mathbb{Z}_2^{2k}/S.
    \end{equation}
    They are simultaneous eigenstates of all operators in $S$, and are stabilizer states. The states $\{\ket{\phi^S_e}\}_{e \in \mathbb{Z}_2^{2k}/S}$ form an orthonormal basis.

    \item A \textit{stabilizer covering} $O = \{S_1,...,S_{\norm{O}}\}$ on $n$ qubits is a set of stabilizer groups (on $n$ qubits) such that $\forall P \in \mathcal{P}_n, \exists i ~ \text{s.t} ~P \in S_i$.

    Fact 1: There exists a stabilizer covering of size $2^k + 1$ for $\mathcal{P}_k$ \cite{lawrence2002mutually}.

    Fact 2: The stabilizer covering generated by all local tensor-product Pauli measurements is of size $3^k$ (since there are $3^k$ possible measurement settings in total).
    
\end{itemize}

\subsubsection{Algorithm for diagonal PTM elements} \label{app:algo_diag_ptm}

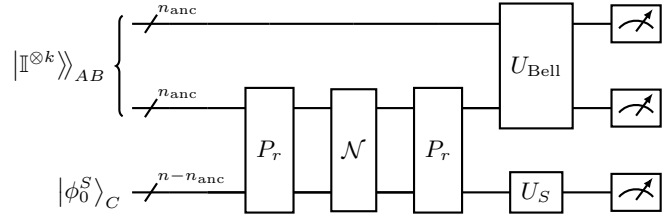
\begin{figure}[h]
\centering
\begin{quantikz}
\lstick[wires=2]{\(\kett{\mathbb{I}^{\otimes k}}_{AB}\)}   & \qwbundle{n_{\text{anc}}}  & \qw & \qw & \qw & \qw & \gate[2]{U_{\text{Bell}}}  &  \meter{} \\
   & \qwbundle{n_{\text{anc}}}   & \qw         &      \gate[2]{P_r}  & \gate[2]{\mathcal{N}} & \gate[2]{P_r} & \qw & \meter{} \\
\lstick{\(\ket{\phi^S_0}_C\)} & \qwbundle{n-n_{\text{anc}}} & \qw & \qw & \qw & \qw & \gate[1]{U_S} & \meter{}
\end{quantikz}
\caption{\textbf{Circuit for computing the diagonal PTM elements of $\mathcal{N}$ using $n+n_{\text{anc}}$ qubits.} An instance of the $(n+n_{\text{anc}})$-qubit circuit to be run, for a specific choice of $S \in O$ and $P_r \in \mathcal{P}_n$. All qubits are measured in the computational basis in the end. $U_{\text{Bell}}$ is the Bell basis transformation, while $U_S$ transforms to the basis $\{\ket{\phi^S_e}\}_{e \in \mathbb{Z}_2^{2n_{\text{anc}}}/S}$.}
\label{fig:cirq_nk}
\end{figure}

First, we describe the algorithm for estimating the PTM elements of general channels $\mathcal{N}$, based on adapting Algorithm 1 of \cite{chen2022quantum}. Partition the joint $(n+n_{\text{anc}})$-qubit system as $\mathcal{H}_A \otimes \mathcal{H}_B \otimes \mathcal{H}_C$, with $n_{\text{anc}}$, $n_{\text{anc}}$, and $n-n_{\text{anc}}$ qubits respectively. Let $\mathcal{N} : \mathcal{L}(\mathcal{H}_B \otimes \mathcal{H}_C) \rightarrow \mathcal{L}(\mathcal{H}_B \otimes \mathcal{H}_C)$ be a queryable channel whose diagonal PTM elements we wish to learn. Let $O$ be a stabilizer covering on $n-n_{\text{anc}}$ qubits, $S \in O$ a stabilizer group, and define the states $\kett{\mathbb{I}}, \kett{P_v} \in \mathcal{H}_A \otimes \mathcal{H}_B$, $\ket{\phi^S_0}, \ket{\phi^S_e} \in \mathcal{H}_C$. The procedure below -- involving circuits of the form Fig.~(\ref{fig:cirq_nk}) -- outputs the $4^n-1$ diagonal PTM elements $\{\frac{1}{2^n} \mathrm{tr} \left( P_{(u,s)} \mathcal{N}(P_{(u,s)}) \right)\}_{(u,s)}$ using $n + n_{\text{anc}}$ qubits, without quantum memory or adaptivity:
\begin{enumerate}
    \item For each stabilizer group $S \in O$, prepare the initial state $\kett{\mathbb{I}} \otimes \ket{\phi^S_0}$.
    
    \item Apply the ``twirled channel" $\mathcal{N}_{r} \equiv \widetilde{P}_r \circ \mathcal{N} \circ \widetilde{P}_r$ to subsystem $\mathcal{H}_B \otimes \mathcal{H}_C$, with $P_r \in \mathcal{P}_{n}$ chosen randomly.

    \item Measure this state in the basis $\{\kett{P_v} \otimes \ket{\phi^S_e}\}_{v,e}$ and record the obtained outcome $(v,e)$.

    \item Compute and record the quantity $f_{v,e,s,u} = (-1)^{\langle u,v \rangle + \langle s, e \rangle}$.

    \item Repeat the above steps and compute the empirical average of $f_{v,e,s,u}$. This yields an estimate of $\frac{1}{2^n} \mathrm{tr} \left( P_{(u,s)} \mathcal{N}(P_{(u,s)}) \right)$ for all $P_{(u,s)}$ that are covered by the stabilizer group $S$.

    \item After having iterated through all stabilizer groups in the stabilizer covering $O$, classical post-processing yields all $4^n-1$ diagonal PTM elements $\{\frac{1}{2^n} \mathrm{tr} \left( P_{(u,s)} \mathcal{N}(P_{(u,s)}) \right)\}_{(u,s)}$.
\end{enumerate}

The only difference with Algorithm 1 of \cite{chen2022quantum}, which estimates the Pauli eigenvalues of arbitrary Pauli channels (equivalently, their diagonal PTM elements), is the presence of twirling by the unitary channels $\tilde{P}_r(\cdot) = P_r(\cdot)P_r$. This operation converts any input channel into a Pauli channel on average, without changing their diagonal PTM elements. The rest of the algorithm is then identical to Algorithm 1 of \cite{chen2022quantum}. 

Intuitively, Algorithm 1 of \cite{chen2022quantum} can be understood as sampling from the Pauli error rates $\bf{p}$ of an $n$-qubit Pauli channel $\mathcal{M}(\cdot)=\sum_{a \in \mathbb{Z}_2^{2n}} p_a P_a(\cdot)P_a$, which yields estimators for its Pauli eigenvalues/diagonal PTM elements via the Walsh-Hadamard transform:
\begin{equation*}
    \frac{1}{2^n}\tr(P_b \mathcal{M} (P_b)) = \sum_{a \in \mathbb{Z}_2^{2n}} p_a (-1)^{\langle a,b \rangle} = \mathbb{E}_{a}[(-1)^{\langle a,b \rangle}].
\end{equation*}
Since twirling preserves the values of the diagonal PTM elements of $\mathcal{N}$, the above procedure samples from the `effective Pauli error rates' of $\mathcal{N}$ (obtained by applying the inverse Walsh-Hadamard transform to the PTM elements), which yields estimators for its diagonal PTM elements. A more formal description and analysis is given by Algorithm~\ref{alg:diag_ptm_nk} and Theorem~\ref{theorem:ptm_diagonal_algo_nk} below:

\vspace{5pt}
\begin{algorithm}[H]
\linespread{1.15}\selectfont
\caption{Algorithm for computing all diagonal PTM elements using $n+k$ qubits} \label{alg:diag_ptm_nk}

\KwData{Stabilizer covering $O$ of $\mathcal{P}_{n-n_{\text{anc}}}$, $N$ samples of the quantum channel $\mathcal{N}$}
\KwResult{List containing $4^n-1$ diagonal OTOCs $\vec{\gamma}$, with $\gamma_a = \tr(P_a \mathcal{N}( P_a))/2^n$}
  \vspace*{2pt}

  $[\hat{\gamma}_1,...,\hat{\gamma}_{4^n-1}] = [0,...,0]$, $[N_1,...,N_{4^n-1}] = [0,...,0]$ \;
  For $S \in O$:\\
  \Indp For $i \in [1,..., \lfloor N/\abs{O} \rfloor]$:\\
    \Indp Prepare $\kett{\mathbb{I}}_{AB} \otimes \ket{\phi^S_0}_C$ \;
    Apply $(\mathbb{I}_{A} \otimes \tilde{P}_r)(\mathbb{I}_{A} \otimes \mathcal{N})(\mathbb{I}_{A} \otimes \tilde{P}_r)$, with $P_r \in \mathcal{P}_n$ chosen randomly \;
    Measure in the basis $\{\kett{P_v} \otimes \ket{\phi^S_e}\}_{v,e}$ and record outcome $a=(v,e)$ \;
    For $u \in \mathbb{Z}_2^{2n_{\text{anc}}}$: \\
    \Indp For $s\in S$: \\
    \Indp $\hat{\gamma}_{u \oplus s} ~+\!\!= (-1)^{\langle u,v \rangle + \langle s, e \rangle}$ \;
    $N_{u \oplus s} ~+\!\!= 1$ \;

  \Indm \Indm \Indm \Indm \textbf{Return} $\vec{\gamma}$, with $\gamma_a = \hat{\gamma}_a/N_{a}$\;
\end{algorithm}
\vspace{5pt}

\begin{theorem}[Algorithm for computing all diagonal PTM elements of general channels] \label{theorem:ptm_diagonal_algo_nk}
Given the ability to query a general $n$-qubit channel $\mathcal{N}$, Algorithm~\ref{alg:diag_ptm_nk} estimates all $4^n-1$ diagonal PTM elements $\{\frac{1}{2^n} \mathrm{tr} \left( P_{(u,s)} \mathcal{N}(P_{(u,s)}) \right)\}_{(u,s)}$ using $\bigo{n 2^{n-n_{\text{anc}}} \log(1/\delta) 1/\epsilon^2}$ queries. The algorithm uses $n+n_{\text{anc}}$ qubits, is non-adaptive, and does not require quantum memory.
\end{theorem}

\begin{proof}
Firstly, observe that the probability of obtaining the measurement outcome $\kett{P_v} \otimes \ket{\phi^S_e}$ after measuring the state $(\mathbb{I}_A \otimes \mathcal{N}_{BC})(\kett{\mathbb{I}} \otimes \ket{\phi^S_0})$ in the basis $\{\kett{P_v} \otimes \ket{\phi^S_e}\}_{v,e}$ is given by:
\begin{align}
    \Pr(v,e, \mathcal{N}) &\equiv \mathrm{tr}\left( \kettbbra{P_v}{P_v} \otimes \ketbra{\phi^S_e}{\phi^S_e} (\mathbb{I}_A \otimes \mathcal{N}_{BC}) \left( \kettbbra{\mathbb{I}}{\mathbb{I}} \otimes \ketbra{\phi^S_0}{\phi^S_0} \right) \right) \\ 
&= \frac{1}{2^{2n+k}} \sum_{u,s,s'} (-1)^{\langle u,v \rangle + \langle s', e \rangle} \mathrm{tr} \left( (P_u \otimes P_{s'}) \mathcal{N}_{BC} (P_u \otimes P_s) \right),
\end{align}
which follow directly from the definitions.

Next, denote $\mathcal{N}_{r} \equiv \widetilde{P}_r \circ \mathcal{N} \circ \widetilde{P}_r$, where $\widetilde{P}_r(\cdot) = P_r (\cdot) P_r$. We have:
\begin{align}
    \mathrm{tr} \left( (P_u \otimes P_{s'}) \mathcal{N}_{r} (P_u \otimes P_s) \right) &= \mathrm{tr} \left(P_r (P_u \otimes P_{s'})P_r ~ \mathcal{N}(P_r (P_u \otimes P_s) P_r) \right) \\
    &= (-1)^{\langle r, (u,s')\rangle + \langle r, (u,s)\rangle} \mathrm{tr} \left( P_{(u,s')} \mathcal{N}(P_{(u,s)}) \right) \\
    &= (-1)^{\langle r, (u,s') \oplus (u,s)\rangle} \mathrm{tr} \left( P_{(u,s')} \mathcal{N}(P_{(u,s)}) \right),
\end{align}
where $(a,b)$ denotes concatenation, and we used the cyclicity of the trace in the first equality, Eq.~(\ref{eq:conjug}) in the second equality, and Eq.~(\ref{eq:bitadd}) in the third.

Defining the `twirled probability distribution' $\widetilde{\mathrm{Pr}}(v,e) \equiv\frac{1}{\norm{\mathcal{P}_n}} \sum_r \Pr(v,e, \mathcal{N}_{r})$, we then have:
\begin{align}
    \widetilde{\mathrm{Pr}}(v,e) &= \frac{1}{\norm{\mathcal{P}_n}} \frac{1}{2^{2n+n_{\text{anc}}}} \sum_{r, u, s, s'} (-1)^{\langle u,v \rangle + \langle s', e \rangle} (-1)^{\langle r, (u,s') \oplus (u,s)\rangle} \mathrm{tr} \left( P_{(u,s')} \mathcal{N}(P_{(u,s)}) \right) \\
    &= \frac{1}{2^{n+n_{\text{anc}}}} \sum_{u, s, s'} \left(\frac{1}{\norm{\mathcal{P}_n}} \sum_r (-1)^{\langle r, (u,s') \oplus (u,s)\rangle} \right) (-1)^{\langle u,v \rangle + \langle s', e \rangle}  \frac{1}{2^n} \mathrm{tr} \left( P_{(u,s')} \mathcal{N}(P_{(u,s)}) \right) \\
    &= \frac{1}{2^{n+n_{\text{anc}}}} \sum_{u, s, s'} \delta_{s,s'} (-1)^{\langle u,v \rangle + \langle s', e \rangle}  \frac{1}{2^n} \mathrm{tr} \left( P_{(u,s')} \mathcal{N}(P_{(u,s)}) \right) \\
    &= \frac{1}{2^{n+n_{\text{anc}}}} \sum_{u, s} (-1)^{\langle u,v \rangle + \langle s, e \rangle} \frac{1}{2^n} \mathrm{tr} \left( P_{(u,s)} \mathcal{N}(P_{(u,s)}) \right),
\end{align}
where in the third equality we applied the twirling identity Eq.~(\ref{eq:twirling_id}).

Subsequently, observing that $\widetilde{\mathrm{Pr}}(v,e)$ and $\frac{1}{2^n} \mathrm{tr} \left( P_{(u,s)} \mathcal{N}(P_{(u,s)}) \right)$ are related by a Walsh-Hadamard transform, it can be inverted to yield (can be proven e.g. directly using Eqs.~(\ref{eq:bitadd}) and (\ref{eq:twirling_id})):
\begin{align}
    \frac{1}{2^n} \mathrm{tr} \left( P_{(u,s)} \mathcal{N}(P_{(u,s)}) \right) &= \sum_{v,e} \widetilde{\mathrm{Pr}}(v,e) (-1)^{\langle u,v \rangle + \langle s, e \rangle} \\
    &= \mathbb{E}_{(v,e) \sim \widetilde{\mathrm{Pr}}(v,e)}[(-1)^{\langle u,v \rangle + \langle s, e \rangle}].
\end{align}
In other words, the random variable $(-1)^{\langle u,v \rangle + \langle s, e \rangle}$ is an unbiased estimator of the diagonal PTM element of $\mathcal{N}$, given the ability to sample from the twirled distribution $\widetilde{\mathrm{Pr}}(v,e)$. By Hoeffding's inequality, it can be estimated to additive error $\epsilon$ and success probability $1-\delta$ using $\bigo{\log(1/\delta)/\epsilon^2}$ samples.

Finally, to recover all diagonal PTM elements, it suffices to observe that there exists a stabilizer covering $O$ of size $2^{n-n_{\text{anc}}}+1$ for $n-n_{\text{anc}}$ qubit systems \cite{lawrence2002mutually}, and apply the union bound to yield the final $\bigo{n 2^{n-n_{\text{anc}}} \log(1/\delta) 1/\epsilon^2}$ scaling.
\end{proof}

Indeed, the above algorithm and scaling is optimal for the problem of estimating all diagonal PTM elements, among the class of algorithms without adaptivity and quantum memory. This is discussed in Appendix~\ref{app:family_diag_ptm} as Theorem~\ref{theorem:diag_ptm_optimal}.

We remark that when $n_{\text{anc}}=n$, the randomization over $P_r \in \mathcal{P}_{n}$ becomes redundant and can be eliminated, in which case it reduces to the $2n$-qubit algorithm based on Bell-sampling. Also, for general stabilizer groups $S$, preparing and measuring subsystem C in the basis $\{\ket{\phi^S_e}\}_{e}$ generally requires entangling gates acting on up to $n-n_{\text{anc}}$ qubits. A more experimentally amenable version involves choosing the stabilizer covering to be the one generated by all possible Pauli measurements, where $\abs{O} = 3^{n-n_{\text{anc}}}$. This yields the suboptimal scaling of $\bigo{n 3^{n-n_{\text{anc}}} \log(1/\delta) 1/\epsilon^2}$, but only requires local Pauli measurements.

\subsubsection{Algorithm for diagonal OTOCs} \label{app:algo_diag_otoc}
To solve the All Diagonal OTOC Estimation Problem, it suffices to input the Loschmidt-echo channel $\mathcal{E}^\dagger \circ \tilde{O} \circ \mathcal{E}$ to Algorithm~\ref{alg:diag_ptm_nk}. Applying Theorem~\ref{theorem:ptm_diagonal_algo_nk}, the query complexity of the resulting algorithm is $\bigo{n 2^{n-n_{\text{anc}}} \log(1/\delta) 1/\epsilon^2}$, directly yielding:
\begin{theorem}[Upper bound of Theorem~\ref{theorem:expsep_diag}.3] \label{theorem:diagonal_otoc_algo_nk}
Given the ability to query the $n$-qubit channel $\mathcal{O}_{\mathcal{E}, O} \equiv \mathcal{E}^\dagger \circ \tilde{O} \circ \mathcal{E}$, there is an algorithm that computes all $4^n-1$ diagonal OTOCs $\left \{\tr(O \mathcal{E}(Q) O \mathcal{E}(Q))/2^n \right \}_{Q \in {\mathcal{P}}_n}$ in $\bigo{n 2^{n-n_{\text{anc}}} \log(1/\delta) 1/\epsilon^2}$ runs. The algorithm uses $n+n_{\text{anc}}$ qubits, is non-adaptive, and does not require quantum memory.
\end{theorem}

This saturates the lower bound of Theorem~\ref{theorem:expsep_diag_nonad}, and is therefore the optimal algorithm that solves the All Diagonal OTOC Estimation Problem using $0 \leq n_{\text{anc}} \leq n$ ancillas, among the class of algorithms without adaptivity and quantum memory. This proves the query complexity $\Theta(n 2^{n-n_{\text{anc}}}/\epsilon^2)$ of Theorem~\ref{theorem:expsep_diag}.3, at constant success probability.

\section{Exponential separation for learning all OTOCs with vs without quantum memory} \label{app:proof_memory}

In this appendix, we are concerned with the following problem:\\

\noindent \textbf{All OTOC Estimation Problem:} Estimate $\left \{\tr(O \mathcal{E}(Q) O \mathcal{E}(Q'))/2^n \right \}_{Q,Q' \in {\mathcal{P}}_n}$ up to additive error $\epsilon$ with success probability $1-\delta$. \\

We will prove Theorem~\ref{theorem:exp_sep_all_otocs_query}, which shows an exponential separation in query complexity for this problem, for protocols with vs without quantum memory. We begin with the upper bound:

\begin{lemma}[Upper bound; algorithm with quantum memory for all OTOCs] \label{app:qmem_ub}
There is an algorithm with quantum memory that solves the All OTOC Estimation Problem using $O(n/\epsilon^4)$ queries of the channel $\mathcal{O}_{\mathcal{E}, O}$.
\end{lemma}
\begin{proof}
By Eq.~(\ref{eq:expt_otoc}), OTOCs correspond to Pauli expectation values $\tr(\rho_{\mathcal{O}_{\mathcal{E}, O}}  Q \otimes Q'^T)$. To estimate them simultaneously, use the Pauli shadow tomography algorithm of \cite{huang2021information}, which requires quantum memory, and uses $O(\log(\frac{16^n}{\delta})/\epsilon^4)$ copies of the Choi state of $\mathcal{O}_{\mathcal{E}, O}$, which can be prepared by applying $(\mathcal{O}_{\mathcal{E}, O}\otimes \mathbb{I})$ to $\kett{\mathbb{I}}$ (Theorem 2, \cite{huang2021information}).
\end{proof}

To prove the lower bound, we employ the proof techniques of \cite{caro2024learning,chen2022exponential}. Consider the observables:
\begin{equation}
    O_{ij} \equiv \frac{1}{2}(Q_i \otimes Q_j^T + Q_j \otimes Q_i^T),
\end{equation}
which yield our desired OTOCs since $\tr(\rho_{\mathcal{O}_{\mathcal{E}, O}} O_{ij}) = \frac{1}{2^n} \tr(P \mathcal{E}(Q_i) P \mathcal{E}(Q_j))$.

Lemma~\ref{lemma:hard_instances_echo_form} implies that the states $\frac{1}{4^n}(\mathbb{I} + O_{ij})$ can be written in the form $\rho_{ij} \equiv \rho_{\mathcal{E}_{ij}^* \circ \tilde{P}_{ij} \circ \mathcal{E}_{ij}}$. Subsequently, since an algorithm that can solve the All OTOC estimation task implies an algorithm for distinguishing between the maximally mixed state and a uniform mixture of $\rho_{ij}$, the sample complexity of the former task inherits the lower bound for the latter task. To obtain this lower bound, \cite{caro2024learning} represents the querying algorithm as a learning tree and applies Le Cam's two-point method.

\begin{lemma}(Modification of \cite{caro2024learning}, Lemma D.5) \label{lemma_bound}
Let $\{O_i\}_{i=1}^M$ be $M$ traceless and self-adjoint $2n$-qubit observables satisfying the following:
\begin{enumerate}
    \item $\forall i = 1,...,M$, $\norm{O_i}_{\infty} \leq 1$.
    \item $\forall i = 1, ..., M/2$, $O_i = -O_{i+M/2}$.
    \item $\forall i = 1,..., M$, there exists a doubly-stochastic quantum channel $\mathcal{E}_i$ and Pauli operator $P_i$ such that:
    \begin{equation} \label{eq:hc_cond}
        \rho_{\mathcal{E}^\dagger_i \circ \tilde{P}_i \circ \mathcal{E}_i} = \frac{1}{4^n} (\mathbb{I}^{\otimes 2n} + 3 \epsilon O_i).
    \end{equation}
\end{enumerate}
Then any protocol without quantum memory, with any number of ancillas requires:
\begin{equation}
    \Omega \left(\frac{1}{\epsilon^2 \Delta(O_{1,...,O_M})} \right)
\end{equation}
queries of an unknown channel $\mathcal{E}^\dagger \circ \tilde{P} \circ \mathcal{E}$ to predict expectation values $\left \{ \tr(O_i \rho_{\mathcal{E}^\dagger \circ \tilde{P} \circ \mathcal{E}}) \right \}_{i=1}^M$ simultaneously to additive error $\epsilon$, with success probability $\geq 2/3$, where:
\begin{equation} \label{eq:delta}
    \Delta(O_1,...,O_M)\equiv \sup_{\substack{k \in \mathbb{N} \\\ket{\phi} \in (\mathbb{C}^2)^{\otimes (n+k)}, \norm{\phi} = 1 \\ \rho \in (\mathbb{C}^2 \times \mathbb{C}^2)^{\otimes (n+k)} , \tr(\rho)=1}} \frac{2}{M} \sum_{i=1}^{M/2} \left( \frac{\bra{\phi} \text{tr}_{in}((\rho^{T_{in}} \otimes \mathbb{I}_{out})(\mathbb{I}_{anc} \otimes O_i)) \ket{\phi}}{\bra{\phi} \rho_{anc} \otimes \mathbb{I}_{out} \ket{\phi}} \right)^2.
\end{equation}
$\rho^{T_{in}}$ denotes the partial transpose of $\rho$ in the ($n$-qubit) input subsystem. 
\end{lemma}

\begin{proof}
The proof is a minor extension of the proof of \cite{caro2024learning}, Lemma D.5, for (i) the more relaxed condition of $O_i$ having eigenvalues $\pm1$ to $\norm{O_i}_{\infty} \leq 1$, and (ii) restricting the queried quantum channel to take the form $\mathcal{E}^\dagger \circ \tilde{P} \circ \mathcal{E}$.

Essentially, observe that the ability to solve the expectation estimation task implies the ability to distinguish between the maximally mixed state and a uniform mixture of the states $\rho_{\mathcal{E}^\dagger_i \circ \tilde{P}_i \circ \mathcal{E}_i}$; a lower bound for the latter task is therefore inherited by the former task. To obtain this bound, \cite{caro2024learning} represents the querying algorithm as a learning tree and applies Le Cam's two-point method to arrive at the lower bound, proving the claim.
\end{proof}

\begin{lemma}
Consider the $M = 4^n(4^n-1)$ observables:
\begin{equation}
    O_{ij} \equiv \pm \frac{1}{2}(Q_i \otimes Q_j^T + Q_j \otimes Q_i^T),
\end{equation}
where $1 \leq i \leq j \leq 4^n$ (so that $Q_i,Q_j \neq \mathbb{I}^{\otimes n}$, and repeated terms due to $O_{ij} = O_{ji}$ are eliminated). They satisfy the conditions of Lemma~\ref{lemma_bound}. Furthermore, we have $\Delta(O_{ij}) = \Theta(1/4^n)$.
\end{lemma}

\begin{proof}
We first show that the observables satisfy the conditions of Lemma~\ref{lemma_bound}, and then compute $\Delta$. For the first step, conditions (1) and (2) are clear. Notably, for condition (1), we only require $\norm{O_i}_\infty = 1$, $\tr(O_i) = 0$, and $\tr(O_i^2) = 4^n$, which guarantees that:
\begin{equation}
    \tr(O_i (\frac{1}{4^n} (\mathbb{I}^{\otimes 2n} + 3 \epsilon O_i))) = \frac{3 \epsilon}{4^n} \tr(O_i^2) = 3 \epsilon, 
\end{equation}
and therefore that the argument of \cite{chen2022exponential} carries over, with the choice $\epsilon \leq 1/3$. This is satisfied by our choice of $O_i$, because $\tr(O_i^2) = 4^n$. It also leads to a small modification of a constant that appears when bounding the log in Eq.~(87) of \cite{chen2022exponential}, without further changes to the proof. By Lemma~\ref{lemma:hard_instances_echo_form}, condition (3) also holds. It suffices to choose $\mathcal{E}_{ij}$ and $P_{ij}$ based on $O_{ij}$, and set $s$ and $\alpha=3\epsilon/2$, with $\epsilon \leq 1/3$.

Next, we evaluate Eq.~(\ref{eq:delta}), i.e. the quantity:
\begin{equation*} \label{eq:proof_delta}
    \Delta(O_1,...,O_M) \equiv \sup_{\substack{k \in \mathbb{N} \\ \ket{\phi} \in (\mathbb{C}^2)^{\otimes (n+k)}, \norm{\phi} = 1 \\ \rho \in (\mathbb{C}^2 \times \mathbb{C}^2)^{\otimes (n+k)} , \tr(\rho)=1}} 
    \frac{2}{4^n(4^n-1)} 
    \sum_{1 \leq i \leq j \leq 4^n} \left( \frac{\bra{\phi} \text{tr}_{in}((\rho^{T_{in}} \otimes \mathbb{I}_{out})(\mathbb{I}_{anc} \otimes O_{ij})) \ket{\phi}}{\bra{\phi} \rho_{anc} \otimes \mathbb{I}_{out} \ket{\phi}} \right)^2.
\end{equation*}
Denoting the summands above as $f(O_{ij})^2$, they can be upper bounded as: $f(O_{ij})^2 = \left( \frac{f(P_i \otimes P_j^T) + f(P_j \otimes P_i^T)}{2}\right)^2 \leq \frac{f(P_i \otimes P_j^T)^2 + f(P_j \otimes P_i^T)^2}{2}$. Applying it and re-arranging terms in the sum yields:
\begin{align*}
    \Delta(O_1,...,O_M) &\leq \sup_{\substack{k \in \mathbb{N} \\ \ket{\phi} \in (\mathbb{C}^2)^{\otimes (n+k)}, \norm{\phi} = 1 \\ \rho \in (\mathbb{C}^2 \times \mathbb{C}^2)^{\otimes (n+k)} , \tr(\rho)=1}} 
    \frac{2}{4^n(4^n-1)} 
    \left( \frac{1}{2}\sum_{i=1}^{4^n} f(P_i \otimes P_i^T)^2 + \frac{1}{2}\sum_{i=1, j=1}^{4^n, 4^n} f(P_i \otimes P_j^T)^2\right) \\
    & \leq  \sup_{\substack{k \in \mathbb{N} \\ \ket{\phi} \in (\mathbb{C}^2)^{\otimes (n+k)}, \norm{\phi} = 1 \\ \rho \in (\mathbb{C}^2 \times \mathbb{C}^2)^{\otimes (n+k)} , \tr(\rho)=1}} 
    \frac{2}{4^n(4^n-1)} \sum_{i=1, j=1}^{4^n, 4^n} f(P_i \otimes P_j^T)^2 \\
    & = \frac{2(4^n-1)}{4^n} \sup_{\substack{k \in \mathbb{N} \\ \ket{\phi} \in (\mathbb{C}^2)^{\otimes (n+k)}, \norm{\phi} = 1 \\ \rho \in (\mathbb{C}^2 \times \mathbb{C}^2)^{\otimes (n+k)} , \tr(\rho)=1}} 
    \frac{1}{(4^n-1)^2} \sum_{i=1, j=1}^{4^n, 4^n} f(P_i \otimes P_j^T)^2,
\end{align*} 
where in the second equality we also used $\sum_{i=1}^{4^n} f(P_i \otimes P_i^T)^2 \leq \sum_{i=1, j=1}^{4^n, 4^n} f(P_i \otimes P_j^T)^2$. At this point, observe that the supremum above is identical to the sum in Lemma D.8 of \cite{caro2024learning}, which by the same Lemma is bounded above by $\frac{1}{(2^n - 1)^2}$ (Eq.~(D.85) of \cite{caro2024learning}). This yields the upper-bound:
\begin{equation}
    \Delta(O_1,...,O_M) \leq \frac{2(4^n-1)}{4^n(2^n-1)^2}.
\end{equation}

It remains to lower-bound $\Delta(O_1,...,O_M)$. To do this, choose $\rho = \ketbra{\phi}{\phi} = \kettbbra{\mathbb{I}}{\mathbb{I}}$, and evaluate Eq.~(\ref{eq:proof_delta}). We will make use of the fact that $\kettbbra{\mathbb{I}}{\mathbb{I}} = \frac{1}{4^n} \sum_{P \in {\mathcal{P}}_n} P \otimes P$, which implies that $\left( \kettbbra{\mathbb{I}}{\mathbb{I}} \right)^{T_{A}} = \frac{1}{2^n} S_{AB}$ and $\text{tr}_A\left( \kettbbra{\mathbb{I}}{\mathbb{I}} \right) = \frac{1}{2^n} \mathbb{I}$, where $S_{AB}$ is the SWAP operator between subsystems $A$ and $B$. The denominator of $f(O_{ij})$ evaluates to $\bra{\phi} \rho_{anc} \otimes \mathbb{I}_{out} \ket{\phi} = \bbra{\mathbb{I}} \mathbb{I}/2^n \kett{\mathbb{I}} = 1/2^n$. For the numerator, note that:
\begin{align*}
    \bra{\phi} \text{tr}_{in}((\rho^{T_{in}} \otimes \mathbb{I}_{out})(\mathbb{I}_{anc} \otimes O_{ij})) \ket{\phi} &= \frac{1}{2^n}  \bra{\phi} \text{tr}_{in}((S_{anc, in} \otimes \mathbb{I}_{out})(\mathbb{I}_{anc} \otimes O_{ij}))) \ket{\phi} \\
    &= \frac{1}{2^n}  \bra{\phi} \text{tr}_{in}((\mathbb{I}_{in} \otimes O_{ij}))(S_{anc, in} \otimes \mathbb{I}_{out})) \ket{\phi} \\
    &= \frac{1}{2^n}  \bra{\phi} O_{ij}  \ket{\phi} \\
    &= \frac{1}{2(2^n)} (\bbra{\mathbb{I}} Q_i \otimes Q_j^T \kett{\mathbb{I}} + \bbra{\mathbb{I}} Q_j \otimes Q_i^T \kett{\mathbb{I}}) \\
    &= \frac{1}{2^n} \delta_{ij}.
\end{align*}
Putting all together, we have:
\begin{align}
    \frac{2}{4^n(4^n-1)}  \sum_{1 \leq i \leq j \leq 4^n} \delta_{ij} ^2 = \frac{2}{4^n(4^n-1)}*(4^n-1) = \frac{2}{4^n},
\end{align}
since only the diagonal terms contribute. This finally yields:
\begin{equation}
    \frac{2}{4^n} \leq \Delta(O_1,...,O_M) \leq \frac{2(4^n-1)}{4^n(2^n-1)^2}.
\end{equation}
\end{proof}

Combining the above two Lemmata and the upper bound of Lemma~\ref{app:qmem_ub} proves Theorem~\ref{theorem:exp_sep_all_otocs_query}.

\section{Learning PTM elements of general and Pauli channels}

In this section, we consider the problem of estimating the Pauli Transfer Matrix (PTM) elements of a quantum channel, defined as the quantities $\tr(P_i \mathcal{N}(P_j))/2^n$, where $P_i, P_j \in \mathcal{P}_n$, and $\mathcal{N}$ is an unknown (but queryable) quantum channel acting on $n$ qubits. They can be regarded as the matrix elements of the matrix representation of the channel $\mathcal{N}$ (i.e. its transfer matrix) in the Pauli basis.

The learning of all PTM elements of an unknown channel has been studied in \cite{caro2024learning}, which provides lower and upper bounds for protocols possessing resources such as adaptivity and quantum memory. When $\mathcal{N}$ is additionally promised to be a Pauli channel, the same problem is also known as the Pauli channel eigenvalue estimation problem \cite{chen2022quantum,chen2024tight,chen2025efficient,kim2026fundamental}. 

Our results in this section -- which are direct implications of our results for learning OTOCs, and known results on the Pauli channel eigenvalue estimation problem -- complement existing results on learning PTM elements of general channels by providing lower and (algorithmic) upper bounds in the cases where the targeted set of PTM elements (1) are not known in advance to the learner, and (2) is the set of all diagonal PTM elements.

\subsection{Algorithms and upper bounds for estimating arbitrary PTM elements}

Algorithms~\ref{algo:pauli_shad_2n}, \ref{algo:pauli_shad_n}, \ref{algo:diag_pauli_shad_2n} and \ref{algo:diag_pauli_shad_n} yield the classical shadow of the Choi state $\rho_{{\mathcal{O}_{\mathcal{E}, O}}}$ of quantum channels of the form ${\mathcal{O}_{\mathcal{E}, O}} = \mathcal{E}^\dagger \circ \tilde{O} \circ \mathcal{E}$, which subsequently enables the efficient estimation of OTOCs. Because OTOCs can be viewed as PTM elements of the channel ${\mathcal{O}_{\mathcal{E}, O}}$, as (c.f. also Eq.~(\ref{eq:otoc_to_distinguish})):
\begin{equation}
    \frac{1}{2^n} \tr(\mathcal{E}(P_i) O \mathcal{E}(P_j) O) = \frac{1}{2^n} \tr(P_i {\mathcal{O}_{\mathcal{E}, O}}(P_j)),
\end{equation}
the above algorithms can be regarded as special instances of the more general problem of learning the classical shadow of the Choi state of general quantum channels $\mathcal{N}$, which enables the efficient estimation of PTM elements $\tr(P_i \mathcal{N}(P_j))/2^n$. We therefore immediately obtain the following algorithms/upper bounds for estimating the PTM elements of general quantum channels $\mathcal{N}$, by simply replacing ${\mathcal{O}_{\mathcal{E}, O}}$ with $\mathcal{N}$ in Algorithms~\ref{algo:pauli_shad_2n}, \ref{algo:pauli_shad_n}, \ref{algo:diag_pauli_shad_2n} and \ref{algo:diag_pauli_shad_n}.

\begin{theorem}[General PTM elements] \label{theorem:general_ptm_shad}
Consider the task of simultaneously estimating $M$ arbitrary PTM elements $\{\tr(Q_i \mathcal{N}(Q_j))/2^n\}_{i,j}$, where $Q_i, Q_j$ are weight-$w$ tensor product operators, to additive error $\epsilon$ and success probability at least $1-\delta$. Algorithms~\ref{algo:pauli_shad_2n} and \ref{algo:pauli_shad_n} achieves this using $\bigo{9^w \log(M/\delta)/\epsilon^2}$ queries of the channel $\mathcal{N}$.
\end{theorem}

\begin{theorem}[Diagonal PTM elements] \label{theorem:general_ptm_shad_diag}
Consider the task of simultaneously estimating $M$ arbitrary diagonal PTM elements $\{\tr(P \mathcal{N}(P))/2^n\}_{i,j}$, where $P$ are weight-$w$ tensor product operators, to additive error $\epsilon$ and success probability at least $1-\delta$. Algorithms~\ref{algo:diag_pauli_shad_2n} and \ref{algo:diag_pauli_shad_n} achieves this using $\bigo{3^w \log(M/\delta)/\epsilon^2}$ queries of the channel $\mathcal{N}$.
\end{theorem}

The algorithm of Theorem~\ref{theorem:general_ptm_shad} is essentially identical to existing proposals such as \cite{mcginley2022quantifying,kunjummen2023shadow,levy2024classical} to obtain the classical shadow of the Choi state $\rho_{\mathcal{N}}$. These references study its usage for tasks such as the computation of operator space entanglement, the prediction of expectation values, and Hamiltonian learning, whereas our results concern the prediction of PTM elements. We further provide matching lower bounds in the following section.

\subsection{Matching lower bounds for estimating arbitrary PTM elements}

Theorems~\ref{theorem:shadow_diag_lowerbound_formal} and \ref{theorem:shadow_gen_lowerbound_formal} apply to protocols (satisfying the stated properties) that query quantum channels of the form ${\mathcal{O}_{\mathcal{E}, O}} = \mathcal{E}^\dagger \circ \tilde{O} \circ \mathcal{E}$ to output estimates of OTOCs, which can be viewed as PTM elements of the channel ${\mathcal{O}_{\mathcal{E}, O}}$. Any protocol (satisfying the same properties) that query a general quantum channel $\mathcal{N}$ to output estimates of PTM elements $\tr(P_i \mathcal{N}(P_j))/2^n$ must therefore also obey the same lower bounds; otherwise they can be used to estimate OTOCs more efficiently, violating Theorems~\ref{theorem:shadow_diag_lowerbound_formal} and \ref{theorem:shadow_gen_lowerbound_formal}. We therefore immediately obtain the following analogues of Theorems~\ref{theorem:shadow_diag_lowerbound_formal} and \ref{theorem:shadow_gen_lowerbound_formal} for the problem of estimating weight-$w$ PTM elements.

\begin{theorem}[Sample complexity lower bound on predicting unknown diagonal PTM elements; arbitrary initialization] \label{theorem:ptm_diag_lowerbound_formal}
Consider any protocol that is able to predict any set of $M$ weight-$w$ diagonal PTM elements $\{\tr(P_i \mathcal{N}(P_i))/2^n\}_{i=1}^M$, i.e. $\forall P_i \in \mathcal{P}_n, ~w(P_i)=w$, based on initializing as arbitrary states of $n$ or more qubits, querying the (arbitrary) channel $\mathcal{N}$, and performing local measurements in each of the $N$ runs. Any such protocol requires at least $N \geq \Omega\left(3^w \log(M)/\epsilon^2\right)$ runs to solve this problem with constant success probability, $\forall w$ and $M \leq 3^w {n \choose w}$.
\end{theorem}

\begin{theorem}[Sample complexity lower bound on predicting unknown PTM elements] \label{theorem:ptm_gen_lowerbound_formal}
Consider any protocol that is able to predict any set of $M$ weight-$w$ PTM elements $\{\tr(Q_i \mathcal{N}(Q_j))/2^n\}_{i,j}$ where $Q_i, Q_j$ are weight-$w$ tensor product operators, i.e. $\forall ~ (i,j), ~w(Q_i)=w(Q_j)=w$, that uses $n$ or more qubits, based on querying the (arbitrary) channel $\mathcal{N}$, and performing local measurements in each of the $N$ runs.
\begin{enumerate}
    \item Any such protocol that can initialize as arbitrary product states requires at least $N \geq \Omega\left(9^w \log(M)/\epsilon^2\right)$ runs to solve this problem, if $w \leq \lfloor n/2 \rfloor$ and $M \leq \frac{3^{2w}}{2} {n \choose w}{n-w \choose w}$.

    \item Any such protocol that can initialize as arbitrary states requires at least $N \geq \Omega\left((9/2)^w \log(M)/\epsilon^2\right)$ runs to solve this problem, if $w \leq \lfloor n/2 \rfloor$ and $M \leq \frac{3^{2w}}{2} {n \choose w}{n-w \choose w}$.

    \item If $w > \lfloor n/2 \rfloor$ and $M \leq \frac{3^{2w}}{2} {2n \choose 2w}$, we have the weaker $N \geq \Omega\left(3^w \log(M)/\epsilon^2\right)$ bound instead, for any such protocol that can initialize as arbitrary states.
\end{enumerate}
\end{theorem}

Theorems~\ref{theorem:ptm_diag_lowerbound_formal} and \ref{theorem:ptm_gen_lowerbound_formal} complement existing lower bounds on the problem of estimating (all) PTM elements of general channels (\cite{caro2024learning}, Theorem 5.6 (for adaptive protocols)), generalizing it to the case of estimating weight-$w$, a priori unknown, PTM elements.

Furthermore, noting that the hard instances used in the proofs of Theorems~\ref{theorem:shadow_diag_lowerbound_formal} and Theorems~\ref{theorem:ptm_diag_lowerbound_formal} are Pauli channels (c.f. Eq.~(\ref{eq:hard_instances_diag})), we also immediately have the following stronger result for estimating the PTM elements of Pauli channels:
\begin{corollary}[Sample complexity lower bound on predicting local diagonal PTM elements of Pauli channels; arbitrary initialization] \label{theorem:ptm_pauli_diag_lowerbound_formal}
Consider any protocol that is able to predict any set of $M$ weight-$w$ diagonal PTM elements $\{\tr(P_i \mathcal{N}(P_i))/2^n\}_{i=1}^M$, i.e. $\forall P_i \in \mathcal{P}_n, ~w(P_i)=w$, based on initializing as arbitrary states of $n$ or more qubits, querying the (arbitrary) Pauli channel $\mathcal{N}$, and performing local measurements in each of the $N$ runs. Any such protocol requires at least $N \geq \Omega\left(3^w \log(M)/\epsilon^2\right)$ runs to solve this problem with constant success probability, $\forall w$ and $M \leq 3^w {n \choose w}$.
\end{corollary}

Corollary~\ref{theorem:ptm_pauli_diag_lowerbound_formal} complements existing results on estimating all PTM elements of Pauli channels \cite{chen2022quantum,chen2024tight,chen2025efficient,kim2026fundamental}, extending them to the case of arbitrary $M$, and when the targeted indices $\{P_i\}_{i=1}^M$ are not known in advance to the learner.

\subsection{Family of tight bounds for estimating all diagonal PTM elements using ancillas} \label{app:family_diag_ptm}

Consider the problem of estimating all $4^n-1$ diagonal PTM elements $\{\tr(P_i \mathcal{N}(P_i))/2^n\}_{P_i \in \mathcal{P}_n}$ of a general channel $\mathcal{N}$. Since this problem is harder than both the Pauli channel eigenvalue problem and the All Diagonal OTOC Estimation Problem, any lower bound on these two problems immediately imply lower bounds on the diagonal PTM elements problem. Such bounds are given by Theorem 3.A of \cite{chen2022quantum} (for the Pauli channel eigenvalue problem) and Theorem~\ref{theorem:expsep_diag_nonad} (for the All Diagonal OTOC Estimation Problem), for the class of protocols that can query general channels $\mathcal{N}$ on arbitrary $(n+n_{\text{anc}})$-qubit density matrices, with $0 \leq n_{\text{anc}} \leq n$, but without adaptivity and quantum memory. We therefore immediately have:
\begin{lemma}[Lower bounds for estimating all diagonal PTM elements of general channels]
Consider protocols for learning all $4^n-1$ diagonal PTM elements of general channels $\mathcal{N}$, based on querying $\mathcal{N}$ on arbitrary $(n+n_{\text{anc}})$-qubit density matrices, with $0 \leq n_{\text{anc}} \leq n$, but without adaptivity and quantum memory. The query complexity lower bound $\Omega(n 2^{n-n_{\text{anc}}}/\epsilon^2)$ holds.
\end{lemma}

Together with the matching upper bound of Theorem~\ref{theorem:ptm_diagonal_algo_nk} on the same problem (for constant success probability $1-\delta$), we have the following result:
\begin{theorem}[Tight bounds for estimating all diagonal PTM elements of general channels] \label{theorem:diag_ptm_optimal}
    Consider the problem of estimating all $4^n-1$ diagonal PTM elements of a general quantum channel $\mathcal{N}$. The query complexity of this problem, for algorithms based on querying $\mathcal{N}$ on arbitrary $(n + n_{\text{anc}})$-qubit initial density matrices, with $0 \leq n_{\text{anc}} \leq n$, and measuring in arbitrary POVMs, without adaptivity or quantum memory, is $\Theta(n 2^{n-n_{\text{anc}}}/\epsilon^2)$.
\end{theorem}

The above result complements the tight bounds of \cite{chen2022quantum} for the Pauli channel eigenvalue estimation problem. It shows that both problems attain the same query complexity of $\Theta(n 2^{n-n_{\text{anc}}}/\epsilon^2)$, and that a simple adaptation of the optimal $(n + n_{\text{anc}})$-qubit algorithm of \cite{chen2022quantum} (by twirling) suffices to also yield an optimal algorithm for the problem of learning all diagonal PTM elements of general channels.

\end{document}